\documentclass{article}
\usepackage{ijcai20}

\usepackage{times}
\usepackage{soul}
\usepackage{url}
\usepackage[hidelinks]{hyperref}
\usepackage[utf8]{inputenc}
\usepackage[small]{caption}
\usepackage{graphicx}
\usepackage{amsmath}
\usepackage{amsthm}
\usepackage{booktabs}
\usepackage{algorithm}
\usepackage{algorithmic}
\usepackage{enumitem}

\newtheorem{example}{Example}

\title{Answering Counting Queries over DL-Lite Ontologies}

\author{
Meghyn Bienvenu$^1$
\and
Quentin Manière$^1$
\And
Michaël Thomazo$^2$
\affiliations
$^1$University of Bordeaux, CNRS,  Bordeaux INP, LaBRI, Talence, France \\
$^2$Inria, DI ENS, ENS, CNRS, University PSL, Paris, France
\emails
\{meghyn.bienvenu, quentin.maniere\}@u-bordeaux.fr,
michael.thomazo@inria.fr
}

\usepackage{macros}
\usepackage{rotating}

\begin{document}

\maketitle

\begin{abstract}
Ontology-mediated query answering (OMQA) is a promising approach to data access and integration that has been actively studied in the knowledge representation and database communities for more than a decade. The vast majority of work on OMQA focuses on conjunctive queries, whereas more expressive queries that feature counting or other forms of aggregation remain largely unexplored. In this paper, we introduce a general form of counting query, relate it to previous proposals, and study the complexity of answering such queries in the presence of DL-Lite ontologies. As it follows from existing work that query answering is intractable and often of high complexity, we consider some practically relevant restrictions, for which we establish improved complexity bounds.
\end{abstract}

\section{Introduction}

Ontology-mediated query answering (OMQA) utilizes ontologies to provide a convenient vocabulary for query formulation and to capture domain knowledge that is exploited during the querying process to obtain more complete sets of answers \cite{DBLP:journals/jods/PoggiLCGLR08,DBLP:conf/rweb/BienvenuO15,DBLP:conf/ijcai/XiaoCKLPRZ18}. 
Much of the work on OMQA considers ontologies formulated using description logics (DLs), a family of knowledge representation languages that provide the logical foundations of the OWL web ontology language. 
Particular attention has been to the DL-Lite family of DLs \cite{calvaneseetal:dllite}, 
which were 
developed with OMQA in mind and enjoy favorable computational properties. 

The vast majority of work on OMQA supposes that user queries are given as conjunctive queries (CQs). 
However, there are many other kinds of database queries, beyond plain CQs, that are relevant in practice. This motivates research into the feasibility of adopting other database query languages for OMQA. While enriching CQs with either negated atoms or inequalities has been shown to lead to undecidability even in very restricted settings \cite{DBLP:journals/ws/Gutierrez-Basulto15}, the situation is more positive for navigational queries (like regular path queries), which can be adopted without losing decidability, sometimes even retaining tractable data complexity \cite{DBLP:journals/jair/BienvenuOS15}. 

%
%
%
%
%
%
%
%
%
%
%
%
%

Aggregate queries, which use numeric operators (e.g.\ count, sum, max) to summarize selected parts of a dataset, constitute another prominent class of database queries. 
Although such queries are widely used for data analysis, they have been little explored in context of OMQA. 
This may be partly due to the fact that it is not at all obvious how to
define the semantics of such queries in the OMQA setting. A first exploration of aggregate queries in OMQA was conducted by \citeauthor{DBLP:conf/cikm/CalvaneseKNT08} (\citeyear{DBLP:conf/cikm/CalvaneseKNT08}). 
They argued that the most straightforward adaptation of classical certain answer semantics to aggregate queries was unsatisfactory, as often values would differ from model to model, leading to no certain answers. For this reason, an epistemic semantics was proposed, in which variables involved in the aggregation are required to match to data constants. However, as discussed in 
\cite{kostylevreutter:count}, this semantics can also give unintuitive results by ignoring ways of mapping aggregate variables to anonymous elements inferred due the ontology axioms. For instance, if no children of $\ind{alex}$ are listed in the data, then a query that asks to return the number of children will yield 0 under epistemic semantics, even if it can be inferred (e.g.\ due to a family tax benefit) that there must be at least 3 children. This led \citeauthor{kostylevreutter:count} to define an alternative semantics for two kinds of counting queries (inspired by the COUNT and COUNT DISTINCT in SQL) which adopts a form of certain answer semantics but considers lower and upper bounds on the count value across different models. For the two considered logics (\dlc\ and \dlh), only the lower bounds on the count value are non-trivial, and a complexity analysis shows that they are challenging to identify: $\coNP$-data complexity for both logics, and $\PIP$-hard (resp.\ $\coNEXP$-hard) in combined complexity for \dlc\ (resp.\ \dlh). Several question were left unanswered by their work, including the exact combined complexity, 
the difficulty of recognizing the optimal lower bound, and the impact of allowing multiple aggregation variables. 


This paper returns to the issue of  handling counting queries in OMQA
and makes several important contributions:
\begin{enumerate}[leftmargin=.42cm,itemsep=-.05cm]
\item We propose a new notion of counting CQ that 
generalizes the two forms of queries from 
 \cite{kostylevreutter:count} and allows arbitrarily many counting variables. 
\item We show that existing complexity results for \dlc\ and \dlh\ KBs continue to hold for our more general notion of counting CQ, and provide  
an improved $\coNEXP$ upper bound 
for the relevant case of finite-depth TBoxes. 
\item We consider the impact of restricting the query structure, focusing on the class of rooted queries, in which every query variable must be connected to an answer variable or individual in the query graph. 
A recent result, obtained as part of a study of bag semantics
for OMQA, identified a case in which rootedness leads to tractable data complexity for counting queries \cite{nikolauetal:bag}. This motivates us to perform a more thorough investigation of rooted counting queries, 
which yields several improvements upon existing complexity bounds. 
\item We prove that the problem of identifying the best certain interval is $\DP$-complete in data complexity. 
\end{enumerate}
Our results close some questions that were 
left open by the work of \citeauthor{kostylevreutter:count} 
and pave the way for further study of counting and aggregate queries in the OMQA setting. 

An appendix with full proofs can be found in 
the long version of this paper, available on arXiv. 

%
%
%
%
%
%
%


\section{Preliminaries}


\allowdisplaybreaks

We recall the basics of description logics (DLs), focusing on DL-Lite, see \cite{DBLP:books/daglib/0041477} for more details. 

\paragraph*{Syntax and Semantics.}
A description logic vocabulary consists of a set $\cnames$ of 
\emph{atomic concepts} (unary predicates), a set $\rnames$ of
\emph{atomic roles} (binary predicates), and a set $\inames$ of \emph{individual names} (constants). 
By \emph{role}, we mean either an atomic role $P \in \rnames$ or an \emph{inverse role} $P^-$ (where $P \in \rnames$).
We let $\rni$ denote the set $\rn \cup \{P^{-} \mid P \in \rn\}$ of roles and use the notation $R^{-}$ to mean $P^{-}$ if $R=P \in \rnames$
and $P$ if $R=P^{-}$. 

A DL \emph{knowledge base (KB)} is a pair $\kb = (\tbox, \abox)$, consisting of an ABox $\Amc$ that contains facts about particular individuals and a TBox $\Tmc$ that expresses general knowledge about the domain. 
Formally, an \emph{ABox} is a finite set of \emph{concept assertions} $A(b)$, with $A \in \cnames$ and $b \in \inames$, 
and \emph{role assertions} $P(a,b)$, with $P \in \rnames$ and $a,b \in \inames$.  We use $\ainds(\Amc)$ to denote the set of individuals in $\Amc$. 
A \emph{TBox} is a finite set of axioms, whose syntax depends on the particular DL. 
In \dlc, axioms take the form of \emph{concept inclusions} $B_{1} \sqsubseteq (\neg) B_{2}$,
where each $B_{i}$ is either $A$ (for $A \in \cn$) or $\exists R$ (with $R \in \rni$).
\dlh\ TBoxes additionally allow \emph{role inclusions} 
$R_{1} \sqsubseteq (\neg) R_{2}$,
where $R_{1}, R_{2} \in \rni$.

\begin{example}
\label{exkb} Our example KB talks 
about leading ($\leads$) and supporting actors ($\supports$) in movies:
\begin{align*}
\Amc_{\mathsf{act}} = \{&\plays (\willis, \monkeys), \supports(\stowe, \monkeys), \\
&\supports(\pitt, \monkeys),  \supports(\pitt, \babel)\}\\
\Tmc_{\mathsf{act}} = \{&\leads \incl \plays, 
 \supports \incl \plays, \\ 
& \exists \supports^- \incl \exists \leads^- \}
\end{align*}
%
\end{example}

An interpretation takes the form $\Imc=(\Delta^\Imc, \cdot^\Imc)$, where 
$\Delta^\I$ is a non-empty set (the \emph{domain} of $\I$), and $.^\I$ is a function that maps each $A \in \cnames$ to a subset $A^\Imc \subseteq \Delta^\I$, each $P \in \rnames$ to a binary relation $P^\Imc \subseteq \Delta^\I \times \Delta^\I$, and each $a \in\inames$  to an element $a^\Imc$ of $\Delta^\I$. We make the \emph{unique names assumption} (UNA) by requiring that $a^\Imc \neq b^\Imc$ for every $a,b \in \inames$ with $a\neq b$. 
The function $\cdot^\Imc$ naturally extends to complex concepts and roles:
$(\exists \posrole)^\I = \lbrace d ~|~ \exists d' : (d, d') \in \posrole^\I \rbrace$, 
$(\atomrole^-)^\I = \lbrace (d_1, d_2) ~|~ (d_2, d_1) \in \atomrole^\I \rbrace $,
$(\neg \posconcept)^\I = \Delta^\I \setminus \posconcept^\I$,
$(\neg \posrole)^\I = \left( \Delta^\I \times \Delta^\I \right) \setminus \posrole^\I$. 
A (concept or role) inclusion $F \sqsubseteq G$ is satisfied in $\Imc$ if $F^\Imc \subseteq G^\Imc$;
assertion 
$A(b)$  is satisfied in $\Imc$ if $b^\Imc \in A^\Imc$;
 $P(a,b)$ is satisfied in $\Imc$ if $(a^\Imc,b^\Imc) \in P^\Imc$. 
We call $\Imc$ a \emph{model} of $\Kmc$, written  $\Imc \models \Kmc$, if it satisfies all inclusions and assertions in $\Kmc$. A KB is \emph{satisfiable} if has at least one model.


\paragraph*{Queries.} We recall that a \emph{conjunctive query} (CQ) takes the form $\exists \ty\, \psi(\tx,\ty)$, 
where $\tx$ and $\ty$ are tuples of variables drawn from an infinite set of variables $\varSet$, and 
$\psi$ is a conjunction of \emph{atoms}, which 
can be either \emph{concept atoms} $\atomconcept(t_1)$ or \emph{role atoms} $\atomrole(t_1, t_2)$, where $\atomconcept \in \cn$, $\atomrole \in \rn$, and \emph{terms} $t_i$ are drawn from $\inames \cup \tx \cup \ty$. 
Consider an interpretation $\Imc$ and CQ $q=\exists \ty\, \psi(\tx,\ty)$ with $|\tx|=n$. 
A tuple $\talp \in \left( \domain{\I} \right)^n$ 
 is an \emph{answer to $q$ in $\Imc$}, written $\I \models q(\talp)$, 
if there is a \emph{homomorphism of $q$ into $\Imc$}, i.e., a function $\sigma$ that maps  
the terms of $q$ to elements of $\Delta^\Imc$ such that (i) $\sigma(a)=a^\Imc$ for $a \in \inames$,
(ii) $\sigma(t) \in A^\Imc$ for every atom $A(t)$ of $q$, and (iii) $(\sigma(t_1), \sigma(t_2))\in P^\Imc$ for every atom $P(t_1,t_2)$ of $q$. 
A tuple $\ta \in \ainds(\Amc)^n$ 
 is a \emph{certain answer to $q$ w.r.t.\ the KB $\Kmc$} iff 
$\I \models q(\ta^\Imc)$ for every model $\Imc$ of $\Kmc$.

\paragraph*{Canonical Model.} 
We recall the definition of the canonical model $\canmod$ of 
a 
\dlh\ KB $\Kmc=(\Tmc,\Amc)$. The domain of $\canmod$ consists of $\ainds(\Amc)$ and 
all words of the form $a R_1\dots R_n$, with $a \in \ainds(\Amc)$, $R_i \in \rni$, and $n\geqslant 1$, such that:
\begin{itemize}[leftmargin=.35cm,itemsep=-.05cm]
\item  
$\Kmc \models \exists R_1 (a)$ and there is no $R_1(a,b) \in \Amc$;
\item for $1 \leq i < n$, $\mathcal{T} \models \exists R_i^-
  \sqsubseteq \exists R_{i+1}$ and \mbox{$R_i^- \ne R_{i+1}$}.   
\end{itemize}
We interpret individuals as themselves ($ a^{\canmod} =a$) and atomic concepts and roles as follows:
\begin{align*}
 A^{\canmod} =   &   \,  \{ a \in \ainds(\Amc) \mid \Kmc \models A(a) \} \\
&  \,\,  \cup \ 
 \{ a R_1\dots R_n \in 
\Delta^{\canmod}\setminus \ainds(\Amc)  \mid 
 \Tmc \models \exists R_n^- \sqsubseteq A  \}   \\
 P^{\canmod} =   &  \,  \{ (a,b) \mid P(a,b) \in \Amc \} \, \cup   \\
&  \,   \{ (e_{1},e_{2})
\mid  
e_{2} = e_{1} R   \text{ and } \Tmc \models R \sqsubseteq P \} \,\cup   \\
&   \,  \{ (e_{2},e_{1}) 
\mid \
e_{2} = e_{1} R \text{ and } \Tmc \models R \sqsubseteq P^{-} \} 
\end{align*}
The term `canonical model' is motivated by 
the following well-known property of $\canmod$ (see e.g.\ \cite{calvaneseetal:dllite}):
\begin{lemma}
\label{cm-lemma}
Let $\kb$ be a satisfiable \dlh \ KB. Then $\canmod \models \kb$ and if 
$\Imc \models \Kmc$, 
there is a homomorphism of $\canmod$ into $\Imc$.
\end{lemma}



\noindent A useful corollary 
is that the \emph{certain answers to a CQ $q$ w.r.t.\ $\Kmc$} are 
the  tuples from $\individuals(\abox)$
 that are 
\emph{answers to $q$ in~$\canmod$}. 

Note that $\canmod$ may be infinite. 
The \emph{depth} of a TBox $\Tmc$ is defined as the maximal length of any word
that appears in the domain of $\canmod$ for any KB $\Kmc$ whose TBox is $\Tmc$. 
If this number is finite, we say that $\Tmc$ is a \emph{finite-depth TBox}; such TBoxes can be
identified in polynomial time  \cite{DBLP:conf/lics/BienvenuKP15}.

\section{Counting Queries}
\label{sectionsemantics}


We now introduce our formalization of counting queries. In addition to the set $\varSet$ of (classical) variables, 
we assume a second infinite set of counting variables $\countVarSet$, disjoint from $\varSet$. 
\begin{definition}
\label{defaggregate}
A counting conjunctive query (CCQ) $q$ takes the form 
$q(\tx)= \exists \ty\exists \tz~\psi(\tx, \ty, \tz)$, where $\tx \cup \ty \subseteq \varSet$, $\tz \subseteq \countVarSet$, and 
$\psi$ is a conjunction of concept and role atoms whose terms are drawn from $\inames \cup \tx \cup \ty \cup \tz$. 
 We call $\tx$ (resp. $\ty$, resp. $\tz$) the \emph{answer} (resp.\ \emph{existential}, resp.\ \emph{counting}) variables of $q$.
\end{definition}


%
%
%


We first define the semantics of counting queries on a single interpretation $\Imc$, 
by considering those pairs $(\ta, n)$ such that $n$ is the number of possible ways to map $\tz$ into $\Imc$ when $\tx$ is mapped to $\ta$. Such pairs are called 
 the \emph{answers} to $q$ in $\I$.

\begin{definition}
A \emph{match} of a CCQ $q(\tx)= \exists \ty\exists \tz~\psi(\tx, \ty, \tz)$ in $\I$ is a homomorphism\footnote{The notion of homomorphism of a CCQ is defined in the same way as for CQs, simply treating variables from $\countVarSet$ like those in $\varSet$.} 
from $q$ into $\Imc$. 
If a match $\sigma$ maps $\tx$ to $\ta$, then the restriction of $\sigma$ to $\tz$ is called a \emph{counting match (c-match)} of $q(\ta)$ in $\Imc$. 
The set of \emph{answers to $q$ in $\I$}, denoted $\qAnswers{q}{\I}$, contains all pairs $(\ta,\alphaMatch{q}{\I}{\ta})$, where $\alphaMatch{q}{\I}{\ta}$ is the number of distinct c-matches of $q(\ta)$ in $\I$. 

\end{definition}

As has been previously noted  (see e.g.\ \cite{kostylevreutter:count}), 
the exact count values of the answers in $\qAnswers{q}{\I}$ are usually too specific to hold across models. 
Considering \emph{bounds} on the exact value provides more insight, while still allowing unnamed elements to be counted. This motivates the following notion of answer interval. 


\begin{definition}
The set $\qIntervals{q}{\I}$ of \emph{answer intervals for a CCQ $q$ in $\I$} contains all pairs
$\left( \ta, [m, M] \right)$ with $\ta \in \individuals^{|\tx|}$ and $m, M$ integers such that $m \leq \alphaMatch{q}{\I}{\ta}\leq M$. 
The set $\qIntervals{q}{\kb}$ of \emph{certain (counting) answers to $q$ w.r.t.\ KB $\Kmc$}  is obtained by considering 
those answer intervals that hold in all models of $\Kmc$:
%
%
$\qIntervals{q}{\kb} = \bigcap_{\I \models \kb} \, \qIntervals{q}{\I}.$
\end{definition}
%
%
%
\noindent 
Note that $(\ta, [m, M]) \in \qIntervals{q}{\kb}$ does not imply that for any $n \in [m,M]$ there exists a model $\Imc$ 
in which $(\ta,n) \in \intervals{q}{\I}$. 

%

Definition \ref{defaggregate} 
is a proper generalization of the two forms of counting query considered by 
 \citeauthor{kostylevreutter:count}. Reusing their notations, a  $Cntd()$-query $q(\tx, Cntd(z)) = \exists \ty\ \psi(\tx, \ty, z)$ corresponds to the CCQ  $q(\tx) = \psi(\tx, \ty, z)$, while a $Count()$-query $q(\tx, Count()) = \exists \ty\ \psi(\tx, \ty)$ corresponds to the CCQ $q(\tx) = \psi(\tx, \emptyset, \hat{\ty})$ (with $\hat{\ty}$ a tuple of variables from $\countVarSet$ in bijection with $\ty$). We will use the term \emph{exhaustive} to refer to the latter CCQs, i.e.\ those in which every non-answer variable is a counting variable. 

\begin{example}
\label{exqueries} Reconsider the 
KB $\Kmc_{\mathsf{act}}=(\Tmc_{\mathsf{act}}, \Amc_{\mathsf{act}})$. 
We can use CCQs to count the pairs of actors (leading role, supporting role) having acted together ($q_1$), 
return movies together with a count of their supporting actors ($q_2$), and 
count the number of actors having acted with Tom Hanks ($q_3$):
\begin{align*}
q_1 &= \exists y \exists z_1 \exists z_2 \ \leads(z_1, y) \wedge \supports(z_2, y)\\
q_2(x) &= \exists z \ \supports(z, x)\\
q_3 &= \exists y \exists z \ \plays(\pitt, y) \wedge \plays(z, y)
\end{align*}
According to our semantics, we have:
\begin{itemize}[leftmargin=.35cm,itemsep=-.05cm]
\item $(\emptyset, [2, +\infty]) \in \qIntervals{q_1}{\Kmc_{\mathsf{act}}}$, since $z_2$ can be mapped to either $\stowe$ or $\pitt$, and $z_1$ mapped to the lead actor (which must exist due to $\Tmc_{\mathsf{act}}$). As the lead actors of the two films could be the same, $(\emptyset, [3, +\infty]) \not \in \qIntervals{q_1}{\Kmc_{\mathsf{act}}}$. 
\item $(\monkeys, [2, +\infty]) \in \qIntervals{q_2}{\Kmc_{\mathsf{act}}}$, mapping $z$  to $\stowe$ and $\pitt$.
\item $(\emptyset, 5) \in q_3^{\mathcal{C}_{\Kmc_{\mathsf{act}}}}$, since in $\mathcal{C}_{\Kmc_{\mathsf{act}}}$,  we can map $z$ to a named actor or the two elements standing in for the lead actors.
\item $(\emptyset, [5, +\infty]) \notin \qIntervals{q_3}{\Kmc_{\mathsf{act}}}$, since the lead actors could possibly be the same or one of the named actors. 
\end{itemize}
The latter two points show that the canonical model does not yield the minimal number of matches. 
\end{example}




%
%

\section{General Counting CQs}
\label{sectiongeneral}

\newcommand{\proj}{\pi}
We shall consider the following CCQ answering decision problem: given a KB $\Kmc = (\Tmc, \Amc)$, CCQ $q$, and candidate answer $(\ta, [m, M])$, decide whether $(\ta, [m, M]) \in \qIntervals{q}{\kb}$. 

As ontology languages, we will consider \dlh\ (which underlies OWL 2 QL) and its sublogic \dlc. We know from \cite{kostylevreutter:count} that in these DLs, the least upper bound $M$ can take one of three values (0, 1, or +$\infty$) and is easily computed. 
The argument\footnote{Briefly, the upper bound is 0 if the tuple is not a certain answer; otherwise, it is either 1 if 
$ \tz = \emptyset$, else $+\infty$.}
transfers to our more general notion of CCQ.  We can therefore \emph{restrict our attention to identifying certain answers of the form $(\ta, [m, +\infty])$.
} 


We will consider the two usual complexity measures: \emph{combined complexity} which is in terms of the size of the whole input ($\tbox$, $\abox$, $q$, $\ta$, $m$), and \emph{data complexity} which is only in terms of the size of $\abox$ and $m$ ($\Tmc$ and $q$ are treated as fixed). We will assume that $m$ is given in binary. 



\begin{table}[!t]
\centering
\begin{tabular}{lcc}
 & \textbf{Data}      & \textbf{Combined}		\\
\midrule
\dlc& \bf $\coNP$-c & $\PIP$\textbf{-h}, $\PP$\textbf{-h} \& in $\coNEXP$ \\
\midrule
\dlh & \bf $\coNP$-c  & $\coNEXP$\textbf{-h} \& in $\coNDEXP$\\
& & $\coNEXP$\textbf{-c} \, ($\Tmc$ of finite depth)\\ 
\bottomrule
\end{tabular}
\caption{Data and combined complexity of CCQ answering}
\label{resultsgeneral}
\end{table}

\subsection{General Case}\label{gencase}
Table~\ref{resultsgeneral} displays complexity results for answering general CCQs over \dlc\ and \dlh\ TBoxes
(we use `\textbf{-c}' and `\textbf{-h}' as abbreviations for `-complete' and `-hard'). 


With the exception of the $\PP$-hardness result (discussed in Section~\ref{srooted-core}), the lower bounds are inherited 
from \cite{kostylevreutter:count}. We will thus concentrate on the upper bounds from Table~\ref{resultsgeneral}, which are obtained by generalizing and clarifying the 
constructions of \citeauthor{kostylevreutter:count}. 
We give an overview of the proof 
both to give the flavor of the techniques involved and to enable us to discuss the necessary adaptations used to prove later results. 


The proof constructs a decision procedure for the complementary problem of deciding whether $(\ta, [m, +\infty]) \not \in \qIntervals{q}{\kb}$. The latter holds iff there exists a \emph{countermodel}, i.e., a model of $\Kmc$ with fewer than $m$ c-matches of $q(\ta)$.
The main ingredient of the proof is the following theorem, which shows that it is sufficient to consider countermodels of bounded size. 

\begin{theoremrep}
\label{upperboundgeneral}
For every \dlh\ (resp.\ \dlc) KB $\Kmc = (\Tmc, \Amc)$ and CCQ $q$, if there is a model of $\Kmc$ with fewer than $ m$ c-matches of $q(\ta)$, then
there exists one of size\footnote{As usual, $|\Tmc|$ (resp.\ $|\Amc|$, $|q|$) denotes the size of $\Tmc$ (resp.\ $\Amc, q$).} 
$O(|\abox|^{|\tbox|^{|q|+1}})$ (resp.\ $O(|\abox|^{|q|})$).
\end{theoremrep}

With Theorem~\ref{upperboundgeneral} in hand, we can easily define non-deterministic procedures that witness the complexity upper bounds from Table~\ref{resultsgeneral}: simply guess an interpretation of polynomial / exponential / double-exponential size (depending on the case) and verify whether it is a countermodel. 

The proof of Theorem~\ref{upperboundgeneral} starts with an arbitrary countermodel $\Imc$ and modifies it in order to make it smaller, being careful not to introduce any new c-matches of $q(\ta)$. We first identify a relevant subset $\domain{*}$  of the domain of $\Imc$, consisting of the interpretations of all individual names from $\Amc$ and the images of all 
 c-matches of $q(\ta)$. We then define a new interpretation that intuitively preserves $\domain{*}$ and replaces the rest of $\Imc$ with parts of the canonical model, to introduce a more regular structure.  
 Formally, we fix a homomorphism $f$ of $\canmod$ into $\Imc$ (see Lemma~\ref{cm-lemma}) and 
  consider the 
 following mapping 
 $f' : \Delta^{\canmod} \rightarrow  \Delta^* \cup \domain{\canmod}$
  from \cite{kostylevreutter:count}:
%
$$
\begin{array}{rcl}
        		  f'(d) &   =& \left\lbrace
     									\begin{array}{ll}
						                f(d) & \text{if } f(d) \in \domain{*} \\
						                d   & \text{otherwise}
                							\end{array} \right.
\end{array}
$$
We define the \emph{interleaving}\footnote{We have slightly modified the definition of interleaving to correct a small bug in the definition from \cite{kostylevreutter:count}.} $\Imc'$ of $\Imc$ as the image of $\canmod$ by $f'$, i.e., with domain 
$f'(\domain{\canmod})$ and interpretation function  
$ f' \circ \cdot^{\canmod}$.



It is not difficult to prove that the interleaving $\Imc'$ is a model of~$\kb$. Moreover, by exhibiting a homomorphism $\rho$ from $\I'$ to $\I$, we can translate matches of $\I'$ into matches in $\I$. As the images of c-matches of $q(\ta)$ are contained in $\domain{*}$, which is left unchanged in $\I'$, the homomorphism $\rho$ is in fact a one-to-one mapping of c-matches of $q(\ta)$ in $\I'$ to those in $\I$. This shows that $\I'$ is also a countermodel.

The interleaving $\Imc'$ may be arbitrarily large, even infinite. To reduce its size, an equivalence relation is introduced, and elements from $\Delta^{\Imc'} \setminus \Delta^* $ that belong to the same equivalence class are merged (elements from $\Delta^*$  are retained). In the case of \dlh, there can be double-exponentially many equivalence classes, as elements are grouped based upon the properties of their $|q|$-neighborhoods, while for \dlc, we can use 
a more refined relation with only exponentially many classes. This means that the resulting models are either of single- or double-exponential size w.r.t.\ combined complexity, depending on the chosen DL. 

A crucial final step 
is to show that the merging of elements does not introduce any new c-matches of $q(\ta)$, 
so the resulting model is still a countermodel. This part of the argument, only sketched in \cite{kostylevreutter:count}, requires 
a detailed and technical analysis of the construction to ensure that 
this property holds for our more general class of CCQs. We show that this is indeed the case, which answers a question left open by \citeauthor{kostylevreutter:count} about counting CQs with both existential variables and multiple counting variables.  


\begin{toappendix}

We recall the construction of the interleaving and some of its basic properties. We start with an arbitrary countermodel $\I$, and consider the subdomain $\domain{*}$ consisting of individual names from $\Amc$ and the images of all 
 c-matches of $q(\ta)$: 
$$\domain{*} =\{a^\I \mid a \in \ainds(\Amc)\} \cup \bigcup_{\match \text{\ match for\ } q(\ta) \text{\ in\ } \I} \match(\tz).$$
We fix a homomorphism $f$ of $\canmod$ into $\Imc$ (see Lemma~\ref{cm-lemma}) and consider the following mapping from \cite{kostylevreutter:count}:
$$
\begin{array}{rcl}
f' : \Delta^{\canmod} & \rightarrow & \Delta^* \cup \domain{\canmod} \\
        		  d &     \mapsto & \left\lbrace
     									\begin{array}{ll}
						                f(d) & \text{if } f(d) \in \domain{*} \\
						                d   & \text{otherwise}
                							\end{array} \right.
\end{array}
$$
The \emph{interleaving}\footnote{We have slightly modified the definition of interleaving to correct a small bug in the definition from \cite{kostylevreutter:count}.
}
$\Imc'$ of $\Imc$ is defined as the image of $\canmod$ by $f'$. More precisely, $\Imc'$  has domain 
$f'(\domain{\canmod})= \{f'(d) \mid d \in \domain{\canmod}\}$ and interpretation function $f' \circ \cdot^{\canmod}$, i.e.,  
$A^{\Imc'} = \{f'(d) \mid d \in A^{\canmod}\}$ and $R^{\Imc'}= \{(f'(d_1), f'(d_2)) \mid (d_1, d_2) \in R^{\canmod} \}$.

It will be useful to exhibit a homomorphism from the interleaving into the original model, which will embed matches.

\begin{lemma}
\label{rho}
The following mapping is a homomorphism from $\I'$ to $\I$.
$$
\begin{array}{rcl}
\rho : \domain{\I'} & \rightarrow & \domain{\I }\\
   	     f'(d)  &     \mapsto & f(d).
\end{array}
$$
\end{lemma}
\begin{proof}
We first check that the definition is well-founded, that is: $\rho(f'(d))$ does not depend on the choice of $d$. To this end, consider $d_1$, $d_2$ such that $f'(d_1) = f'(d_2)$. Since $f'$ maps to $\Delta^* \cup \domain{\canmod} $, we have two cases to examine:
\begin{itemize}

\item
if $f'(d_1)=f'(d_2) \in \Delta^*$, that means $f'(d_1) = f(d_1)$ and $f'(d_2) = f(d_2)$, thus ensuring $f(d_1) = f(d_2)$.

\item
if $f'(d_1)=f'(d_2) \in \Delta^{\canmod}$, that means $f'(d_1) = d_1$ and $f'(d_2) = d_2$, thus ensuring again $f(d_1) = f(d_2)$.
\end{itemize}
In both cases, we obtain $f(d_1) = f(d_2)$, so the function is well-founded. 

To show that $\rho$ is a homomorphism of $\I'$ into $\I$, we use the definition of $f'$ and fact that $f$ is a homomorphism of $\canmod$ to $\I$. Suppose first that $f'(d) \in A^\Imc$. Then $d \in A^{\canmod}$, and since $f$ is a homomorphism, $\rho(f'(d))=f(d) \in A^{\I}$. Suppose next that $(f'(d_1), f'(d_2)) \in R^\Imc$. Then there exist $(d_1', d_2') \in R^{\canmod}$ 
such that $f'(d_1)=f'(d_1')$ and $f'(d_2)=f'(d_2')$. Using again the fact that $f$ is a homomorphism, we obtain $(f(d_1'), f(d_2')) \in R^{\I'}$. By the preceding paragraph, this means  $(f(d_1), f(d_2)) \in R^{\I'}$. 
\end{proof}


\begin{lemma}
\label{interleavingismodel}
The interleaving $\I'$ is a model of $\kb$.
\end{lemma}

\begin{proof}
We check that all axioms and assertions from the KB are satisfied.
\begin{itemize}
\item
All ABox assertions from $\Amc$ are satisfied in $\canmod$. Since $f'=f$ when restricted to the ABox individuals, 
and $f$ is a homomorphism of $\canmod$ in $\I$, it follows that 
these ABox assertions must also be satisfied  in $\I$.
\item Since $\I'$ maps homomorphically into $\I$ (by Lemma \ref{rho}), any violation of an axiom from the TBox in $\I'$ implies a similar violation in $\I$. Since $\I$ is a model of $\Kmc$, this cannot occur. 
\qedhere
\end{itemize}
\end{proof}

\begin{theorem}
\label{interleaving}
The interleaving $\I'$ is a countermodel, and every c-match $\match$ in $\I'$ satisfies that $\match(\tz) \subseteq \domain{*}$.
\end{theorem}
%
%
\begin{proof}
Assume we have a c-match in the interleaving $\match : \tz \rightarrow \domain{\I'}$, which has an associated match $\cmatch$ for $q(\ta)$. Since $\rho$ is a homomorphism, $\rho \circ \cmatch$ is a match for $q(\ta)$ in the original model $\I$, and its restriction to $\tz$, that is, $\rho \circ \match$, is a c-match in $\I$.
Hence, it follows from the definition of $\domain{*}$ that $(\rho \circ \match)(\tz) \subseteq \domain{*}$. As $\rho^{-1}(\domain{*}) = \domain{*}$, this implies $\match(\tz) \subseteq \domain{*}$.

Moreover, since $\rho_{|\domain{*}} = id$, we in fact have $\rho \circ \match = \match$. We have thus shown that every c-match for $q(\ta)$ in $\Imc'$ is also a c-match for $q(\ta)$ in $\Imc$, which means the number of c-matches in $\Imc'$ cannot exceed the number of c-matches. As $\I$ was assumed to be a countermodel (i.e. having less than $m$ c-matches), it follows that the same holds for $\I'$. 
\end{proof}

In general, the interleaving has an unbounded size. To reduce the size, we will merge some domain elements, while paying attention not to introduce any new matches. To decide which elements can be merged, we will look at their local properties, which will be formalized using the following notions of chains and neighbourhoods (as in \cite{kostylevreutter:count}).

\newcommand{\M}{\mathcal{M}}
\newcommand{\D}{\mathcal{D}}
\begin{definition}[($k$-chains with respect to a subdomain)]
A $k$-chain in a model $\M$ with respect to a subdomain $\D \subseteq \domain{\M}$ is a sequence $(d_0, \dots d_k)$ with $d_i$ in $\domain{\M}$, such that for all $0 \leq i < k$, we have (i) $d_i \notin \D$, and (ii) 
 there exists a positive role $\rolestyle{R}_i$ such that $(d_i, d_{i+1}) \in \rolestyle{R}_i^{\M}$.
Note that the final element $d_k$ might belong to $\D$. 
\end{definition}

\begin{definition}[($n$-neighbourhood with respect to a subdomain)]
Consider a model $\M$ and an element $d \in \domain{\M}$. Its $n$-neighbourhood $\N_n(d, \M, \D)$ w.r.t. a subdomain $\D$ is the set of elements $d' \in \domain{\M}$ such that there exists a $k$-chain $(d_0, \dots d_k)$  in $\M$ with respect to $\D$ such that
$k \leq n$, $d_0 = d$, and  $d_k = d'$.
\end{definition}

Recall that the definition of $\domain{\I'}$ ensures that any $d \in \domain{\I'} \setminus \domain{*}$ is actually an element of $\domain{\canmod}$ and therefore we have $d = \ind{a}w$ for some individual name $\ind{a}$ and word $w$. The tree-shaped structure of $\domain{\canmod}$ ensures that there exists a unique prefix $r_{n, d}$ of $\ind{a}w$ such that :
\begin{itemize}
\item $f'(r_{n, d}) \in \N_n(d, \I', \domain{*})$ ;
\item for any $d' \in \N_n(d, \I', \domain{*})$, there exists a unique word $w_{n, d}^{d'}$ such that $d' = f'(r_{n, d} w_{n, d}^{d'})$.
\end{itemize}
We denote by $\Omega_{n}$ the set of words over the alphabet of role names occuring in the TBox $\tbox$ and with length less or equal to $2n$. 
%
Local properties around $d$ are then captured by the following function.
$$
\begin{array}{rcl}
\chi_{n, d} : \Omega_{n} & \rightarrow & \domain{*} \cup \lbrace \emptyset \rbrace \\
     		  w &     \mapsto & \left\lbrace
     									\begin{array}{ll}
						                f'(r_{n, d} w) & \text{if } r_{n, d} w \in \domain{\canmod} \text{ and } f'(r_{n, d} w) \in \domain{*} \\
						                \emptyset  & \text{otherwise}
                							\end{array} \right. .
\end{array}
$$
The next definition groups together elements having the same local properties. 

\begin{definition}[(Equivalent elements in the interleaving)]
The equivalence relation $\sim_n$ on $\domain{\I'}$ is defined as follows:
\begin{itemize}
\item for $d \in \domain{\I'} \setminus \domain{*}$, we have $d \sim_n e$ iff 
$w^d_{n, d} = w^e_{n, e}$,
$\chi_{n, d} = \chi_{n, e}$, and $|d| = |e| \mod 2|q|+3$,
\item for $d \in \domain{*}$, $d \sim_n e$ iff $d=e$. 
\end{itemize}
\end{definition}


\begin{remark}
\label{simsrelations}
Notice that if $d \sim_n e$, 
then $d \sim_m e$ for any $m \leq n$. This property will be used several times without mention.

\end{remark}

We can now define a smaller countermodel for our CCQ $q$ by merging elements with respect to $\sim_{|q| + 1}$. We will use $\overline{d}$ for the equivalence class of $d$ w.r.t.\ $\sim_{|q| + 1}$, and we denote by $\proj$ the canonical projection, which maps elements to their respective equivalence classes:
\begin{align*}
\proj : \domain{\I'} & \rightarrow \domain{\I'}/{ \sim_{|q|+1} } \\
 d & \mapsto \overline{d}
\end{align*}
\begin{definition}[(Reduced interleaving)]
The \emph{reduced interleaving} $\J$ is the interpretation with domain $\domain{\I'}/\sim_{|q| + 1}$ and interpretation of individual names, atomic concepts and roles given by $\cdot^{\J} := \proj \circ \cdot^{\I'}$.
\end{definition}

Once again, it follows from the definition that $\proj : \I' \rightarrow \J$ is a homomorphism and that $\J$ is a model of $\Kmc$. Since we are considering a quotient, we will not be able to build a general homomorphism from $\J$ to $\I'$ as in Lemma~\ref{rho}. However, local solutions are possible. To improve the readability of the following theorem and later material, we introduce the notation $\overline{\domain{*}}$ for the set $\{\overline{\sigma} \mid \sigma \in \Delta^*\}$. 

%

\begin{theorem}
\label{morphismneighbourhoods}
For any $d \in \domain{\I'}$, there exists a homomorphism $\rho_d : \N_{|q|}(\overline{d}, \J, \overline{\domain{*}}) \rightarrow \N_{|q|}(d, \I', \domain{*})$ satisfying that :
\begin{enumerate}
\item if $\overline{e} \in \overline{\domain{*}}$, then $\rho_d(\overline{e}) = e$ ;
\item $\rho_d^{-1}(\domain{*}) = \overline{\domain{*}}$.
\end{enumerate}
\end{theorem}

Let us first explain how this will conclude our proof, through the following consequence.


\begin{corollary}
\label{quotient}
If $\I$ is a countermodel, then its reduced interleaving $\J$ is a countermodel.
\end{corollary}

\begin{proof}
Assume we have a match $\match$ in $\J$. Consider a minimal covering of $\match(\tx \cup \ty \cup \tz)$ by $|q|$-neighbourhoods in $\J$ : that is a family of neighbourhoods $(\N_{|q|}(\overline{d_1}, \J), \dots \N_{|q|}(\overline{d_l}, \J))$, with $l$ being minimal and such that $\match(\tx \cup \ty \cup \tz) \subseteq \bigcup_{k=1}^l \N_{|q|}(\overline{d_k}, \J)$. In particular, the minimality ensures that the only possible overlapping elements between two different neighbourhoods are elements of $\overline{\domain{*}}$. Along with condition~$1$ from Theorem~\ref{morphismneighbourhoods}, this ensures that the following mapping is well defined:
$$
\begin{array}{rcl}
\match' : \tx \cup \ty \cup \tz & \rightarrow & \I' \\
v & \mapsto & \rho_{d_k} (\match(v)) \qquad \text{ if } \match(v) \in \N_{|q|}(\overline{d_k}, \J)
\end{array}
$$

Furthermore, as the mappings $\rho_{d_k}$ provided by Theorem~\ref{morphismneighbourhoods} are homomorphisms, it follows that $\match'$ is a match in $\I'$. Hence, by Theorem~\ref{interleaving}, we have $\match'(\tz) \subseteq \domain{*}$. However, condition~$2$ from Theorem~\ref{morphismneighbourhoods} ensures $(\rho_{d_k})^{-1}(\domain{*}) = \overline{\domain{*}}$, hence $\match(\tz) \subseteq \overline{\domain{*}}$. Therefore, for each $z \in \tz$, we have $\match(z) = \{ e_z \}$ with $e_z \in \domain{*}$. From condition~$1$, it follows that $\match'(z) = e_z$. In particular, the mapping $\match_{|\tz} \mapsto \match'_{|\tz}$ is injective, so there are at most as many counting matches in $\J$ than in $\I'$. Hence, if $\I'$ is a countermodel, then $\J$ also is.

Recalling Theorem~\ref{interleaving}, we obtain that if $\I$ is a countermodel, then $\J$ also is.
\end{proof}

To conclude the proof of Theorem~\ref{upperboundgeneral}, notice that an equivalence class $\overline{d}$ is fully characterized by :
\begin{itemize}
\item $|d| \mod 2 |q| + 3$, that is one equivalent class among $2 |q| + 3$ possible classes,
\item $w_{|q|+1, d}^d$, that is a word over an alphabet with at most $|\tbox|$ symbols and a length at most $|q| + 1$,
\item $\chi_{|q|+1, d}$, that is a function from words over an alphabet with at most $|\tbox|$ symbols and length at most $2(|q|+1)$, to a set with size at most $|\domain{*}| + 1$.
\end{itemize}
Therefore, the amount of possibly different equivalence classes, that is $|\domain{\J}|$, is at most:
$$(2|q| + 3) \times |\tbox|^{|q|+2} \times (|\domain{*}| + 1)^{|\tbox|^{2|q|+3}}.$$
Since $|\domain{*}| \leq |\individuals| + n_0 |q|$ (recall $n_0$ is the amount of c-matches in $\I$, and that we can assume $n_0 \leq (|\individuals| + |\tbox|)^{|q|}$), we have the claimed bounds for the size of $\J$, which proves Theorem~\ref{resultsgeneral}.




\vskip 5pt
Coming back to the proof of Theorem~\ref{morphismneighbourhoods}, we start by building the mappings $\rho_d$. To do so, we need to transform $k$-chains in the reduced interleaving into $k$-chains in the interleaving.

\newcommand{\core}[1]{{\widetilde{#1}}}
\begin{definition}(primary role, core of a $k$-chain)
Given a couple $(d_1, d_2)$ of elements in $\domain{\I'} \setminus \domain{*}$ (resp $(\overline{d_1}, \overline{d_2})$ in $\domain{\J} \setminus \overline{\domain{*}}$), if there exists a positive role $\rolestyle{R}$ connecting these two elements, then we call the \emph{primary role of the edge $(d_1, d_2)$} the role $\rolestyle{S} \rnames$ such that either $d_2 = d_1 \rolestyle{S}$ or $d_1 = d_2 \rolestyle{S^-}$ (resp $w_{|q|+1, d_2}^{d_2} = w_{|q|, d_1}^{d_1} \rolestyle{S}$ or $w_{|q|+1, d_1}^{d_1} = w_{|q|, d_2}^{d_2} \rolestyle{S^-}$). Notice that $\tbox \models \rolestyle{S} \incl \rolestyle{R}$.

The action of a role $\rolestyle{R}$ on a word $w$, denoted $\rolestyle{R} \diamond w$, is either $w'$ if $w = w' \rolestyle{R}^-$, or $w \rolestyle{R}$ otherwise.

Given a $k$-chain $C = (d_0, \dots d_k)$ in $\I'$ (resp in $\J$), we consider its \emph{core}, denoted $\core{C}$, being the $k$-sequence of role names such that $\core{C}_i$ is the primary role of $(d_{i-1}, d_i)$.

Given the core $\core{C}$ of a $k$-chain $C$, we define its action on a word $w$, denoted $\core{C} \diamond w$, by $\core{C}_k \diamond \dots \diamond \core{C}_1 \diamond w$.
\end{definition}

\begin{remark}
For any couple $(d_1, d_2)$ of elements in $\domain{\I'} \setminus \domain{*}$, the role $\rolestyle{S}$ is the primary role of the edge $(d_1, d_2)$ (in $\I'$) iff $\rolestyle{S}$ is the primary role of the edge $(\overline{d_1}, \overline{d_2})$ (in $\J$).
\end{remark}

\begin{lemma}
For any $k$-chain in $\J$ from $\overline{d}$ to $\overline{e}$, we have $w_{|q|+1 - k, e}^e = w_{|q| +1 - k, \core{C} \diamond d }^{\core{C} \diamond d}$.
\end{lemma}

\begin{proof}
We proceed by a straightforward induction on the length of $C$.
\begin{itemize}
\item If $C$ is a $0$-chain, that is $\overline{e} = \overline{d}$, then $\core{C}$ is empty and thus $\core{C} \diamond d = d$.
\item Otherwise, $C$ is a $(k+1)$-chain $(\overline{d}, \overline{d_1}, \dots \overline{d_k}, \overline{e})$, then consider the primary role $\rolestyle{S}$ of the edge $(\overline{d_k}, \overline{e})$. If $w_{|q|+1 - k, (\core{C}_1 \dots \core{C_k}) \diamond d}^{(\core{C}_1 \dots C_k) \diamond d} = w' \rolestyle{S^-}$, then we have :
\begin{align*}
w_{|q|+1 - (k+1), \core{C} \diamond d}^{\core{C} \diamond d} & = \rolestyle{S} \diamond w_{(q+1-k, (\core{C}_1, \dots \core{C}_k) \diamond d_k}^{(\core{C}_1, \dots \core{C}_k) \diamond d_k} \\
& = \rolestyle{S} \diamond w_{|q|+1 - k, d_k}^{d_k} \\
& = w_{|q| +1 - (k+1), e}^e.
\end{align*}
Otherwise, we have :
\begin{align*}
w_{|q|+1 - (k+1), \core{C} \diamond d}^{\core{C} \diamond d} & = \rolestyle{S} \diamond w_{(q+1-(k+2), (\core{C}_1, \dots \core{C}_k) \diamond d_k}^{(\core{C}_1, \dots \core{C}_k) \diamond d_k} \\
& = \rolestyle{S} \diamond w_{|q|+1 - (k+2), d_k}^{d_k} \\
& = w_{|q| +1 - (k+1), e}^e.
\end{align*}
\end{itemize}
\end{proof}

\begin{lemma}
\label{action}
For any $k \leq |q|$, any $d, e \in \domain{\I'} \setminus \domain{*}$, and any two $k$-chains $C$ and $C'$ from $\overline{d}$ to $\overline{e}$ in $\J$, we have $\core{C} \diamond d = \core{C'} \diamond d$.
\end{lemma}

\begin{proof}
Since $C$ and $C'$ have the same endpoints $\overline{d}$ and $\overline{e}$, we have $\core{C} \diamond w_{|q|+1, d}^d = \core{C'} \diamond w_{|q|+1, d}^d$.

If $|w_{|q|+1, d}^d| = |q| + 1$, then $\core{C} \diamond d = \core{C} \diamond (r_{|q+1|, d} w_{|q|+1, d}^d) = r_{|q+1|, d} (\core{C} \diamond w_{|q|+1, d}^d) = r_{|q+1|, d} (\core{C'} \diamond w_{|q|+1, d}^d) = \core{C'} \diamond d$.

Otherwise, $|w_{|q|+1, d}^d| < |q| + 1$, we have $r_{|q|+1, d} \in \domain{*}$. Therefore the action of $\core{C}$ (resp $\core{C'}$) cannot empty the initial word $w_{|q|+1, d}^d$, since it would lead, by Lemma~\ref{action}, to an element along $C$ (resp $C'$) being in $\overline{\domain{*}}$. Thus, we still have $\core{C} \diamond d = \core{C'} \diamond d$.
\end{proof}

Intuitively, the latter proof means that $\core{C}$ equals $\core{C'}$, up to deleting "dummy" steps in both chains, that are subsequences with shape $\rolestyle{S}_1 \dots \rolestyle{S}_p \rolestyle{S}_p^- \dots \rolestyle{S}_1^-$. But since the action of such dummy steps on any word is the identity, then the action of $C$ and $C'$ on $d$ are equal. Notice that in general, these dummy steps are necessary to go from $\overline{d}$ to $\overline{e}$, hence we cannot get rid of them by asking for some sort of minimality about chains.

This allows us to define an image for elements in a neighbourhood in the reduced interleaving regardless of the $k$-chain used to reach this element.
\begin{lemma}
The following mapping is well defined :
\begin{align*}
\rho_d : \N_{|q|}(\overline{d}, \J) & \rightarrow \N_{|q|}(d, \I') \\
\overline{e} & \mapsto \left\{
\begin{array}{lll}
e & & \text{if } \overline{e} \in \overline{\domain{*}} \\
\core{C} \diamond d \text{ with } C \text{ any } k \text{-chain from } \overline{d} \text{ to } \overline{e}, \text{ with } k \leq |q| & & \text{otherwise}
\end{array} \right.
\end{align*}
Furthermore, $\rho_d$ satisfies that for any $\overline{e} \in \domain{\J}$, we have $\rho_d(\overline{e}) \sim_1 e$. In particular, it satisfies conditions $1$ and $2$ from Theorem~\ref{morphismneighbourhoods}.
\end{lemma}

\begin{proof}
It only remains to prove that the action of a $k$-chain always provide an actual element in $\domain{\I'}$. We proceed by induction on $k$, building intermediate mappings $\rho_{d, k} : \N_k(\overline{d}, \J) \rightarrow \N_k(d, \I')$. We also prove that, at each step, we have for any $\overline{e} \in \N_k(\overline{d}, \J), \rho_{d, k}(\overline{e}) \sim_{|q| + 1 - k} e$.

\paragraph{Base case: $ k=0$.}
We have $\N_0(\overline{d}, \J) = \{ \overline{d} \}$. If $\overline{d} \in \overline{\domain{*}}$, we set $\rho_{d, 0} := d$ which is well-defined. Otherwise, consider the $0$-chain $(\overline{d})$, we have $\rho_{d, 0}(\overline{d}) := \varepsilon \diamond d = d$, which is well defined. In both cases, we obviously have $\rho_{d, 0}(\overline{d}) \sim_{|q| + 1} d$.

\paragraph{Induction step: $k \Rightarrow k+1$.}
Assume the mapping $\rho_{d, k} : \N_k(\overline{d},\J) \rightarrow \N_k(d, \I')$ is well defined for some $k < |q|$. We explain how to \emph{extend} it to a mapping $\rho_{d, k+1} : \N_{k+1}(\overline{d},\J) \rightarrow \N_{k+1}(d,\I')$. Consider an element $\overline{e} \in \N_{k+1}(\overline{d},\J) \setminus \N_k(\overline{d},\J)$.

If $e \in \domain{*}$, then we set $\rho_{d, k+1}(\overline{e}) := e$, which is well-defined and satisfies $\rho_{k+1}(\overline{e}) \sim_{|q| + 1 - (k+1)} e$.

Otherwise, $e \notin \domain{*}$, we know that there is a $k+1$-chain $(\overline{d_0}, \dots \overline{d_{k+1}})$ linking $\overline{d}$ and $\overline{e}$.
In particular, we have a role $\rolestyle{R}$ such that $(\overline{d_k}, \overline{e}) \in \rolestyle{R}^\J$. From the definition of $\rolestyle{R}^\J$, we can infer the existence of $\epsilon', \epsilon \in \domain{\canmod}$ such that:
$$\epsilon' \sim_{|q| + 1} d_k \qquad \epsilon \sim_{|q| + 1} e,$$
and either $\epsilon' = \epsilon \rolestyle{S}^-$ or  $ \epsilon = \epsilon' \rolestyle{S}$, where $\rolestyle{S}$ denotes the primary role of the edge $(\overline{d_k}, \overline{e})$. We consider these two cases in turn: 
\begin{itemize}
\item If $\epsilon' = \epsilon \rolestyle{S}^-$, we have $\rho_{d, k}(\overline{\epsilon'}) \sim_{|q| + 1 - k} d_k \sim_{|q| + 1 - k} \epsilon'$ due to $\overline{d_k} \in \N_k(\overline{d},\J)$ and the assumption for $k$, which implies that $\rho_k(\overline{d_k})$ ends with $\rolestyle{S}^-$. The action of $\rolestyle{S}$ on $\rho_{d, k}(d_k)$ hence provides the well-defined word obtained from $\rho_k(\overline{d_k})$ by removing its final symbol $\rolestyle{S}^-$. Therefore, $\rho_{k+1}(\overline{e})$ is well defined. Since the equivalence class of $\rho_{d, k}(\overline{d_k})$ for $\sim_{|q| + 1 -k}$ fully determines the equivalence class of its immediate neighbour $\rho_{d, k+1}(\overline{e})$ for $\sim_{|q| + 1 -(k+1)}$, and since we know that $\rho_{d, k}(\overline{d_k}) \sim_{|q| + 1 - k} d_k$ by the induction hypothesis, we obtain $\rho_{d, k+1}(\overline{e}) \sim_{|q| + 1 - (k+1)} e$.

\item Otherwise, if $\epsilon = \epsilon' \rolestyle{S}$, the action of $\rolestyle{S}$ on $\rho_{d, k}(d_k)$ provides $\rho_{d, k}(\overline{d_k}) \rolestyle{S}$, which is well defined since $\epsilon'$ and $\rho_{d, k}(d_k)$ must end by the same letter, as the induction hypothesis ensures $\epsilon' \sim_1 \rho_{d, k}(d_k)$ (recall $k < |q|$). Again, since the equivalence class of $\rho_{d, k}(\overline{d_k})$ for $\sim_{|q| + 1 -k}$ fully determines the equivalence class of its immediate neighbour $\rho_{d, k+1}(\overline{e})$ for $\sim_{|q| + 1 -(k+1)}$, and since we know that $\rho_{d, k}(\overline{d_k}) \sim_{|q| + 1 - k} d_k$ by the induction hypothesis, we obtain $\rho_{d, k+1}(\overline{e}) \sim_{|q| + 1 - (k+1)} e$.
\end{itemize}

The mapping $\rho_d$ is obtained as $\rho_{d, |q|}$. Notice the two conditions are satisfied.

\end{proof}

We now need to prove $\rho_d$ is a homomorphism. 
We start by proving the following lemma, which states the links of an element $e_1$ to elements in $\domain{*}$ fully determines such links for any other element $e_2$ that is $1$-equivalent to $e_1$.

\begin{lemma}
\label{caschi}
If $(\overline{e_1}, \overline{d}) \in \rolestyle{R}^\J$ for some $d \in \domain{*}$, and if $e_1 \sim_1 e_2$, then $(e_2, d) \in \rolestyle{R}^{\I'}$.
\end{lemma}

\begin{proof}
The definition of $\rolestyle{R}^\J$ provides $\epsilon_1, \delta \in \domain{\canmod}$ such that :
$$ \overline{f'(\epsilon_1)} = \overline{e_1} \qquad \overline{f'(\delta)} = \overline{d} \qquad (\epsilon_1, \delta) \in \rolestyle{R}^{\canmod}.$$
Notice that since $d \in \domain{*}$, we have $\overline{d} = \{ d \}$ and therefore $f'(\delta) = d$.
\begin{itemize}
\item
If $e_1 \in \domain{*}$, then $\overline{e_1}= \{ e_1 \}$. It follows that $f'(\epsilon_1) = e_1$ and $e_2 = e_1$.  We therefore obtain $(e_2, d) = (e_1, d) = (f'(\epsilon_1), f'(d)) \in f'(\rolestyle{R}^{\canmod}) = \rolestyle{R}^{\I'}$.
\item
Otherwise, $e_1 \notin \domain{*}$, which means that $f'(\epsilon_1) \not \in \domain{*}$, and hence  $f'(\epsilon_1) = \epsilon_1$ and
 $f'(\epsilon_1) \in  \Delta^{\canmod} \setminus \ainds(\Amc)$.
 Then, the definition of $\rolestyle{R}^\canmod$ provides a role $\rolestyle{S} \in \rni$ such that $\tbox \models \rolestyle{S} \incl \rolestyle{R}$ and 
either $\epsilon_1 = \delta \rolestyle{S}^- $ or $\delta = \epsilon_1 \rolestyle{S}$. Let us consider these two cases in turn:
\begin{itemize}
\item
If $\epsilon_1 = \delta \rolestyle{S}^-$, then the $1$-root of $f'(\epsilon_1)=\epsilon$ is $f'(\delta)$ and $w_{1, e_1}^{e_1} = \rolestyle{S}^-$. We thus have: $\chi_{1, f'(\epsilon_1)}(\varepsilon) = f'(\delta) = d$ (where $\varepsilon$ denotes the empty word). But since $f'(\epsilon_1) \sim_1 e_1 \sim_1 e_2$, we have $\chi_{1, e_2} = \chi_{1, f'(\epsilon_1)}$ and $w_{1, e_2}^{e_2} = w_{1, e_1}^{e_1}$. Combining the preceding facts, we obtain $(e_2, d) = (r_{1, e_2} w_{1, e_2}^{e_2}, \chi_{1, e_2}(\varepsilon)) = (f'(r_{1, e_2} \rolestyle{S}^-), f'(r_{1, e_2})) \in f'(\rolestyle{R}^{\canmod}) = \rolestyle{R}^{\I'}$.
\item
If $\delta = \epsilon_1 \rolestyle{S}$, then we have $\chi_{1, f'(\epsilon_1)}(w_{1, f'(\epsilon_1)}^{f'(\epsilon_1)}\rolestyle{S}) = f'(\delta) = d$. But since $f'(\epsilon_1) \sim_1 e_1 \sim_1 e_2$, we have $\chi_{1, e_2} = \chi_{1, f'(\epsilon_1)}$ and $w_{1, e_2}^{e_2} = w_{1, f'(\epsilon_1)}^{f'(\epsilon_1)}$. Hence : $(e_2, d) = (r_{1, e_2} w_{1, e_2}^{e_2}, \chi_{1, f'(\epsilon_1)}(w_{1, f'(\epsilon_1)}^{f'(\epsilon_1)}\rolestyle{S})) = (r_{1, e_2} w_{1, e_2}^{e_2}, \chi_{1, e_2}(w_{1, e_2}^{e_2}\rolestyle{S})) = (f'(r_{1, e_2}w_{1, e_2}^{e_2} ), f'(r_{1, e_2} w_{1, e_2}^{e_2} \rolestyle{S})) \in f'(\rolestyle{R}^{\canmod}) = \rolestyle{R}^{\I'}$.
\end{itemize}
\end{itemize}
\end{proof}

\begin{proof}[Proof of Theorem~\ref{morphismneighbourhoods}]
Let $e, e' \in \domain{\I'}$, and let $\rolestyle{R} \in \rnames$ be such that $(\overline{e}, \overline{e'}) \in \rolestyle{R}^\J$.

If $\overline{e} \in \overline{\domain{*}}$, then Lemma~\ref{caschi} applies by setting $(e_1, e_2, d, \rolestyle{R}) := (e', \rho_d(\overline{e'}), e, \rolestyle{R})$ and therefore $(\rho_d(\overline{e}), \rho_d(\overline{e'}) \in \rolestyle{R}^{\I'}$.

Otherwise, if $\overline{e'} \in \overline{\domain{*}}$, then Lemma~\ref{caschi} also applies, by setting $(e_1, e_2, d, \rolestyle{R}) := (e, \rho_d(\overline{e}), e', \rolestyle{R})$ and again $(\rho_d(\overline{e}), \rho_d(\overline{e'}) \in \rolestyle{R}^{\I'}$.

Otherwise, $\overline{e}, \overline{e'} \notin \overline{\domain{*}}$. Notice that both cannot be in $\N_{|q|}(\overline{d}, \J) \setminus \N_{|q|-1}(\overline{d}, \J)$ at the same time. Indeed, if both were, there would be a $2|q|+1$-chain from $\overline{d}$ to $\overline{d}$ (recall $\overline{e}, \overline{e'} \notin \overline{\domain{*}}$). However, the depth modulo $2|q| + 3$ encoded in each equivalent class along this chain only increases or decreases by $1$ at each step (since none of its element belongs to $\overline{\domain{*}}$). Hence, $|d| \mod 2|q| + 3$ would equal itself up to $2|q|+1$ such $1$-steps, which is impossible modulo $2|q|+3$. Therefore, we can assume that $e \in \N_{|q|-1}(\overline{d}, \J)$. We have a $k$-chain $C_{\overline{d} \rightarrow \overline{e}}$ from $\overline{d}$ to $\overline{e}$, with $k < |q|$. Complete it by $\rolestyle{R}$ into a $k+1$ chain $C_{\overline{d} \rightarrow \overline{e'}}$ to reach $\overline{e'}$. Since $k + 1 \leq |q|$, we have by definition $\rho_d(e') = \rolestyle{S} \diamond \rho_d(\overline{e})$, with $\rolestyle{S}$ the primary role of the edge $(e, e')$. In both cases, $\rho_d(\overline{e'})$ ending by $\rolestyle{S}^-$ or not, it ensures that $(\rho_d(e), \rho_d(e') \in \rolestyle{S}^{\I'} \subseteq \rolestyle{R}^{\I'}$.

The preservation of positive concepts follows. Indeed, if we have an element $\overline{e}$ and a concept name $\rolestyle{A}$ such that $\overline{e} \in \rolestyle{A}^{\J}$, then, from the definition of $\rolestyle{A}^{\J}$, either $\overline{e}$ is the interpretation of an individual name $\ind{e}$ and $\rolestyle{B}(\ind{e}) \in \abox$ for some concept $\rolestyle{B}$ such that $\tbox \models \rolestyle{B} \incl \rolestyle{A}$, or there exists another element $\overline{e'}$ connected to $\overline{e}$ by a positive role $\rolestyle{S}$ such that $\tbox \models \exists \rolestyle{S}^- \incl \rolestyle{A}$.

In the first case, we have in particular $\overline{e} \in \overline{\domain{*}}$, thus $\rho_d(\overline{e}) = e = \ind{e}^{\I'} \in \rolestyle{A}^{\I'}$ since $\I'$ is a model.

In the second case, since $\rho_d$ preserves positive roles, we have $(\rho_d(e'), \rho_d(e)) \in \rolestyle{S}^{\I'}$, and therefore $\rho_d(e) \in \rolestyle{A}^{\I'}$ since $\I'$ is a model.

\end{proof}

\end{toappendix}

\subsection{Case of Finite-Depth TBoxes}
We give an improved upper bound for finite-depth TBoxes (which arguably cover many practical ontologies \cite{DBLP:journals/jair/GrauHKKMMW13}), pinpointing the exact combined complexity. 

\begin{theorem}
\label{fcm}
For finite-depth \dlh\ TBoxes, CCQ answering is $\coNEXP$-complete w.r.t.\ combined complexity.
\end{theorem}

\begin{proof}[Proof sketch]Fix a KB $\Kmc=(\Tmc, \Amc)$. 
If 
$\tbox$ has finite depth, then $\canmod$ contains at most $|\ainds(\Amc)| \times |\tbox|^{|\tbox|}$ elements, 
which implies that, for every model $\Imc$ of $\Kmc$, the interleaving of $\Imc$  
is finite and of single exponential size in $|\kb|$. 
Since the interleaving of a countermodel is itself a countermodel, 
this shows that the smallest countermodel is of single-exponential size, from which derives the improved $\coNEXP$ upper bound.
\end{proof}

We note that the proofs of the $\coNP$ and $\PIP$ lower bounds listed in Table \ref{resultsgeneral} 
already use finite-depth TBoxes.

\section{Rooted Counting CQs}\label{sec:rooted}


\newcommand{\rootedproblem}{\problemstyle{\problem^{rooted}}}
\newcommand{\stronglyproblem}{\problemstyle{\problem{strg. root.}}}
We next explore whether structural restrictions on CCQs allow us to obtain lower complexity. 
As the lower bounds from \cite{kostylevreutter:count} use disconnected counting variables,
a natural idea is to consider 
the subclass of
 \emph{rooted} queries that were introduced in 
\cite{DBLP:conf/kr/BienvenuLW12} 
and are believed to capture a large portion of real-world CQs. 

Rooted CCQs can be defined analogously to rooted CQs. The definition utilizes the notion of a Gaifman graph of a CCQ, 
whose vertices are the query terms, and which has an undirected edge $\{t_1,t_2\}$ iff $t_1,t_2$ co-occur in a role atom. 
\label{subsectionrooted}
\begin{definition}
\label{defrooted}
A CCQ $q(\tx) := \exists \ty\exists \tz \psi(\tx, \ty, \tz)$ is \emph{rooted} if every connected component of the
Gaifman graph of $q$ contains at least one answer variable or individual name.
\end{definition}
\noindent Example queries $q_2$ and $q_3$ are rooted, while $q_1$ is not. 

Rootedness has been shown to lower the complexity of reasoning in several settings. Most relevant to us 
is a recent result by \citeauthor{nikolauetal:bag}\ (\citeyear{nikolauetal:bag}) that rooted CQ answering under bag semantics\footnote{Bag semantics, which underly practical database systems, interprets relations using multisets rather than sets \cite{DBLP:conf/vldb/Albert91}.} 
 has tractable data complexity in \dlc, and furthermore, the same holds for rooted versions of the $Count()$-queries of \citeauthor{kostylevreutter:count} under suitable restrictions on the TBox. These techniques can be adapted to show tractability for arbitrary \dlc TBoxes: 

\begin{theoremrep}\textbf{(Implicit in \cite{nikolauetal:bag,DBLP:conf/semweb/CimaNKKGH19})}\label{tcz}
In \dlc, \exhaustive rooted CCQ answering is $\TC^0$-complete\footnote{We recall that 
$\TC^0$ is a circuit complexity class defined similarly to $\AC^0$ but additionally allowing threshhold gates. It is known that $\AC^0 \subsetneq \TC^0 \subseteq \mathsf{NC}^1 \subseteq \mathsf{LogSpace} \subseteq \mathsf{PTime}$.} w.r.t.\ data complexity.
\end{theoremrep}
\begin{proofsketch}
\citeauthor{nikolauetal:bag}\ prove that answering rooted CQs under bag semantics can be done via a rewriting to BCALC,  whose evaluation problem is known to be in $\TC^0$ due to \cite{DBLP:conf/icdt/Libkin01}, see \cite{DBLP:conf/semweb/CimaNKKGH19} for discussion. Moreover, they further show that for a syntactically restricted class of \dlc\ TBoxes, it is possible to reduce \exhaustive rooted CCQ answering to rooted CQ answering under bag semantics. To obtain $\TC^0$ membership for unrestricted TBoxes, the BCALC rewriting can be adapted to set-based rather than bag interpretations. In the long version, we provide an alternative self-contained proof which directly constructs a family of $\TC^0$ circuits.  
A matching lower bound has not been stated, 
but can be shown by a simple reduction (using an empty TBox) from the $\TC^0$-complete problem that asks, given a binary string $s$ and number $k$, whether the number of $1$-bits in $s$ exceeds $k$ \cite{aehlig:tc0}. 
%
\end{proofsketch}

\begin{proof}
We start with the $\TC^0$ hardness. The reduction from the $\numones$ problem works as follows: given an instance $(s, k)$, we create an ABox $\abox_s := \{ \rolestyle{R}(\ind{a}, \ind{s}_k) \suchthat s_k \in s \wedge s_k = 1 \}$, along with the empty TBox $\tbox = \emptyset$ and \exhaustive rooted CCQ $q := \exists z\ \rolestyle{R}(\ind{a}, z)$. It is clear that $(\emptyset, [k, +\infty]) \in \qIntervals{q}{(\tbox, \abox_s)} \Longleftrightarrow (s, k) \in \numones$. It can be verified that this simple reduction can be implemented by $\AC^0$ circuits (so constitutes an $\AC^0$-reduction, as required).  

As explained in the body of the paper, $\TC^0$ membership follows from results in \cite{nikolauetal:bag}. While that paper only formally states membership in $\mathsf{LogSpace}$, a follow-up paper on bag semantics \cite{DBLP:conf/semweb/CimaNKKGH19} states $\TC^0$ membership for \dlf\ (which properly contains \dlc), by making use of prior complexity results for bag relational algebra. We believe it is nevertheless instructive to have a direct proof and therefore describe what follows how to construct a family of $\TC^0$ circuits to decide our problem.

We need a family of circuits in order to be able to handle ABoxes of different sizes. More precisely, we will create one circuit for each possible number $\ell$ of individual names. We can assume w.l.o.g.\ that the same set of individuals, denoted $\ainds_\ell$, is used for all of the ABoxes having $\ell$ individuals. Let us now explain how to represent an input $(\abox^*, \ta^*, m^*)$ to the circuit that handles $\ell$-individual ABoxes. 
%
\begin{itemize}
\item
Each atomic role $\rolestyle{P}$ appearing in $\Tmc$ and/or $q$ is represented by input gates $\inputgate_{\rolestyle{P}(\ind{a},\ind{b}) \in \abox ?}$ for $\ind{a},\ind{b} \in \individuals_\ell$. The gate $\inputgate_{\rolestyle{P}(\ind{a},\ind{b}) \in \abox ?}$ is set to $1$ iff $\rolestyle{P}(\ind{a},\ind{b}) \in \abox^*$.
\item
Each atomic concept $\rolestyle{A}$ appearing in $\Tmc$ and/or $q$ is represented by input gates $\inputgate_{\rolestyle{A}(\ind{a}) \in \abox ?}$ for $\ind{a} \in \individuals_\ell$. The gate $\inputgate_{\rolestyle{A}(\ind{a}) \in \abox ?}$ is set to $1$ iff $\rolestyle{A}(\ind{a}) \in \abox^*$.
\item
The tuple $\ta^*$ is represented by input gates $\inputgate_{\ta_k = \ind{a}}$ for $1\leq k \leq |\tx|$ and $\ind{a} \in \individuals_\ell$. The gate  is set to $1$ iff $\ta^*_k = \ind{a}$. 
\item 
The integer $m^*$ is represented in binary by input gates $\inputgate_{b_k = 1}$ for each $0 \leq k < \log_2(|\ainds(\Amc^*)| + | \tbox |)^{|q|})$.  The gate $\inputgate_{b_k = 1}$ is set to $1$ iff the $k^{th}$ bit of $m^*$ is $1$ (with $0^{th}$-bit being the least significant bit).
\end{itemize}
Regarding the last point, we use the observation from \cite{kostylevreutter:count} that if $(\ta^*, [m^*, +\infty]) \in \qIntervals{q}{{(\Tmc, \Amc^*)}}$, then $m^*$ cannot exceed $(|\ainds(\Amc^*)| + | \tbox |)^{|q|}=(|\ainds_\ell| + | \tbox |)^{|q|}$. This is a direct consequence of the fact that every satisfiable \dlh\ KB $\Kmc = (\Tmc, \Amc)$ has a model with at most $|\ainds(\Amc)| + | \tbox |$ elements.

We now describe the other parts of the circuit. 
%
%
%
%
%
%
%
We introduce, for each relevant positive concept $\rolestyle{C}$ (i.e., atomic concept or existential concept $\exists R$ that uses concept and role names from $\Tmc$ and/or $q$) and each individual name $\ind{a} \in \ainds_\ell$, a disjunctive gate $\orgate_{\kb \models \rolestyle{C}(\ind{a}) ?}$ taking as inputs:
\begin{itemize}
\item
$\inputgate_{\rolestyle{A}(\ind{a}) \in \abox ?}$ for each atomic concept $\rolestyle{A}$ such that $\tbox \models \rolestyle{A} \incl \rolestyle{C}$.
\item
$\inputgate_{\rolestyle{P}(\ind{a}, \ind{b}) \in \abox ?}$ for all $\ind{b} \in \individuals(\abox)$ such that $\tbox \models \exists \rolestyle{P} \incl \rolestyle{C}$.
\item
$\inputgate_{\rolestyle{P}(\ind{b}, \ind{a}) \in \abox ?}$ for all $\ind{b} \in \individuals(\abox)$ such that $\tbox \models \exists \rolestyle{P^-} \incl \rolestyle{C}$.
\end{itemize}

The preceding gates determine the ABox part of the canonical model. We next need to decide the existence of each element of the form $\ind{a} w$, where $a \in \individuals(\abox)$ and $w \in \Gamma_{q, \tbox} \setminus \varepsilon$ (by Lemma \ref{words}, these are the only anonymous elements that can occur in a match for $q$). 
For each such $\ind{a} w$, we denote by $\rolestyle{R}_w$ the first role name of $w$ and introduce a conjunctive gate $\andgate_{
 \ind{a} w \in \domain{\canmod} ?}$ which takes as input:
\begin{itemize}
\item The negation $\notgate_{\forall \ind{b} \in \individuals(\abox)\ \neg \rolestyle{R}(\ind{a}, \ind{b}) ?}$ of a disjunctive gate $\orgate_{\exists \ind{b} \in \individuals(\abox)\ \rolestyle{R}(\ind{a}, \ind{b}) ?}$ taking as inputs the gates:
\begin{itemize}
\item $\inputgate_{\rolestyle{P}(\ind{a}, \ind{b})}$ for each $\ind{b} \in \individuals_\ell$, if $\rolestyle{R} = \rolestyle{P} \in \rnames$. 
\item $\inputgate_{\rolestyle{P}(\ind{b}, \ind{a})}$ for each $\ind{b} \in \individuals_\ell$ if $\rolestyle{R} = \rolestyle{P^-}$ with $\rolestyle{P} \in \rnames$.
\end{itemize}
which verifies that there is not already a $\rolestyle{R}_w$-successor to $\ind{a}$. 
\item The gate $\orgate_{\kb \models \exists \rolestyle{R}_w(\ind{a}) ?}$ that checks that a witnessing $\rolestyle{R}_w$-successor is needed.
\end{itemize}

The circuit next determines for each mapping $\match : \tx \cup \tz \mapsto \{\ind{a} w| \ind{a} \in  \individuals_\ell, w \in \Gamma_{q, \tbox}\}$, whether $\match$ is a match for $q(\ta^*)$. Notice that, regardless of the input ABox, 
we can restrict to a set of relevant mappings by keeping only those which map the answer variables $\tx$ to individuals from $\individuals_\ell$ and which map variables $v_1,v_2$ occurring in a role atom $\rolestyle{R}(v_1, v_2)$ from $q$ onto either:
\begin{itemize}
\item a pair of individual names, or
\item a pair $w_1, w_2$ such that $w_2 = w_1\rolestyle{R}$ or $w_1 = \rolestyle{R^-}w_2$.
\end{itemize}
Similarly, we can restrict the set of relevant mappings by keeping only those which map variable $v$ occuring in a concept atom $\rolestyle{A}(v)$ from $q$ onto either an individual name, or an element $aw\rolestyle{R}$, where $\kb \models \exists\rolestyle{R}^- \sqsubseteq \rolestyle{A}$.
Clearly, any mapping $\match$ that does not respect these conditions cannot be a match, due to 
 the definition of $\rolestyle{R}^{\canmod}$. 
This restriction simplifies the process of checking if a mapping is a match for $q(\ta^*)$: we are only left with verifying the existence of the anonymous elements in its image, as well as the validity of the atoms mapped onto the ABox part of the canonical model. 

For each relevant mapping $\match$, we introduce a conjunctive gate $\andgate_{\match \text{ match ?}}$ taking as inputs all gates:
\begin{itemize}
\item
$\inputgate_{\ta_k = \match(x_k) ?}$ for each $1\leq k \leq |\tx|$ (to check $\tx$ is mapped on $\ta^*$).
\item 
$\andgate_{
\match(z) \in \domain{\canmod} ?}$ for each $z \in \tz$ such that $\match(z) \notin \individuals_\ell$ (to check for existence of $\match(z)$ under input $\abox^*$).
\item
$\inputgate_{\rolestyle{R}(\match(v_1), \match(v_2)) \in \abox ?}$ for each $v_1, v_2 \in \tx \cup \tz$ such that $\rolestyle{R}(v_1,v_2) \in q$ and $\match(v_1), \match(v_2) \in \individuals(\abox)$ (to check the validity of the mapping for pairs of variables mapped on individual names).
\item 
$\orgate_{\kb \models \rolestyle{A}(\match(v)) ?}$ for each $v \in \tx \cup \tz$ such that $\rolestyle{A}(v) \in q$ and $\match(v) \in \individuals(\abox)$  (to check the validity of the mapping for variables mapped on individual names).
\end{itemize}

We will next use threshold gates in order to compute the total number of matches. Introduce, for each $k = 0, \dots, (\individuals_\ell \times \Gamma_{q, \tbox})^{|q|}$, a threshold gate $\tgate{k}_{\qAnswers{q}{\canmod}_{\ta} \geq k ?}$ taking as input every $\andgate_{\match \text{ match ?}}$. The gate $\tgate{k}_{\qAnswers{q}{\canmod}_{\ta} \geq k ?}$ returns 1 iff at least $k$ of its inputs are 1. By construction, the latter holds iff there are at least $k$ matches for $q(\ta^*)$. 

In parallel, we introduce a conjunctive gate $\andgate_{m = k ?}$ for each $k = 0, \ldots, (\individuals_\ell \times \Gamma_{q, \tbox})^{|q|}$ taking as inputs:
\begin{itemize}
\item  the input gates $\inputgate_{b_j = 1 ?}$ such that the $j^{th}$ bit of the binary encoding of $k$ is $1$
\item the negation of each input gate $\inputgate_{b_j = 1 ?}$ such that the $j^{th}$ bit of the binary encoding of $k$ is $0$
\end{itemize}
The gate $\andgate_{m = k?}$ returns $1$ iff $m^*=k$. 

We combine the preceding two types of gates to compare $m^*$ and the computed number of matches. For each $k = 0, \ldots, (\individuals_\ell \times \Gamma_{q, \tbox})^{|q|}$, we introduce a conjunctive gate $\andgate_{\qAnswers{q}{\canmod}_\ta \geq m ?}$ taking as input $\tgate{k}_{\qAnswers{q}{\canmod}_{\ta} \geq k ?}$ and $\andgate_{m = k ?}$. 

Finally, our output gate is a disjunctive gate $\orgate_{\mathsf{output}}$ taking as inputs all gates $\andgate_{\qAnswers{q}{\canmod}_\ta \geq m ?}$. By construction, this gate outputs 1 iff there are at least $m^*$ matches of $q(\ta^*)$  in the canonical model of the considered KB.


The depth of the circuit is $7$, and is hence constant, showing membership in $\TC^0$.
\end{proof}

\begin{toappendix}
\begin{example}

Let $\tbox$ be the following $\DLc$ TBox
$$\tbox = \{ \rolestyle{D} \incl \exists \rolestyle{R}, \exists \rolestyle{R}^- \incl \exists \rolestyle{R} \}$$
and $q$ be the \exhaustive rooted CCQ given by 
$$q(x) = \exists z ~ \rolestyle{R}(x, z)$$
Observe that even if the second axiom from $\tbox$ suggests the need to consider suffixes $\rolestyle{R} \dots \rolestyle{R}$ of arbitrary length, we only have $\Gamma_{q, \tbox} = \{ \varepsilon, \rolestyle{R} \}$.

We propose to illustrate the construction of the circuit designed for $2$-individual ABoxes, with individual names $\ind{a}$ and $\ind{b}$.
We thus require $\lceil { \log_2( (2 + 2)^1 )} \rceil + 1 = 3$ input gates representing the input integer, and we have $6$ relevant matches given by:
\begin{itemize}
\item $(\match_1) \qquad x \mapsto \ind{a} \qquad z \mapsto \ind{a}$
\item $(\match_2) \qquad x \mapsto \ind{a} \qquad z \mapsto \ind{b}$
\item $(\match_3) \qquad x \mapsto \ind{a} \qquad z \mapsto \ind{aR}$
\item $(\match_4) \qquad x \mapsto \ind{b} \qquad z \mapsto \ind{b}$
\item $(\match_5) \qquad x \mapsto \ind{b} \qquad z \mapsto \ind{a}$
\item $(\match_6) \qquad x \mapsto \ind{b} \qquad z \mapsto \ind{bR}$
\end{itemize}
The corresponding circuit is depicted in Figure~\ref{circuitexample}.

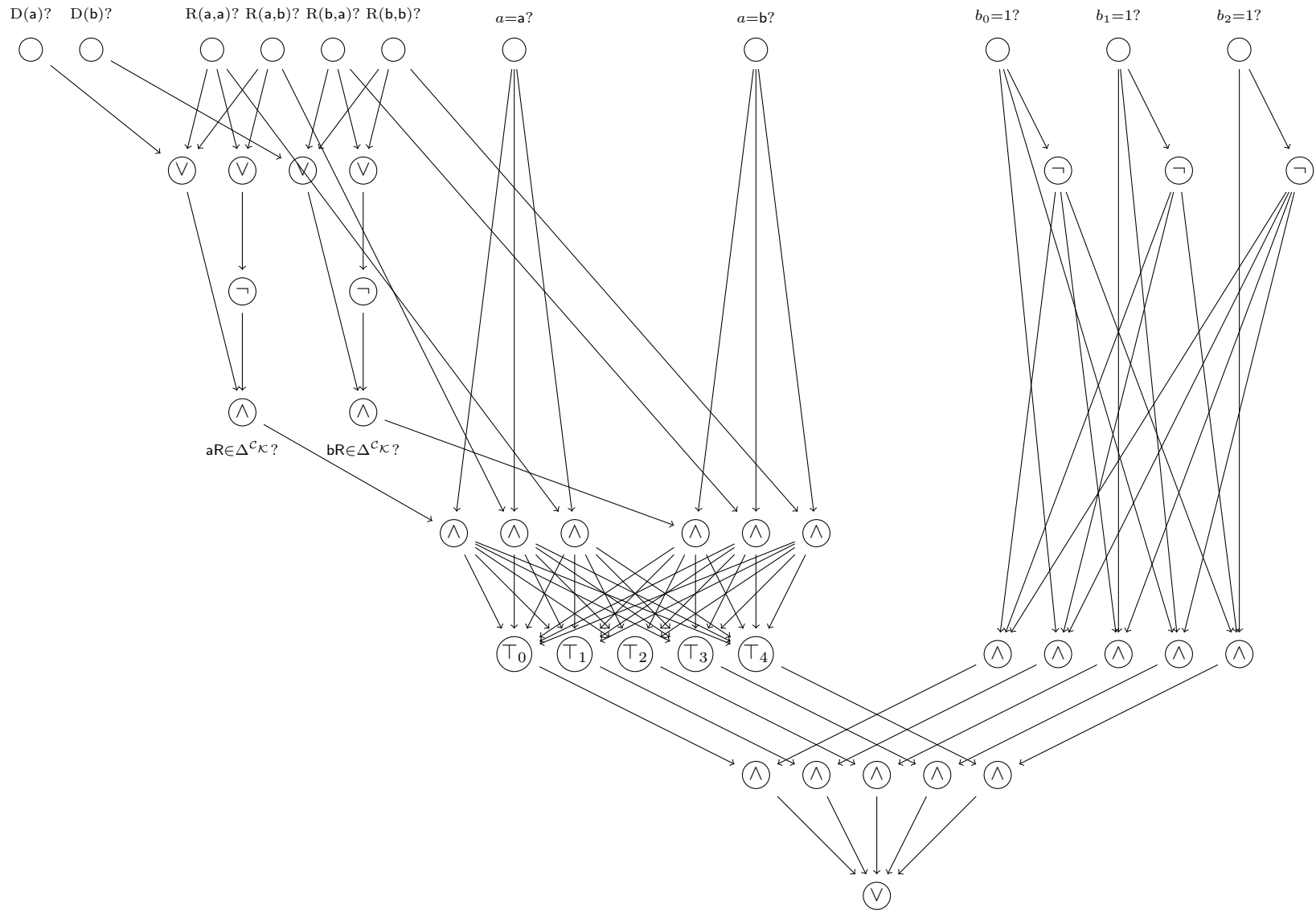
\begin{sidewaysfigure}
\begin{tikzpicture}
\node at ( 1, 14) (ica) 		[label=above:$ \scriptstyle{\rolestyle{D}(\ind{a})?}$]		{$\inputgate$};
\node at ( 2,14) (icb) 		[label=above:$ \scriptstyle{\rolestyle{D}(\ind{b})?}$]		{$\inputgate$};
\node at ( 4,14) (iraa) 		[label=above:$\scriptstyle {\rolestyle{R}(\ind{a}, \ind{a})?}$]		{$\inputgate$};
\node at ( 5,14) (irab) 		[label=above:$\scriptstyle {\rolestyle{R}(\ind{a}, \ind{b})?}$]		{$\inputgate$};
\node at ( 6,14) (irba) 		[label=above:$\scriptstyle {\rolestyle{R}(\ind{b}, \ind{a})?}$]		{$\inputgate$};
\node at ( 7,14) (irbb) 		[label=above:$\scriptstyle {\rolestyle{R}(\ind{b}, \ind{b})?}$]		{$\inputgate$};
\node at ( 9,14) (iaa) 		[label=above:$\scriptstyle {a = \ind{a}?}$]		{$\inputgate$};
\node at (13,14) (iab) 		[label=above:$\scriptstyle {a = \ind{b}?}$]		{$\inputgate$};
\node at (17,14) (in0) 		[label=above:$\scriptstyle {b_0 = 1 ?}$]		{$\inputgate$};
\node at (19,14) (in1) 		[label=above:$\scriptstyle {b_1 = 1 ?}$]		{$\inputgate$};
\node at (21,14) (in2) 		[label=above:$\scriptstyle {b_2 = 1 ?}$]		{$\inputgate$};

\node at (3.5,12) (cera) 		[]		{$\orgate$};
\node at (5.5,12) (cerb) 		[]		{$\orgate$};
\node at (4.5,12) (eira) 		[]		{$\orgate$};
\node at (6.5,12) (eirb) 		[]		{$\orgate$};

\node at ( 18,12) (nin0) 		[]		{$\notgate$};
\node at ( 20,12) (nin1) 		[]		{$\notgate$};
\node at ( 22,12) (nin2) 		[]		{$\notgate$};

\node at (4.5,10) (neira) 		[]		{$\notgate$};
\node at (6.5,10) (neirb) 		[]		{$\notgate$};

\node at ( 4.5,8) (ar) 		[label=below:$\scriptstyle {\ind{aR} \in \domain{\canmod} ?}$]		{$\andgate$};
\node at ( 6.5,8) (br) 		[label=below:$\scriptstyle {\ind{bR} \in \domain{\canmod} ?}$]		{$\andgate$};

\node at (10,6) (m1) 		[]		{$\andgate$};
\node at (9,6) (m2) 		[]		{$\andgate$};
\node at (8,6) (m3) 		[]		{$\andgate$};
\node at (14,6) (m4) 		[]		{$\andgate$};
\node at (13,6) (m5) 		[]		{$\andgate$};
\node at (12,6) (m6) 		[]		{$\andgate$};

\node at ( 9,4) (t0) 		[]		{$\tgate{0}$};
\node at (10,4) (t1) 		[]		{$\tgate{1}$};
\node at (11,4) (t2) 		[]		{$\tgate{2}$};
\node at (12,4) (t3) 		[]		{$\tgate{3}$};
\node at (13,4) (t4) 		[]		{$\tgate{4}$};

\node at ( 17,4) (i0) 		[]		{$\andgate$};
\node at ( 18,4) (i1) 		[]		{$\andgate$};
\node at ( 19,4) (i2) 		[]		{$\andgate$};
\node at ( 20,4) (i3) 		[]		{$\andgate$};
\node at ( 21,4) (i4) 		[]		{$\andgate$};

\node at ( 13,2) (mt0) 		[]		{$\andgate$};
\node at ( 14,2) (mt1) 		[]		{$\andgate$};
\node at ( 15,2) (mt2) 		[]		{$\andgate$};
\node at (16,2) (mt3) 		[]		{$\andgate$};
\node at (17,2) (mt4) 		[]		{$\andgate$};

\node at ( 15, 0) (o) 		[]		{$\orgate$};

\path[every node/.style={sloped}]

(ica) edge [->] node[below] 			{} (cera)
(iraa) edge [->] node[below] 			{} (cera)
(irab) edge [->] node[below] 			{} (cera)

(icb) edge [->] node[below] 			{} (cerb)
(irbb) edge [->] node[below] 			{} (cerb)
(irba) edge [->] node[below] 			{} (cerb)

(iraa) edge [->] node[below] 			{} (eira)
(irab) edge [->] node[below] 			{} (eira)

(irbb) edge [->] node[below] 			{} (eirb)
(irba) edge [->] node[below] 			{} (eirb)

(in0) edge [->] node[below] 			{} (nin0)
(in1) edge [->] node[below] 			{} (nin1)
(in2) edge [->] node[below] 			{} (nin2)

(eira) edge [->] node[below] 			{} (neira)
(eirb) edge [->] node[below] 			{} (neirb)

(neira) edge [->] node[below] 			{} (ar)
(cera) edge [->] node[below] 			{} (ar)

(neirb) edge [->] node[below] 			{} (br)
(cerb) edge [->] node[below] 			{} (br)

(iaa) edge [->] node[below] 			{} (m1)
(iraa) edge [->] node[below] 			{} (m1)

(iaa) edge [->] node[below] 			{} (m2)
(irab) edge [->] node[below] 			{} (m2)

(iaa) edge [->] node[below] 			{} (m3)
(ar) edge [->] node[below] 			{} (m3)

(iab) edge [->] node[below] 			{} (m4)
(irbb) edge [->] node[below] 			{} (m4)

(iab) edge [->] node[below] 			{} (m5)
(irba) edge [->] node[below] 			{} (m5)

(iab) edge [->] node[below] 			{} (m6)
(br) edge [->] node[below] 			{} (m6)

(m1) edge [->] node[below] 			{} (t0)
(m2) edge [->] node[below] 			{} (t0)
(m3) edge [->] node[below] 			{} (t0)
(m4) edge [->] node[below] 			{} (t0)
(m5) edge [->] node[below] 			{} (t0)
(m6) edge [->] node[below] 			{} (t0)

(m1) edge [->] node[below] 			{} (t1)
(m2) edge [->] node[below] 			{} (t1)
(m3) edge [->] node[below] 			{} (t1)
(m4) edge [->] node[below] 			{} (t1)
(m5) edge [->] node[below] 			{} (t1)
(m6) edge [->] node[below] 			{} (t1)

(m1) edge [->] node[below] 			{} (t2)
(m2) edge [->] node[below] 			{} (t2)
(m3) edge [->] node[below] 			{} (t2)
(m4) edge [->] node[below] 			{} (t2)
(m5) edge [->] node[below] 			{} (t2)
(m6) edge [->] node[below] 			{} (t2)

(m1) edge [->] node[below] 			{} (t3)
(m2) edge [->] node[below] 			{} (t3)
(m3) edge [->] node[below] 			{} (t3)
(m4) edge [->] node[below] 			{} (t3)
(m5) edge [->] node[below] 			{} (t3)
(m6) edge [->] node[below] 			{} (t3)

(m1) edge [->] node[below] 			{} (t4)
(m2) edge [->] node[below] 			{} (t4)
(m3) edge [->] node[below] 			{} (t4)
(m4) edge [->] node[below] 			{} (t4)
(m5) edge [->] node[below] 			{} (t4)
(m6) edge [->] node[below] 			{} (t4)

(nin0) edge [->] node[below] 			{} (i0)
(nin1) edge [->] node[below] 			{} (i0)
(nin2) edge [->] node[below] 			{} (i0)

(in0) edge [->] node[below] 			{} (i1)
(nin1) edge [->] node[below] 			{} (i1)
(nin2) edge [->] node[below] 			{} (i1)

(nin0) edge [->] node[below] 			{} (i2)
(in1) edge [->] node[below] 			{} (i2)
(nin2) edge [->] node[below] 			{} (i2)

(in0) edge [->] node[below] 			{} (i3)
(in1) edge [->] node[below] 			{} (i3)
(nin2) edge [->] node[below] 			{} (i3)

(nin0) edge [->] node[below] 			{} (i4)
(nin1) edge [->] node[below] 			{} (i4)
(in2) edge [->] node[below] 			{} (i4)

(i0) edge [->] node[below] 			{} (mt0)
(t0) edge [->] node[below] 			{} (mt0)

(i1) edge [->] node[below] 			{} (mt1)
(t1) edge [->] node[below] 			{} (mt1)

(i2) edge [->] node[below] 			{} (mt2)
(t2) edge [->] node[below] 			{} (mt2)

(i3) edge [->] node[below] 			{} (mt3)
(t3) edge [->] node[below] 			{} (mt3)

(i4) edge [->] node[below] 			{} (mt4)
(t4) edge [->] node[below] 			{} (mt4)

(mt0) edge [->] node[below] 			{} (o)
(mt1) edge [->] node[below] 			{} (o)
(mt2) edge [->] node[below] 			{} (o)
(mt3) edge [->] node[below] 			{} (o)
(mt4) edge [->] node[below] 			{} (o)
;

\end{tikzpicture}
\caption{The $\TC^0$ circuit built for $2$-individual ABoxes w.r.t. $\tbox$ and $q$.}
\label{circuitexample}
\end{sidewaysfigure}

\end{example}
\medskip
\end{toappendix}

The preceding 
result naturally 
leads us to ask whether rootedness also bring benefits for general CCQs.  Unfortunately, we show that restricting to rooted CCQs (without exhaustiveness) 
does not allow us to escape existing hardness results:


\begin{theoremrep}
\label{rootedconphard}
In \dlc, rooted CCQ answering is $\coNP$-complete w.r.t.\ data complexity.
\end{theoremrep}
\def\TmcCol{\Tmc_{\mathsf{col}}}
\begin{proofsketch}
The proof borrows some ideas from the proofs of Lemmas 12 and 16 from \cite{kostylevreutter:count}. It proceeds by reduction from the well-known $\coNP$-complete  $\tcol$ problem: given an undirected graph $\graph = (\vertices, \edges)$,
return yes iff $\graph$ has \emph{no} 3-coloring, i.e., a mapping from $\vertices$ to $\{\ind{red},\ind{green},\ind{blue}\}$ such that adjacent vertices map to different colors (equivalently: there is no monochromatic edge). 

The reduction uses atomic roles $\rolestyle{Edge}$ and $\rolestyle{Vertex}$ to encode the graph and $\rolestyle{HasCol}$ to assign colors. 
The TBox $\TmcCol$ has a single axiom: $\exists \rolestyle{Vertex}^- \sqsubseteq \exists \rolestyle{HasCol}$.
The ABox $\abox_{\graph}$ contains an individual $\mathsf{v}$ for each vertex $v \in \vertices$ and an assertion $\rolestyle{Edge}(\ind{u}, \ind{v})$ for each edge $\{u, v\} \in \edges$. All vertices are connected to a special root individual $\ind{a}$: 
$\rolestyle{Vertex}(\ind{a}, \ind{u})$, for each $u \in \vertices$. The three colors are represented by individuals $\ind{r}, \ind{g}$ and $\ind{b}$. To ensure that the query has matches in every model, we include a `dummy' vertex individual $\ind{a_v}$ and the following assertions: $\rolestyle{Vertex}(\ind{a}, \ind{a_v})$, $\rolestyle{Edge}(\ind{a_v}, \ind{a_v})$, $\rolestyle{HasCol}(\ind{a_v}, \ind{r})$, $\rolestyle{HasCol}(\ind{a_v}, \ind{g})$, and $\rolestyle{HasCol}(\ind{a_v}, \ind{b})$.

The 
query $q$ is the conjunction of the two subqueries:
\begin{align*}
q^{edge} = &\, \exists y_c ~ \exists z_1 ~ \exists z_2 ~ \rolestyle{Vertex}(\ind{a}, z_1)  \wedge  \rolestyle{Vertex}(\ind{a}, z_2) \wedge \\
& \rolestyle{Edge}(z_1, z_2)  \wedge  \rolestyle{HasCol}(z_1, y_c)  \wedge  \rolestyle{HasCol}(z_2, y_c)\smallskip\\
q^{col}= &\, \exists y ~ \exists z  ~ \rolestyle{Vertex}(\ind{a}, y) \wedge \rolestyle{HasCol}(y, z)
\end{align*}
serving respectively to detect monochromatic edges and to check 
whether any additional colors have been introduced. 

By construction, there are at least 3 c-matches for $q(\emptyset)$
in any model of the KB $\Kmc_{\mathsf{col}}=(\TmcCol, \abox_{\graph})$. Moreover, it can be verified 
that  $(\emptyset, [4, +\infty])$ is a certain answer to $q$ w.r.t.\ $\Kmc_{\mathsf{col}}$ iff 
$\graph$ is not 3-colorable.
 \end{proofsketch}
\begin{proof}
We briefly recall the reduction sketched in the body of the paper. Starting from an instance $\graph = (\vertices, \edges)$ of the decision problem $\tcol$, we consider the ABox $\abox_\graph$ given by:

\begin{align*}
\abox_\graph =& \{\rolestyle{Vertex}(\ind{a}, \ind{u}) \suchthat u \in \vertices \} \cup \{ \rolestyle{Edge}(\ind{u_1}, \ind{u_2}) \suchthat (u_1, u_2) \in \edges \} \\
& \cup \{ \rolestyle{Vertex}(\ind{a}, \ind{a_v}), \rolestyle{Edge}(\ind{a_v}, \ind{a_v}), \rolestyle{HasCol}(\ind{a_v}, \ind{r}), \rolestyle{HasCol}(\ind{a_v}, \ind{g}), \rolestyle{HasCol}(\ind{a_v}, \ind{b}) \}
\end{align*}
and the TBox $\tbox := \{ \exists \rolestyle{Vertex^-} \incl \exists \rolestyle{HasCol} \}$, and we denote by $\kb_\graph = (\tbox, \abox_\graph)$ the resulting KB. A part of the canonical model of $\kb_\graph$ is depicted in Figure~\ref{canmodconp}.

\begin{figure}[h]
\centering
\begin{tikzpicture}
[every node/.style={scale=0.8, sloped}]
\node at ( 0, 0) (root)	[label=below:$\ind{a}$]	  		{$\bullet$};
\node at (-2, 1) (u)	[label=above:$\ind{u_1}$]	  		{$\bullet$};
\node at (-2,-1) (v)	[label=below:$\ind{u_2}$]			{$\bullet$};
\node at (-4, 1) (uc) 	[label=above:$\ind{u_1}\rolestyle{HasCol}$]							{$\circ$};
\node at (-4,-1) (vc) 	[label=below:$\ind{u_2}\rolestyle{HasCol}$]							{$\circ$};
\node at ( 4, 1) (r) 	[label=right:$\ind{r}$]		{$\bullet$};
\node at ( 4, 0) (g)   	[label=right:$\ind{g}$]	{$\bullet$};
\node at ( 4,-1) (b)   [label=right:$\ind{b}$]	{$\bullet$};
\node at ( 2, 0) (a)	[label=below:$\ind{a_v}$]			{$\bullet$};
	
\path 
(root) edge [->] node[above] {$\rolestyle{Vertex}$} (u)
(root) edge [->] node[below] {$\rolestyle{Vertex}$} (v)
(root) edge [->] node[above] {$\rolestyle{Vertex}$} (a)

(u) edge [->] node[below] 	{$\rolestyle{Edge}$} (v)
(u) edge [->] node[below] 	{$\rolestyle{HasCol}$} (uc)
(v) edge [->] node[above] 	{$\rolestyle{HasCol}$} (vc)


(a) edge [loop above] node [above]	{$\rolestyle{Edge}$} (a)
(a) edge [->] node [above]	{$\rolestyle{HasCol}$} (r)
(a) edge [->] node [above]	{$\rolestyle{HasCol}$} (g)
(a) edge [->] node [below]	{$\rolestyle{HasCol}$} (b);

\end{tikzpicture}
\caption{A part of $\canmodof{\kb_\graph}$ with $(u_1, u_2) \in \edges$.}
\label{canmodconp}
\end{figure}

We consider the two following rooted CCQs:
\begin{align*}
q^{edge} &= \exists y_c ~ \exists z_1 ~ \exists z_2 ~ \rolestyle{Vertex}(\ind{a}, z_1)  \wedge  \rolestyle{Vertex}(\ind{a}, z_2) \wedge \rolestyle{Edge}(z_1, z_2)  \wedge  \rolestyle{HasCol}(z_1, y_c)  \wedge  \rolestyle{HasCol}(z_2, y_c) \\
q^{col} &= \exists y ~ \exists z  ~ \rolestyle{Vertex}(\ind{a}, y) \wedge \rolestyle{HasCol}(y, z) 
\end{align*}
We let $q$ be the query obtained by taking the conjunction of these two queries and keeping all of the variables existentially quantified. The query $q$ is displayed in Figure \ref{rootedqueryconp}. The three counting variables ($z_1, z_2, z$) are indicated by large gray dots.

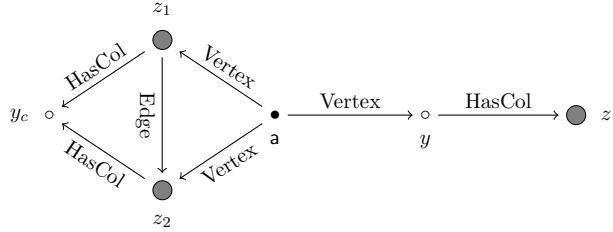
\begin{figure}[h]
\centering
\begin{tikzpicture}
[every node/.style={scale=0.8, sloped}]
\node at ( 0, 3) (r) [label=below:$\ind{a}$]  {$\gx$};
\node at (-1.5, 4) (z1)	[label=above:$z_1$]	{$\gz$};
\node at (-1.5, 2) (z2) 	[label=below:$z_2$]	{$\gz$};
\node at (-3, 3) (yc) 	[label=left:$y_c$]	{$\gy$};
\node at (2, 3) (yv)	[label=below:$y$]  {$\gy$};
\node at ( 4, 3) (zc) 	[label=right:$z$]	{$\gz$};

\path
(r) edge [->] node[above] {$\rolestyle{Vertex}$} 	(z1)
(r) edge [->] node[below] {$\rolestyle{Vertex}$} 	(z2)
(r) edge [->] node[above] {$\rolestyle{Vertex}$} 	(yv)
(yv) edge [->] node[above] {$\rolestyle{HasCol}$} 	(zc)
(z1) edge [->] node[below] {$\rolestyle{Edge}$} 	(z2)
(z1) edge [->] node[above] {$\rolestyle{HasCol}$} 	(yc)
(z2) edge [->] node[below] {$\rolestyle{HasCol}$} 	(yc); 
\end{tikzpicture}
\caption{The rooted CCQ $q$, which is the conjunction of $q^{edge}$ (left part) and $q^{col}$ (right part).}
\label{rootedqueryconp}
\end{figure}

\noindent It is not hard to see that $(\ta_\emptyset, [3, +\infty]) \in \qIntervals{q}{\kb_\graph}$. 
Indeed, there are at least $9$ matches of $q$ in any model $\I$ of $\kb$, given by:
$$
z_1, z_2, y  \mapsto  \ind{a_v} \qquad
y_c  \mapsto  \ind{r} \,|\, \ind{g} \,|\, \ind{b}\qquad
z  \mapsto  \ind{r} \,| \ind{g} \,|\, \ind{b}
$$
These $9$ matches give rise to $3$ c-matches for $q$, corresponding to the three ways of mapping counting variable $z$. 
To complete the proof, we establish the following claim. \medskip 

\noindent \textbf{Claim.} $(\ta_\emptyset, [4, +\infty]) \in \qIntervals{q}{\kb_\graph} \Longleftrightarrow \graph \notin \tcol$.

\paragraph{$(\Rightarrow)$}
Assume $(\ta_\emptyset, [4, +\infty]) \in \qIntervals{q}{\kb_\graph}$, and take some possible coloring $\tau : \vertices \rightarrow \{ \ind{r}, \ind{g}, \ind{b} \}$ of the graph $\graph$. Let $\I^\graph_\tau$ be the model of $\kb_\graph$ whose domain is $\ainds(\Amc_\graph)$ and which interprets roles $\rolestyle{Vertex}$ and $\rolestyle{Edge}$ exactly following the ABox, and which interprets $\rolestyle{HasCol}$ according to $\tau$:
 $$\rolestyle{HasCol}^{\I^\graph_\tau}= \{(\ind{a_v}, \ind{r}), (\ind{a_v}, \ind{g}), (\ind{a_v}, \ind{b}) \} \cup \{(\ind{v},\tau(v)) \mid v \in \vertices \}$$
Intuitively, $\I_\tau$ is obtained from the canonical model by replacing the element $\ind{v}\rolestyle{HasCol}$ with $\tau(v)$.
 
By hypothesis, there is a fourth c-match $\match$ for $q$ in $\I^\graph_\tau$. It is easily verified that the additional match can only result from the atom $\rolestyle{Edge}(z_1, z_2)$ being mapped onto an edge $\rolestyle{Edge}(\ind{u_1}, \ind{u_2})$ that is different from $\rolestyle{Edge}(\ind{a_v}, \ind{a_v})$. From the definition of $\I^\graph_\tau$, this implies that the edge $(u_1, u_2)$ of $\graph$ is monochromatic, both vertices sharing the color $\match(y_c)$. Thus, $\tau$ is not a $3$-coloring. As this construction holds for any possible coloring $\tau$, we obtain $\graph \notin \tcol$.
 
\paragraph{$(\Leftarrow)$}
Assume $\graph \notin \tcol$, and take some model $\I$ of $\kb_\graph$. By Lemma \ref{cm-lemma}, there is a homomorphism $f : \canmodof{\kb_\graph} \rightarrow \I$ (which preserves individual names). Define $\tau: \vertices \rightarrow \domain{\I}$ as follows: $\tau(u)= f(\ind{u}\rolestyle{HasCol})$. There are two cases to consider:
\begin{itemize}
\item If there exists $u \in \vertices$ such that $\tau(u) \notin \{ \ind{r}, \ind{g}, \ind{b} \}$, then this provides a match of $q$ in $\I$ given by $z \mapsto \tau(u)$ and $y \mapsto \ind{u}^\I$, whose restriction to the counting variables is a new c-match.
\item Else, since $\graph \notin \tcol$, there exists an edge $(u_1, u_2) \in \edges$ such that $\tau(u_1) = \tau(u_2)$. It provides a new match given by: 
 $$
  z  \mapsto  \ind{r} \quad
 y \mapsto  \ind{a_v} \quad
 z_1 \mapsto  \ind{u_1} \quad
 z_2  \mapsto  \ind{u_2} \quad
  y_c  \mapsto  \tau(u_1) ~ ( = \tau(u_2)) 
 $$
\end{itemize}
%
In both cases, there is a fourth c-match for $q$. We thus obtain 
 $(\ta_\emptyset, [4, +\infty]) \in \qIntervals{q}{\kb_\graph}$.
\end{proof}

\begin{theoremrep}
\label{rootedconexphard}
In \dlh, rooted CCQ answering is $\coNEXP$-hard w.r.t.\ combined complexity.
\end{theoremrep}
\begin{proofsketch} The proof adapts a reduction from the exponential grid tiling problem (Lemma 18 from \cite{kostylevreutter:count}),
the key difference being the use 
of existential query variables to access (and count) the colors and bits. 
%
%
 \end{proofsketch}
 
\begin{proof}
\newcommand{\Ho}{\mathcal{H}} 
\newcommand{\V}{\mathcal{V}} 
\newcommand{\C}{\mathcal{C}} 
\newcommand{\Dir}{\mathcal{D}} 
The proof is by reduction from the exponential grid tiling problem ($\tiling$). We recall that an instance of this problem consists of a set $\C$ of colors, two relations $\Ho, \V \subseteq \C \times \C$ that give the horizontal and vertical tiling conditions, and a number $n$. The task is to decide whether there exists a valid $(\Ho, \V)$-tiling of an $2^n \times 2^n$ grid, i.e., a mapping $\tau:\{0, \dots, 2^n -1 \} \times \{0, \dots, 2^n -1 \} \mapsto \C$ such that  $(\tau(i,j), \tau(i+1,j)) \in \Ho$ for every $0 \leq i < 2^n-1$ and $(\tau(i,j), \tau(i,j+1)) \in \V$ for every $0 \leq j < 2^n-1$. 
%
%
%
%
In what follows, we consider an instance $(n, \C, \Ho, \V)$ be an instance of the $\tiling$ problem. 

To be able to test for the existence of a tiling of a $2^n \times 2^n$ grid, we must start by ensuring we can find such a grid in each model. Furthermore, we will need to detect horizontal and vertical adjacency in this grid, it is thus appropriate to use horizontal/vertical coordinates. To ensure a polynomial reduction, we need to use a binary encoding of these coordinates.
We start from a root $\ind{a}$ and an initial element $\ind{b}$ and use TBox axioms to build two witnesses to represent the two possible values for the $n^{th}$ bit of the horizontal coordinates:
$$\rolestyle{Roots}(\ind{a}, \ind{b}) \qquad \exists \rolestyle{Roots}^- \incl \exists \rolestyle{H}^n_0 \qquad \exists \rolestyle{Roots}^- \incl \exists \rolestyle{H}^n_1$$
We use further axioms to generate all possible horizontal coordinates, and we proceed similarly with the vertical coordinates, until we generate all possible pairs of coordinates. Concretely, we include the following axioms: 
$$\exists (\rolestyle{H}^i_b)^- \incl \exists \rolestyle{H}^{i-1}_{b'} \qquad \exists (\rolestyle{H}^1_b)^- \incl \exists \rolestyle{V}^n_{b'} \qquad \exists (\rolestyle{V}^i_b)^- \incl \exists \rolestyle{V}^{i-1}_{b'} \qquad \qquad \text{ for all } b,b' \in \{0,1\}, 1 <i \leq n$$
The preceding axioms will generate a binary tree of height $2n$ in the canonical model, whose leaves represent all possible grid positions. We use the following two axioms assign a color to each of the points representing a grid position:
$$\exists (\rolestyle{V}^1_0)^- \incl \exists \rolestyle{HasCol} \qquad \exists (\rolestyle{V}^1_1)^- \incl \exists \rolestyle{HasCol}$$
To help us compare positions, 
we will include the following TBox axioms, for all $b \in \{0,1\}$ and $1 \leq i \leq n$:
$$\exists (\rolestyle{H}_b^i)^- \incl \exists \rolestyle{HasBit}_b \qquad \exists (\rolestyle{V}_b^i)^- \incl \exists \rolestyle{HasBit}_b$$
We will also introduce a general role ($\rolestyle{HV}$) to more compactly navigate the tree:
$$\rolestyle{H}^i_b \incl \rolestyle{HV} \qquad \rolestyle{V}^i_b \incl \rolestyle{HV} \qquad  (b \in \{ 0, 1 \}, 1 \leq i \leq n)$$
This completes our description of the TBox.
We will finish our description of the ABox later in the proof, but it will be useful to know that it will contain an ABox individual $\ind{c}$ for every color $c \in \C$ and two ABox individuals ($\ind{one}, \ind{zero}$) to represent bits. 






Let us now define the query $q$. To keep track of the colors used in a candidate tiling, we will use the following subquery:
$$
q^{col} = \exists y_0^{col} \dots \exists y^{col}_{2n} \exists z^{col}\ \rolestyle{Roots}(\ind{a}, y_0^{col}) \wedge \bigwedge_{i = 0}^{2n-1} \rolestyle{HV}(y_i^{col}, y_{i+1}^{col}) \wedge \ \rolestyle{HasCol}(y_{2n}^{col}, z^{col})
$$
Observe that $z^{col}$ is the only counting variable. 
%
We also need to be able to detect if other bits than the intended ones ($\ind{one}, \ind{zero}$) are being used to satisfy the axioms $\rolestyle{H}_b^- \incl  \exists \rolestyle{HasBit}_b$ and $\rolestyle{V}_b^- \incl  \exists \rolestyle{HasBit}_b$. For this purpose, we will introduce the two following subqueries:
\begin{align*}
q^{0} = &\exists y_0^{0} \dots \exists y^{0}_{2n} \exists z^{0}\ \rolestyle{Roots}(\ind{a}, y_0^{0}) \wedge \bigwedge_{i = 0}^{2n-1} \rolestyle{HV}(y_i^{0}, y_{i+1}^{0}) \wedge \ \rolestyle{HasBit}_0(y_{2n}^{0}, z^{0}) \\
q^{1} = &\exists y_0^{1} \dots \exists y^{1}_{2n} \exists z^{1}\ \rolestyle{Roots}(\ind{a}, y_0^{1}) \wedge \bigwedge_{i = 0}^{2n-1} \rolestyle{HV}(y_i^{1}, y_{i+1}^{1}) \wedge \ \rolestyle{HasBit}_1(y_{2n}^{1}, z^{1})
\end{align*}
We note that each of the preceding queries has a single counting variable ($z^{0}$ or $z^{1}$). The axioms for $\rolestyle{HV}$ together with the construction of the ABox will ensure that every element used as a bit (i.e., in the second argument of $\rolestyle{HasBit}$) gives rise to a c-match for one of these two queries. 

We next discuss the parts of the query that are used to check the tiling conditions. To detect adjacency, we remark that two
grid positions $(h_1, v_1), (h_2, v_2) \in \{0, \dots, 2^n -1 \} \times \{0, \dots, 2^n -1 \}$ are vertically adjacent iff:
\begin{itemize}
\item $h_1 = h_2$, so the binary encodings of $h_1$ and $h_2$ are the same;
\item $v_2 = v_1 + 1$, so the binary encodings of $v_2$ and $v_1$ are the same until, at some point, $v_2$ ends with $1\cdot 0^k$ while $v_1$ ends with $0 \cdot 1^k$.
\end{itemize}
To detect a violation of the vertical tiling condition (i.e.\ two vertically adjacent tiles with colors $c$ and $c'$ such that $(c, c') \notin \V$), we need $n$ queries, one for each possible position where the bit from the vertical coordinates differ. 
For each $1 \leq k \leq n$, we create a subquery $q^{\V, (c, c'), k}$ defined as follows. Note that the variables in $q^{\V, (c, c'), k}$ all have the superscript $\cdot^{\V, (c, c'), k}$, which means they do not occur in any other subquery, but these superscripts are omitted in the definition for the sake of readability.
\begin{align*}
q^{\V, (c, c'), k} = & \, \exists z \exists y_{l, 1} \dots \exists y_{l, 2n} \exists y_{r, 1} \dots \exists y_{r, 2n}  \exists y_{s, 1} \dots \exists y_{s, n+k} \\
&\,  \rolestyle{Roots}(\ind{a}, z) \wedge \rolestyle{HV}(z, y_{l, 1}) \wedge \rolestyle{HV}(z, y_{r, 1}) \wedge \left( \bigwedge_{i = 1}^{2n-1}  \rolestyle{HV}(y_{l, i}, y_{l, i+1}) \wedge \rolestyle{HV}(y_{r, i}, y_{r, i+1}) \right)   \\
& \, \wedge \rolestyle{HasCol}(y_{l, 2n}, \ind{c}) \wedge \rolestyle{HasCol}(y_{r, 2n}, \ind{c'}) \wedge \left( \bigwedge_{i = 1}^{n+k-1}  \rolestyle{HasBit}(y_{l, i}, y_{s, i}) \wedge \rolestyle{HasBit}(y_{r, i}, y_{s, i}) \right) \\
&\,  \wedge \rolestyle{HasBit}(y_{l, n+k}, \ind{zero}) \wedge \rolestyle{HasBit}(y_{r, n+k}, \ind{one}) \wedge \left( \bigwedge_{i = n+k+1}^{2n}  \rolestyle{HasBit}(y_{l, i}, \ind{one}) \wedge \rolestyle{HasBit}(y_{r, i}, \ind{zero}) \right) \\
\end{align*}
Note that $z$ is the only counting variable of $q^{\V, (c, c'), k}$. We can similarly define a set of subqueries $q^{\Ho, (c, c'), k}$ ($1 \leq k \leq n$) that detect violations of the horizontal tiling conditions. 

Finally, we let $q$ be the conjunction of the all of the preceding subqueries. It is displayed in Figure~\ref{conexp_rooted_query}. The set of counting variables of $q$ is the union of the counting variables of its subqueries. We observe that $q$ is rooted, as it has a single connected component which contains the individual $\ind{a}$. 

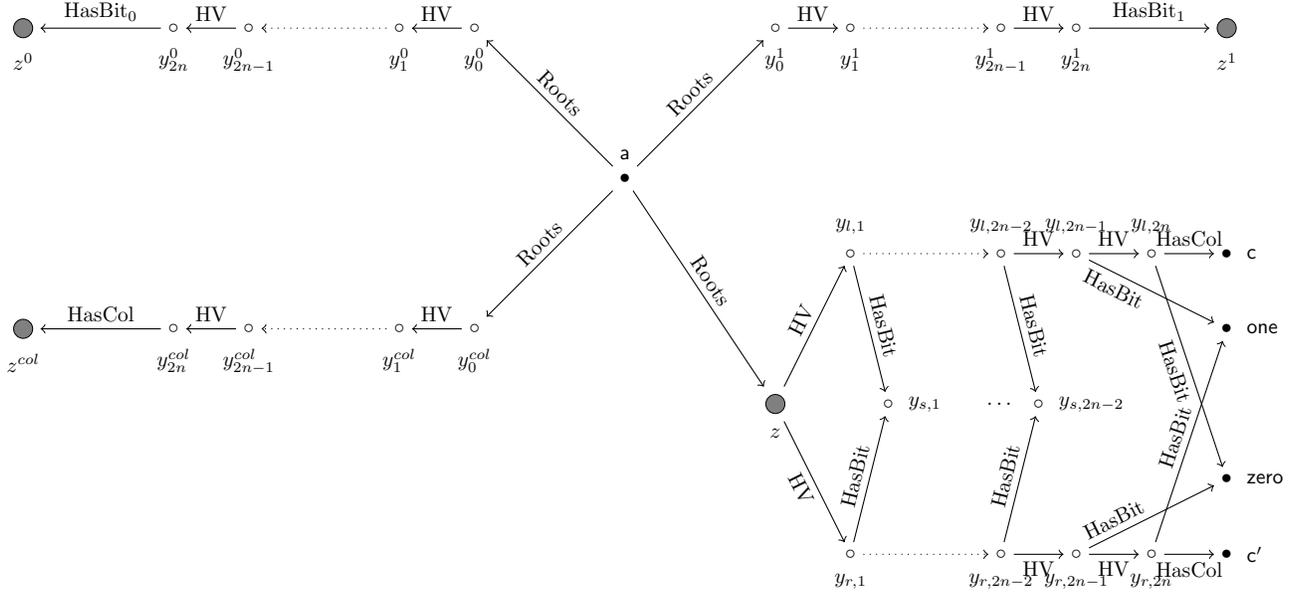
\begin{figure}[!h]
\centering
\begin{tikzpicture}
\tikzset{every node/.style={scale=.8, sloped}}
\node at ( 0, 0) (a) [label=above:$\ind{a}$] {$\gx$};

\node at (-2, 2) (y00) [label=below:$y^0_0$] {$\gy$};
\node at (-3, 2) (y01) [label=below:$y^0_1$] {$\gy$};
\node at (-5, 2) (y0n) [label=below:$y^0_{2n-1}$] {$\gy$};
\node at (-6, 2) (y0nn) [label=below:$y^0_{2n}$] {$\gy$};
\node at (-8, 2) (z0) [label=below:$z^0$] {$\gz$};

\node at ( 2, 2) (y10) [label=below:$y^1_0$] {$\gy$};
\node at ( 3, 2) (y11) [label=below:$y^1_1$] {$\gy$};
\node at ( 5, 2) (y1n) [label=below:$y^1_{2n-1}$] {$\gy$};
\node at ( 6, 2) (y1nn) [label=below:$y^1_{2n}$] {$\gy$};
\node at ( 8, 2) (z1) [label=below:$z^1$] {$\gz$};

\node at (-2,-2) (yc0) [label=below:$y^{col}_0$] {$\gy$};
\node at (-3,-2) (yc1) [label=below:$y^{col}_1$] {$\gy$};
\node at (-5,-2) (ycn) [label=below:$y^{col}_{2n-1}$] {$\gy$};
\node at (-6,-2) (ycnn) [label=below:$y^{col}_{2n}$] {$\gy$};
\node at (-8,-2) (zc) [label=below:$z^{col}$] {$\gz$};

\node at ( 2, -3) (z) [label=below:$z$] {$\gz$};

\node at ( 3,-1) (yl1) [label=above:$y_{l,1}$] {$\gy$};
\node at ( 5,-1) (yln) [label=above:$y_{l, 2n-2}$] {$\gy$};
\node at ( 6,-1) (ylnn) [label=above:$y_{l, 2n-1}$] {$\gy$};
\node at ( 7,-1) (ylnnn) [label=above:$y_{l, 2n}$] {$\gy$};

\node at ( 3,-5) (yr1) [label=below:$y_{r, 1}$] {$\gy$};
\node at ( 5,-5) (yrn) [label=below:$y_{r, 2n-2}$] {$\gy$};
\node at ( 6,-5) (yrnn) [label=below:$y_{r, 2n-1}$] {$\gy$};
\node at ( 7,-5) (yrnnn) [label=below:$y_{r, 2n}$] {$\gy$};

\node at ( 8,-1) (c) [label=right:$\ind{c}$] {$\gx$};
\node at ( 8,-5) (cc) [label=right:$\ind{c'}$] {$\gx$};

\node at ( 3.5,-3) (y1) [label=right:$y_{s, 1}$] {$\gy$};
\node at ( 5,-3) (yi) [] {$\dots$};
\node at ( 5.5,-3) (yn) [label=right:$y_{s, 2n-2}$] {$\gy$};
\node at ( 8,-2) (o) [label=right:$\ind{one}$] {$\gx$};
\node at ( 8,-4) (i) [label=right:$\ind{zero}$] {$\gx$};

\path[every node/.style={scale=.8, sloped}]
(a) edge [->] node[above] {$\rolestyle{Roots}$} (y00)
(y00) edge [->] node[above] {$\rolestyle{HV}$} (y01)
(y01) edge [->, dotted] (y0n)
(y0n) edge [->] node[above] {$\rolestyle{HV}$} (y0nn)
(y0nn) edge [->] node[above] {$\rolestyle{HasBit_0}$} (z0)

(a) edge [->] node[above] {$\rolestyle{Roots}$} (y10)
(y10) edge [->] node[above] {$\rolestyle{HV}$} (y11)
(y11) edge [->, dotted] (y1n)
(y1n) edge [->] node[above] {$\rolestyle{HV}$} (y1nn)
(y1nn) edge [->] node[above] {$\rolestyle{HasBit_1}$} (z1)

(a) edge [->] node[above] {$\rolestyle{Roots}$} (yc0)
(yc0) edge [->] node[above] {$\rolestyle{HV}$} (yc1)
(yc1) edge [->, dotted] (ycn)
(ycn) edge [->] node[above] {$\rolestyle{HV}$} (ycnn)
(ycnn) edge [->] node[above] {$\rolestyle{HasCol}$} (zc)

(a) edge [->] node[above] {$\rolestyle{Roots}$} (z)

(z) edge [->] node[above] {$\rolestyle{HV}$} (yl1)
(yl1) edge [->, dotted] (yln)
(yln) edge [->] node[above] {$\rolestyle{HV}$} (ylnn)
(ylnn) edge [->] node[above] {$\rolestyle{HV}$} (ylnnn)
(ylnnn) edge [->] node[above] {$\rolestyle{HasCol}$} (c)

(yl1) edge [->] node[above] {$\rolestyle{HasBit}$} (y1)
(yln) edge [->] node[above] {$\rolestyle{HasBit}$} (yn)
(ylnn) edge [->] node[below, near start] {$\rolestyle{HasBit}$} (o)
(ylnnn) edge [->] node[below] {$\rolestyle{HasBit}$} (i)

(z) edge [->] node[below] {$\rolestyle{HV}$} (yr1)
(yr1) edge [->, dotted] (yrn)
(yrn) edge [->] node[below] {$\rolestyle{HV}$} (yrnn)
(yrnn) edge [->] node[below] {$\rolestyle{HV}$} (yrnnn)
(yrnnn) edge [->] node[below] {$\rolestyle{HasCol}$} (cc)

(yr1) edge [->] node[above] {$\rolestyle{HasBit}$} (y1)
(yrn) edge [->] node[above] {$\rolestyle{HasBit}$} (yn)
(yrnn) edge [->] node[above, near start] {$\rolestyle{HasBit}$} (i)
(yrnnn) edge [->] node[above] {$\rolestyle{HasBit}$} (o)

;
\end{tikzpicture}
\caption{A part of the rooted query $q$, being the conjunction of $q^0$ (above left), $q^1$ (above right), $q^{col}$ (below left), several $q^{\V, (c, c'), k}$ (below right, only one is depicted with omitted superscripts), and several $q^{\Ho, (c, c'), k}$ (none is depicted).}
\label{conexp_rooted_query}
\end{figure}


We can now define the ABox, which introduces individuals for the intended colors and bits and a further individual $\ind{d}$ that serves to ensure that all parts of the query can be matched: 
\begin{align*}
\abox =
& \,\{ \rolestyle{Roots}(\ind{a}, \ind{b}), \rolestyle{Roots}(\ind{a}, \ind{d}), \rolestyle{HV}(\ind{b}, \ind{b}), \rolestyle{HasBit}_0(\ind{d}, \ind{zero}), \rolestyle{HasBit}_1(\ind{d}, \ind{one}) \} \\
&\, \cup  \{ \rolestyle{H}_b^k(\ind{d}, \ind{d}) \suchthat b \in \{ 0, 1 \}, k = 1, \dots n \} \cup  \{ \rolestyle{V}_b^k(\ind{d}, \ind{d}) \suchthat b \in \{ 0, 1 \}, k = 1, \dots n \} \\
&\, \cup \{ \rolestyle{HasCol}(\ind{d}, \ind{c}) \suchthat c \in \C\}.
\end{align*}

\begin{figure}[!h]
\centering
\begin{tikzpicture}
\tikzset{every node/.style={scale=.8, sloped}}
\node at ( 1, 8) (a) [label=above:$\ind{a}$] {$\bullet$};
\node at (-8, 8) (b) [label=below:$\ind{b}$] {$\bullet$};
\node at (-9, 7) (ho) [] {$\circ$};
\node at (-7, 7) (hi) [] {$\circ$};
\node at (-9.5, 6) (hoo) [] {};
\node at (-8.5, 6) (hoi) [] {};
\node at (-7.5, 6) (hio) [] {};
\node at (-6.5, 6) (hii) [] {};
\node at (-6, 5) (hn) [] {$\circ$};
\node at (-7, 4) (hno) [] {$\circ$};
\node at (-5, 4) (hni) [] {$\circ$};
\node at (-5.5, 3) (vo) [] {$\circ$};
\node at (-4.5, 3) (vi) [] {$\circ$};
\node at (-5.75, 2) (voo) [] {};
\node at (-5.25, 2) (voi) [] {};
\node at (-4.75, 2) (vio) [] {};
\node at (-4.25, 2) (vii) [] {};
\node at (-4, 1) (vn) [] {$\circ$};
\node at (-5, 0) (vno) [] {$\circ$};
\node at (-3, 0) (vni) [] {$\circ$};

\node at (-10,6.5) (hoc) [] {$\circ$};
\node at (-8, 3.5) (hnoc) [] {$\circ$};
\node at (-6.5, 2.5) (voc) [] {$\circ$};
\node at (-6,-0.5) (vnoc) [] {$\circ$};

\node at (-6, 6.5) (hic) [] {$\circ$};
\node at (-4, 3.5) (hnic) [] {$\circ$};
\node at (-3.5, 2.5) (vic) [] {$\circ$};
\node at (-2,-0.5) (vnic) [] {$\circ$};

\node at (-5,-1) (vnot) [] {$\circ$};
\node at (-3,-1) (vnit) [] {$\circ$};

\node at (-12, 3) (o) [label=left:$\ind{zero}$] {$\bullet$};
\node at (-1, 3) (i) [label=right:$\ind{one}$] {$\bullet$};
\node at ( 1, 0) (d) [label=left:$\ind{d}$] {$\bullet$};
\node at (-6,-3) (c1) [label=below:$\ind{c_1}$] {$\bullet$};
\node at (-4, -3) (ci) [] {$\dots$};
\node at ( -2,-3) (cp) [label=below:$\ind{c_p}$] {$\bullet$};

\path[every node/.style={scale=.8, sloped}]
(a) edge [->] node[above] {$\rolestyle{Roots}$} (b)
(a) edge [->] node[above] {$\rolestyle{Roots}$} (d)
(b) edge [loop above] node[above] {$\rolestyle{HV}$} (b)

(b) edge [->] node[above] {$\rolestyle{H^n_0}$}		(ho)
(b) edge [->] node[above] {$\rolestyle{H^n_1}$} 	(hi)
(ho) edge [->,dotted] node[above] {}	(hoo)
(ho) edge [->,dotted] node[above] {} 	(hoi)
(hi) edge [->,dotted] node[above] {}	(hio)
(hi) edge [->,dotted] node[above] {} 	(hii)
(hn) edge [->] node[above] {$\rolestyle{H^0_0}$}	(hno)
(hn) edge [->] node[above] {$\rolestyle{H^0_1}$} 	(hni)
(hno) edge [->,dotted] node[above] {}	(-7.5, 3)
(hno) edge [->,dotted] node[above] {} 	(-6.5, 3)
(hni) edge [->] node[above, near end] {$\rolestyle{V^n_0}$}	(vo)
(hni) edge [->] node[above, near end] {$\rolestyle{V^n_1}$} 	(vi)
(vo) edge [->,dotted] node[above] {}	(voo)
(vo) edge [->,dotted] node[above] {} 	(voi)
(vi) edge [->,dotted] node[above] {}	(vio)
(vi) edge [->,dotted] node[above] {} 	(vii)
(vn) edge [->] node[above] {$\rolestyle{V^0_0}$}	(vno)
(vn) edge [->] node[above] {$\rolestyle{V^0_1}$} 	(vni)

(hii) edge [->,dotted] node[above] {} 	(hn)
(vii) edge [->,dotted] node[above] {} 	(vn)

(vni) edge[->] node[below] {$\rolestyle{HasCol}$} (vnit)
(vno) edge[->] node[above] {$\rolestyle{HasCol}$} (vnot)

(d) edge [loop right] node[above] {$\rolestyle{H_b^k}, \rolestyle{V_b^k}$} (d)

(d) edge[->] node[above, near end] {$\rolestyle{HasCol}$} (c1)
(d) edge[->] node[below, near end] {$\rolestyle{HasCol}$} (cp)

(d) edge [->] node[above] {$\rolestyle{HasBit_1}$}		(i)

(ho) edge [->] node[above] {$\rolestyle{HasBit_0}$}		(hoc)
(hno) edge [->] node[above] {$\rolestyle{HasBit_0}$}		(hnoc)
(vo) edge [->] node[above] {$\rolestyle{HasBit_0}$}		(voc)
(vno) edge [->] node[above] {$\rolestyle{HasBit_0}$}		(vnoc)

(hi) edge [->] node[above] {$\rolestyle{HasBit_1}$}		(hic)
(hni) edge [->] node[above] {$\rolestyle{HasBit_1}$}		(hnic)
(vi) edge [->] node[above] {$\rolestyle{HasBit_1}$}		(vic)
(vni) edge [->] node[above] {$\rolestyle{HasBit_1}$}		(vnic)
;

\draw[rounded corners, ->] (d) -- ( 1, -4) -- node [below] {$\rolestyle{HasBit_0}$} ( -12 ,-4)  -- (o) ;

\end{tikzpicture}
\caption{A part of the canonical model $\canmod$.}
\label{conexp_rooted_canmod}
\end{figure}
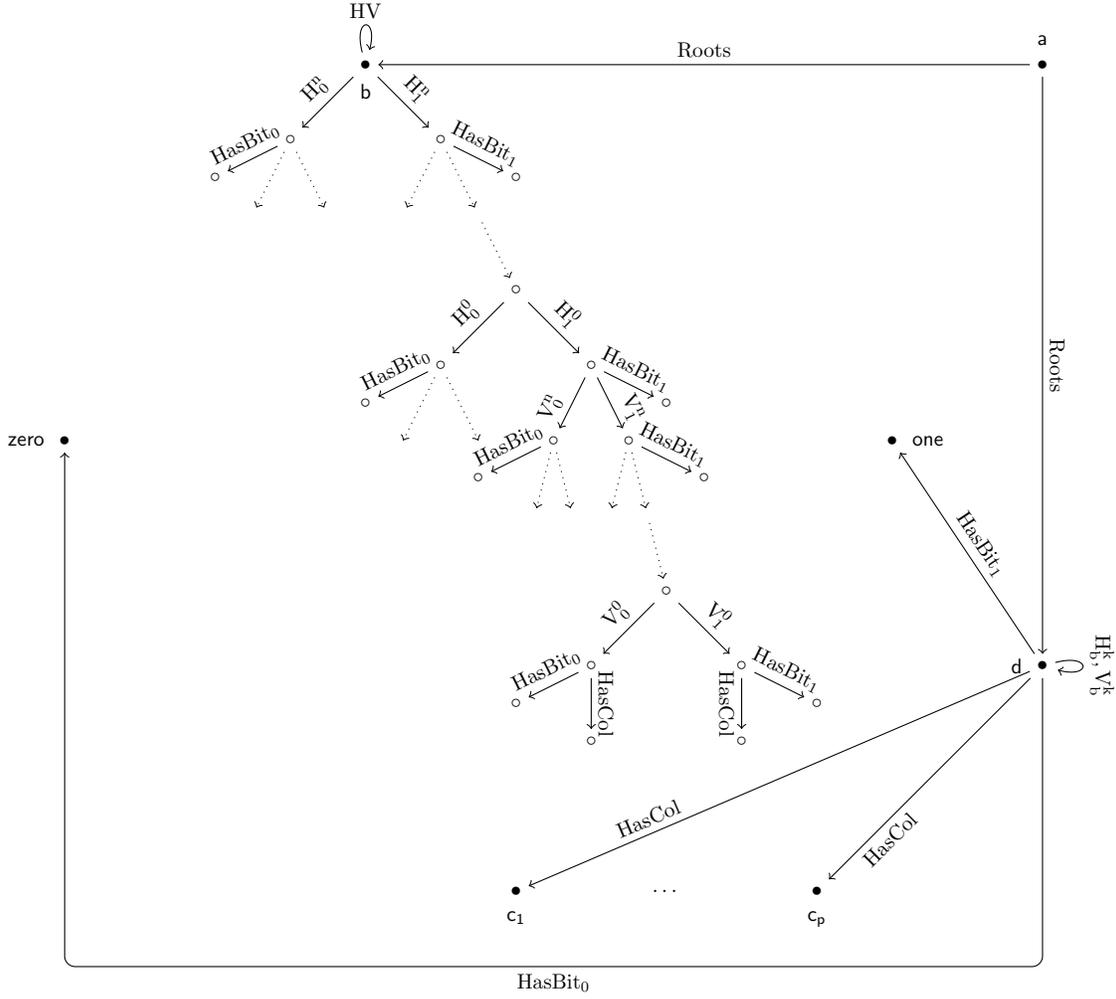

Let $p =  |\C|$, and let $\Kmc$ be the KB with the preceding TBox and ABox. A part of the canonical model $\canmod$ is displayed in Figure~\ref{conexp_rooted_canmod}. To complete the proof, it suffices to establish the following claim:

\paragraph*{Claim} $(\emptyset, [p+1, +\infty]) \in \qIntervals{q}{\kb} \Longleftrightarrow (n, \C, \Ho, \V) \notin \tiling.$\smallskip\\
The proof of this claim is similar in spirit to the proof of Theorem~\ref{rootedconphard}. First observe that there are always at least $p$ c-matches given by mapping the counting variables as follows:
$$ z^{col} \mapsto \ind{c}_1 \suchthat \dots \suchthat \ind{c}_p \qquad z^0 \mapsto \ind{zero} \qquad z^1 \mapsto \ind{one} \qquad z^{\Ho, (c, c'), k}, z^{\V, (c, c'), k}  \mapsto \ind{d}$$
and mapping all of the existential variables to $\ind{d}$. 
 
\paragraph{$(\Rightarrow)$}
Assume $(\emptyset, [p+1, +\infty]) \in \qIntervals{q}{\kb}$, and take some potential tiling $\tau : \{0, \dots 2^n-1 \} \times \{0, \dots 2^n-1 \} \rightarrow \{ \ind{c} \suchthat c \in \C \}$. Let $\I_\tau$ be the model of $\kb$ that is obtained from $\canmod$  as follows:
\begin{itemize}
\item $\Delta^{\Imc_\tau}$ contains all elements from $\domain{\canmod}$ except those anonymous elements whose last symbol is
$\rolestyle{HasCol}$, $\rolestyle{HasBit}_0$, or $\rolestyle{HasBit}_1$ (i.e.\ witnesses for axioms involving $\exists \rolestyle{HasCol}$, $\exists \rolestyle{HasBit}_0$, or $\exists \rolestyle{HasBit}_1$);
 \item the roles  $\rolestyle{HasCol}$, $\rolestyle{HasBit}_0$, $\rolestyle{HasBit}_1$ are interpreted as follows:
\begin{align*}
\rolestyle{HasBit}_0^{\I_\tau} := & \{ (\ind{d}, \ind{zero}) \} \cup \{ (\ind{b} \rolestyle{H}^n_{h_n} \dots \rolestyle{H}^k_{h_k} \rolestyle{H}^{k-1}_0, \ind{zero}) \suchthat h_n, \dots h_k \in \{ 0, 1 \}, k = 1, \dots n+1 \}  \\
& \cup \{ (\ind{b} \rolestyle{H}^n_{h_n} \dots \rolestyle{H}^1_{h_1} \rolestyle{V}^n_{v_n} \dots \rolestyle{V}^k_{v_k} \rolestyle{V}^{k-1}_0, \ind{zero}) \suchthat h_n, \dots h_1, v_n, \dots v_k \in \{ 0, 1 \}, k = 0, \dots n+1 \} \\
\rolestyle{HasBit}_1^{\I_\tau} := & \{ (\ind{d}, \ind{one}) \} \cup \{ (\ind{b} \rolestyle{H}^n_{h_n} \dots \rolestyle{H}^k_{h_k} \rolestyle{H}^{k-1}_1, \ind{one}) \suchthat h_n, \dots h_k \in \{ 0, 1 \}, k = 1, \dots n+1 \}  \\
& \cup \{ (\ind{b} \rolestyle{H}^n_{h_n} \dots \rolestyle{H}^1_{h_1} \rolestyle{V}^n_{v_n} \dots \rolestyle{V}^k_{v_k} \rolestyle{V}^{k-1}_1 , \ind{one}) \suchthat h_n, \dots h_1, v_n, \dots v_k \in \{ 0, 1 \}, k = 0, \dots n+1 \} \\
\rolestyle{HasCol}^{\I_\tau} := & \{ (\ind{d}, \ind{c}_k) \suchthat k = 1, \dots p \} \\
& \cup  \{ (\ind{b} \rolestyle{H}^n_{h_n} \dots \rolestyle{H}^1_{h_1} \rolestyle{V}^n_{v_n} \dots \rolestyle{V}^1_{v_1}, \tau(h_n \dots h_1, v_n \dots v_1) \suchthat h_n, \dots h_1, v_n, \dots v_1 \in \{ 0, 1 \} \}
\end{align*}
where by a slight abuse of notation, we use $\tau(h_n \dots h_1, v_n \dots v_1)$ to mean $\tau(h,v)$, with $h$ and $v$ the numbers whose binary encodings are $h_n \dots h_1$ and $v_n \dots v_1$ respectively;
\item the remaining roles are interpreted exactly as in $\canmod$.
\end{itemize}

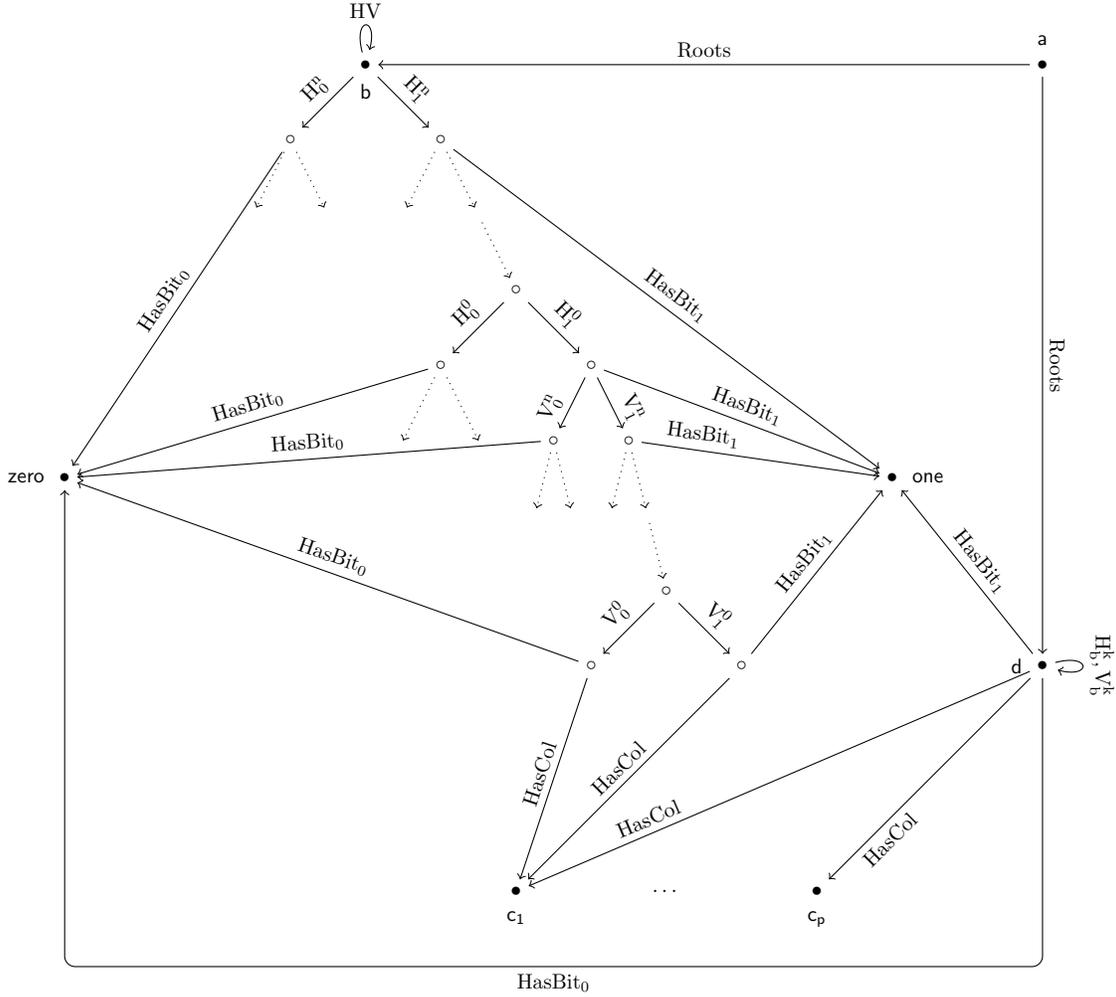
\begin{figure}[!h]
\centering
\begin{tikzpicture}
\tikzset{every node/.style={scale=.8, sloped}}
\node at ( 1, 8) (a) [label=above:$\ind{a}$] {$\bullet$};
\node at (-8, 8) (b) [label=below:$\ind{b}$] {$\bullet$};
\node at (-9, 7) (ho) [] {$\circ$};
\node at (-7, 7) (hi) [] {$\circ$};
\node at (-9.5, 6) (hoo) [] {};
\node at (-8.5, 6) (hoi) [] {};
\node at (-7.5, 6) (hio) [] {};
\node at (-6.5, 6) (hii) [] {};
\node at (-6, 5) (hn) [] {$\circ$};
\node at (-7, 4) (hno) [] {$\circ$};
\node at (-5, 4) (hni) [] {$\circ$};
\node at (-5.5, 3) (vo) [] {$\circ$};
\node at (-4.5, 3) (vi) [] {$\circ$};
\node at (-5.75, 2) (voo) [] {};
\node at (-5.25, 2) (voi) [] {};
\node at (-4.75, 2) (vio) [] {};
\node at (-4.25, 2) (vii) [] {};
\node at (-4, 1) (vn) [] {$\circ$};
\node at (-5, 0) (vno) [] {$\circ$};
\node at (-3, 0) (vni) [] {$\circ$};

\node at (-12, 2.5) (o) [label=left:$\ind{zero}$] {$\bullet$};
\node at (-1, 2.5) (i) [label=right:$\ind{one}$] {$\bullet$};

\node at ( 1, 0) (d) [label=left:$\ind{d}$] {$\bullet$};
\node at (-6,-3) (c1) [label=below:$\ind{c_1}$] {$\bullet$};
\node at (-4, -3) (ci) [] {$\dots$};
\node at ( -2,-3) (cp) [label=below:$\ind{c_p}$] {$\bullet$};

\path[every node/.style={scale=.8, sloped}]
(a) edge [->] node[above] {$\rolestyle{Roots}$} (b)
(a) edge [->] node[above] {$\rolestyle{Roots}$} (d)
(b) edge [loop above] node[above] {$\rolestyle{HV}$} (b)

(b) edge [->] node[above] {$\rolestyle{H^n_0}$}		(ho)
(b) edge [->] node[above] {$\rolestyle{H^n_1}$} 	(hi)
(ho) edge [->,dotted] node[above] {}	(hoo)
(ho) edge [->,dotted] node[above] {} 	(hoi)
(hi) edge [->,dotted] node[above] {}	(hio)
(hi) edge [->,dotted] node[above] {} 	(hii)
(hn) edge [->] node[above] {$\rolestyle{H^0_0}$}	(hno)
(hn) edge [->] node[above] {$\rolestyle{H^0_1}$} 	(hni)
(hno) edge [->,dotted] node[above] {}	(-7.5, 3)
(hno) edge [->,dotted] node[above] {} 	(-6.5, 3)
(hni) edge [->] node[above, near end] {$\rolestyle{V^n_0}$}	(vo)
(hni) edge [->] node[above, near end] {$\rolestyle{V^n_1}$} 	(vi)
(vo) edge [->,dotted] node[above] {}	(voo)
(vo) edge [->,dotted] node[above] {} 	(voi)
(vi) edge [->,dotted] node[above] {}	(vio)
(vi) edge [->,dotted] node[above] {} 	(vii)
(vn) edge [->] node[above] {$\rolestyle{V^0_0}$}	(vno)
(vn) edge [->] node[above] {$\rolestyle{V^0_1}$} 	(vni)

(hii) edge [->,dotted] node[above] {} 	(hn)
(vii) edge [->,dotted] node[above] {} 	(vn)

(vni) edge[->] node[above] {$\rolestyle{HasCol}$} (c1)
(vno) edge[->] node[above] {$\rolestyle{HasCol}$} (c1)

(d) edge [loop right] node[above] {$\rolestyle{H_b^k}, \rolestyle{V_b^k}$} (d)

(d) edge[->] node[above, near end] {$\rolestyle{HasCol}$} (c1)
(d) edge[->] node[below, near end] {$\rolestyle{HasCol}$} (cp)

(d) edge [->] node[above] {$\rolestyle{HasBit_1}$}		(i)

(ho) edge [->] node[above] {$\rolestyle{HasBit_0}$}		(o)
(hno) edge [->] node[above] {$\rolestyle{HasBit_0}$}		(o)
(vo) edge [->] node[above] {$\rolestyle{HasBit_0}$}		(o)
(vno) edge [->] node[above] {$\rolestyle{HasBit_0}$}		(o)

(hi) edge [->] node[above] {$\rolestyle{HasBit_1}$}		(i)
(hni) edge [->] node[above] {$\rolestyle{HasBit_1}$}	(i)
(vi) edge [->] node[above, near start] {$\rolestyle{HasBit_1}$}		(i)
(vni) edge [->] node[above] {$\rolestyle{HasBit_1}$}	(i)
;

\draw[rounded corners, ->] (d) -- ( 1, -4) -- node [below] {$\rolestyle{HasBit_0}$} ( -12 ,-4)  -- (o) ;

\end{tikzpicture}
\caption{A part of $\I_\tau$ with $\tau(2^n -1, 2^n-2) = \tau(2^n-1, 2^n-1) = c_1$.}
\label{conexp_rooted_example}
\end{figure}

The model $\I_\tau$ is displayed in Figure~\ref{conexp_rooted_example}.
By hypothesis, there is an additional c-match $\match$ for $q$ in $\I_\tau$. It is easily verified  that the additional match can only result from an atom $\rolestyle{Roots}(\ind{a}, z^{\Dir, (c, c'), k})$, with $\Dir \in \{ \Ho, \V \}$, $(c, c') \in (\C \times \C) \setminus \Dir$ and $k \in \{ 1, ... n \}$, being mapped onto $\rolestyle{Edge}(\ind{a}, \ind{b})$. From the definition of $\I_\tau$, this implies that there are two horizontally (or vertically) adjacent tiles, which positions are given by the elements $\match(y_{l, 2n}^{\Dir, (c, c'), k})$ and $\match(y_{r, 2n}^{\Dir, (c, c'), k})$, whose respective colors $c$ and $c'$ violate $\Dir$. Thus, $\tau$ is not an $(\Ho, \V)$-tiling. As this construction holds for any possible tiling $\tau$, we can infer that $(n, \C, \Ho, \V) \notin \tiling$.

\paragraph{$(\Leftarrow)$}
Assume $(n, \C, \Ho, \V) \notin \tiling$, and take some model $\I$ of $\kb$. By Lemma \ref{cm-lemma}, there is a homomorphism $f : \canmodof{\kb} \rightarrow \I$. 
Define $\tau: \{0, \dots, 2^n-1 \} \times \{0, \dots n \} \rightarrow \domain{\I}$ as follows: $\tau(h_n\dots h_1, v_n \dots v_1) := f(\ind{b} \rolestyle{H}^n_{h_n} \dots \rolestyle{H}^1_{h_1} \rolestyle{V}^n_{v_n} \dots \rolestyle{V}^1_{v_1}\rolestyle{HasCol})$ (again slightly abusing notation by applying working with binary encodings of numbers). There are three cases to consider:
\begin{itemize}
\item If there exists $(h_n\dots h_1, v_n \dots v_1)$ such that $\tau(h_n\dots h_1, v_n \dots v_1) \notin \{ \ind{c} \mid c \in \C \}$, then this provides a match of $q$ in $\I$ in which the subquery $q^{col}$ is mapped as follows:
$$z^{col} \mapsto \tau(h_n\dots h_1, v_n \dots v_1) \qquad y^{col}_0 \mapsto \ind{b} \qquad y^{col}_1 \mapsto f(\ind{b} \rolestyle{H}^n_{h_n}) \qquad \dots \qquad y^{col}_{2n} \mapsto f(\ind{b} \rolestyle{H}^n_{h_n} \dots \rolestyle{H}^1_{h_1} \rolestyle{V}^n_{v_n} \dots \rolestyle{V}^1_{v_1}),$$
and whose restriction to the counting variables provides a new c-match.
\item Otherwise, suppose there exists an element that is in the range of $\rolestyle{HasBit}_0$ that is not $\ind{zero}$, or that is in the range of $\rolestyle{HasBit}_1$ but not equal to $\ind{one}$, then this also provides a new c-match of $q$, in which either $z^{0}$ or $z^{1}$ is mapped to this element. Note that this kind of `error' may occur at any level of the tree of positions. This is why we included the ABox assertion $\rolestyle{HV}(\ind{b}, \ind{b})$, which makes it possible for us to loop as long as needed in order to obtain a sufficiently long chain of $\rolestyle{HV}$ to satisfy the query $q^0$ or $q^1$. 
\item Else, since $(n, \C, \Ho, \V) \notin \tiling$, there exist two adjacent positions with coordinates $(h_n\dots h_1, v_n \dots v_1)$ and $(h'_n\dots h'_1, v'_n \dots v'_1)$ such that $(\tau(h_n\dots h_1, v_n \dots v_1), \tau(h'_n\dots h'_1, v'_n \dots v'_1)) \in (\C \times \C) \setminus \Dir$, for $\Dir$ either $\Ho$ or $\V$. Letting $k$ be the bit from which the encoding of the non-$\Dir$ coordinate differs, we obtain a new c-match for $q$, in which the subquery $q^{\Dir, (\tau(h_n\dots h_1, v_n \dots v_1), \tau(h'_n\dots h'_1, v'_n \dots v'_1)), k)}$ is satisfied by mapping $z^{\Dir, (\tau(h_n\dots h_1, v_n \dots v_1), \tau(h'_n\dots h'_1, v'_n \dots v'_1)), k)}$  to  $\ind{b}$. 
\end{itemize}
In every case, there is an additional c-match for $q$. We thus obtain 
 $(\ta_\emptyset, [p+1, +\infty]) \in \qIntervals{q}{\kb}$.
\end{proof}

\section{\Exhaustive Rooted Counting CQs}
We have seen in Section \ref{sec:rooted} that the rootedness restriction is not by itself sufficient to lower the complexity of CCQ answering, whereas imposing both rootedness and exhaustiveness can sometimes yield better results. 
This motivates us to take a closer look at the case of \exhaustive rooted CCQs. The emerging complexity landscape is summarized in Table \ref{resstrongly}. 

Note that exhaustive CCQs constitute a very natural form of counting query, which ask for the number of different query matches for a given answer tuple. The query $q_2$ from Example~\ref{exqueries} is an \exhaustive rooted CCQ. 


%
%



\subsection{\Exhaustive Rooted CCQs in \dlc}\label{srooted-core}
We first consider \dlc\ KBs and pinpoint the precise combined complexity, 
which had not yet been considered. 


An essential ingredient is the following result that shows that it is possible to 
focus on query matches in the canonical model. 
It can be obtained by adapting a similar result 
about canonical bag interpretations \cite{nikolauetal:bag}. 

\begin{table}[t]
\centering
\begin{tabular}{lcc}
 & \textbf{Data}      & \textbf{Combined}		\\
\midrule
\dlc & \bf $\TC^0$-c   	 & \bf $\PP$-c \\
\dlh & \bf $\coNP$-c  & $\PIP$\textbf{-h}, $\PP$\textbf{-h} \& in $\coNEXP$ \\
\bottomrule
\end{tabular}
\caption{Complexity results for \exhaustive rooted CCQs}
\label{resstrongly}
\end{table}


%
%
\begin{toappendix}
We start by recalling an important property of the canonical model construction for $\DLc$ KBs.
\begin{lemma}
\label{canmodunique}
For any role $R \in \rni$ and anonymous element $d_1$ in the canonical model $\canmod$ of a $\DLc$ KB $\kb$, there is at most one element $d_2 \in \canmod$ such that $(d_1, d_2) \in R^{\canmod}$.
\end{lemma}
\begin{proof} We provide a proof for the sake of completeness.
From the definition of $R^{\canmod}$, if $d_1$ is an anonymous domain element and $(d_1, d_2) \in R^{\canmod}$, then either:
\begin{itemize}
\item $d_1 = d_2 \rolestyle{S^-}$ for some role $\rolestyle{S}$ such that $\tbox \models \rolestyle{S} \incl \rolestyle{R}$, or
\item $d_2 = d_1 \rolestyle{S}$ for some role $\rolestyle{S}$ such that $\tbox \models \rolestyle{S} \incl \rolestyle{R}$.
\end{itemize}
In both cases, since  $\Tmc$ is a $\DLc$ TBox, the condition on $\rolestyle{S}$ holds only if $\rolestyle{S} = \rolestyle{R}$. Moreover, we observe that if the first case holds, i.e., $d_1 = d_2 \rolestyle{R^-}$, then the definition of $\domain{\canmod}$ prevents the creation of an element $d_1 \rolestyle{R}$. It follows that only one of the preceding cases can hold, and so there can be at most one $d_2$ with $(d_1, d_2) \in R^{\canmod}$. 
\end{proof}

\end{toappendix}

\begin{theoremrep}
\label{canomini}
For every \dlc\ KB $\kb$ and \exhaustive rooted CCQ $q$, it holds that $
\qIntervals{q}{\kb} = \qIntervals{q}{\canmod}$.
\end{theoremrep}

\begin{proofsketch}
Exploiting the structure of 
 \dlc\ canonical models, one can show 
that if $\match_1, \match_2$
are distinct matches of an \exhaustive rooted CCQ $q$ in $\canmod$, 
then there exists a variable $v$ such that $\match_1(v) \neq \match_2(v)$ and $\match_1(v), \match_2(v) \in \ainds(\Amc)$. 
It follows that if we take an arbitrary model $\Imc$ of $\Kmc$, and let $f$ be a homomorphism of $\canmod$ into $\Imc$,
then $f$ injectively maps query matches in $\canmod$ to query matches in $\Imc$. 
\end{proofsketch}

\begin{proof}
The general argument was given in the sketch. All that remains to show is that if $\match_1, \match_2$
are distinct matches of a \exhaustive rooted CCQ $q$ in $\canmod$, 
then there exists a variable $v$ such that $\match_1(v) \neq \match_2(v)$ and $\match_1(v), \match_2(v) \in \ainds(\Amc)$. 

Suppose for a contradiction that this is not the case. There there are distinct matches $\match_1, \match_2$
 of $q$ in $\canmod$ such that for all variables  $v$ such that $\match_1(v) \neq \match_2(v)$, 
either $\match_1(v) \not \in \ainds(\Amc)$ or $\match_2(v) \not \in \ainds(\Amc)$. As $q$ is \exhaustive rooted, every variable $v$ is connected to either an answer variable or individual in the Gaifman graph. Let $d(v)$ denote the length of the shortest path from $v$ to an answer variable of individual. Note that $d(v)=0$ iff $v$ is an answer variable. Since $\match_1$ and $\match_2$
are distinct, there exists a variable $v$ such that $\match_1(v) \neq \match_2(v)$. Choose such a variable $v^*$ with minimal $d$-value, i.e., if $d(u) < d(v^*)$, then $\match_1(u) = \match_2(u)$. By assumption, either  $\match_1(v^*) \not \in \ainds(\Amc)$ or $\match_2(v^*) \not \in \ainds(\Amc)$. We'll suppose the former (the other case is treated analogously). Note that $v^*$ cannot be an answer variable (else we would have $\match_1(v^*) \in \ainds(\Amc)$). It follows that $d(v^*)>0$, and so we can find another variable $u^*$ and role name $R \in \rni$, with $d(u^*) = d(v^*)-1$ and either $R(u^*,v^*) \in q$ or $R^-(v^*,u^*) \in q$ (recall that if $R=P^-$, then $R^-=P$). As $\match_1$ and $\match_2$ are matches of $q$ in $\canmod$, we therefore have $(\match_1(u^*), \match_1(v^*)) \in R^{\canmod}$ and $(\match_2(u^*), \match_2(v^*)) \in R^{\canmod}$. Moreover, since $d(u^*) < d(v^*)$, we have $\match_1(u^*) = \match_2(u^*)$.
There are two cases to consider:
\begin{itemize}
\item Case 1: $\match_1(u^*)=\match_2(u^*) = c \in \ainds(\Amc)$. From the proof of Lemma \ref{canmodunique}, we know that $\match_1(v^*)= c R$.  The fact that $c R \in \Delta^{\canmod}$ implies that there is no individual $b$ such that $(c,b) \in R^{\canmod}$. Thus, we must have $\match_2(v^*)= c R$, which yields $\match_1(v^*)=\match_2(v^*)$, contradicting our earlier assumption.
\item Case 2: $\match_1(u^*)=\match_2(u^*)  \not \in \ainds(\Amc)$. 
By Lemma \ref{canmodunique}, there is a unique element $e$ such that $(\match_1(u^*), e) \in R^{\canmod}$. We thus obtain 
$\match_1(v^*) = e = \match_2(v^*)$, a contradiction. 
\end{itemize}
As both cases lead to a contradiction, it must therefore be the case that the statement holds. 
\end{proof}

We will also use the next lemma, implicit in \cite{DBLP:conf/ijcai/BienvenuOSX13}, constraining the possible images of matches in $\canmod$:

\begin{lemma}\label{words}
For every \dlc\ TBox $\Tmc$ and CCQ $q$, we can construct in polynomial time 
a set of words 
$\Gamma_{q,\Tmc}$ such that for every KB $\Kmc=(\Tmc, \Amc)$, match $\sigma$ of $q$ in $\canmod$, and variable $v$ of~$q$:
$\sigma(v)= a w$ for some $a \in \ainds(\Amc)$ and $w \in \Gamma_{q,\Tmc}$.
\end{lemma}

We are now ready to show that the problem is $\PP$-complete in combined complexity, 
and hence in $\mathsf{PSpace}$.

\begin{theoremrep}
\label{ppc}
In \dlc, \exhaustive rooted CCQ answering is $\PP$-complete w.r.t.\ combined complexity.
\end{theoremrep}
\begin{proofsketch}
The class $\PP$ contains all decision problems for which there exists a non-deterministic Turing machine (TM) such that,   when the input is a 
`yes' instance, then at least half of the computation paths accept, while on `no' instances, less than half of the computation paths accept. 

The lower bound is obtained by a reduction from the following $\PP$-complete problem \cite{bailey:ppcomplete}: given a propositional formula $\psi$ in  CNF and number $n$, decide whether $\psi$ has at least $n$ satisfying assignments.

We sketch the TM used to show $\PP$ membership, which takes as input a \dlc\ KB $\Kmc=(\Tmc, \Amc)$, an \exhaustive rooted CCQ $q(\tx)$, and candidate answer $(\ta,[m,+\infty])$: 
\paragraph*{Phase 1.} The TM 
constructs the set $\Gamma_{q,\Tmc}$ 
 from Lemma \ref{words}. 
\paragraph*{Phase 2.} 
The TM guesses a mapping $\match$ of the variables in $q$ to elements from $\{a w \mid a \in \ainds(\Amc), w \in \Gamma_{q,\Tmc} \}$. It then compares $m$ with the number  $C = |\Gamma_{q,\Tmc}|^{|q|}$ of possible mappings and proceeds accordingly:
 \begin{itemize}[leftmargin=.42cm,itemsep=-.05cm]
   \item if $m \geq \frac{C}{2} + 1$, the TM guesses an integer $i$ with $0 \leq i \leq 2m-3$ and accepts iff $\match$ is a c-match of $q(\ta)$ 
   and $i < C$;
   \item if $m < \frac{C}{2} +1$, the TM guesses an integer $i$ with $0 \leq i \leq 2C -2m + 1$ and accepts iff $\match$ is c-match for $q(\ta)$ or $i < C - 2m +2$.
 \end{itemize}
The guessed integer and comparisons ensure a suitable number of accepting paths. It can be verified that at least half of the paths are accepting iff  $(\ta,[m,+\infty]) \in \qIntervals{q}{\canmod}$.
\end{proofsketch}

\begin{toappendix}
We start by completing the argument for the $\PP$ upper bound.

\begin{proof}
%
%
Recall the algorithm described in the proof sketch.\smallskip\\
\noindent\textbf{Phase 1} The TM 
deterministically constructs the set $\Gamma_{q,\Tmc}$ of words from Lemma \ref{words}. 
\smallskip\\
\noindent\textbf{Phase 2} 
The TM guesses a mapping $\match$ of the variables in $q$ to elements from $\{a w \mid a \in \ainds(\Amc), w \in \Gamma_{q,\Tmc} \}$. It then compares $m$ with the number  $C = |\Gamma_{q,\Tmc}|^{|q|}$ of possible mappings and proceeds accordingly:
 \begin{itemize}[leftmargin=.42cm,itemsep=-.05cm]
   \item if $m \geq \frac{C}{2} + 1$, the TM guesses an integer $i$ with $0 \leq i \leq 2m-3$ and accepts iff $\match$ is a c-match of $q(\ta)$ 
   and $i < C$;
   \item if $m < \frac{C}{2} +1$, the TM guesses an integer $i$ with $0 \leq i \leq 2C -2m + 1$ and accepts iff $\match$ is c-match for $q(\ta)$ or $i < C - 2m +2$.
 \end{itemize}

Due to Theorem \ref{canomini} and Lemma~\ref{words}, an input is a `yes' instance iff $\qAnswers{q}{\canmod}_\ta \geq m$ (recall that $\qAnswers{q}{\canmod}_\ta$ denotes the exact number of c-matches for $q(\ta)$ in $\canmod$). To finish the proof of $\PP$ membership, we need to examine the number of accepting computation paths for the described TM and show that when $\qAnswers{q}{\canmod}_\ta \geq m$, at least half of the computation paths accept, and when $\qAnswers{q}{\canmod}_\ta < m$, less than half of the computation paths accept.  Let us consider the two cases from Phase 2:
\begin{itemize}
\item
If $m \geqslant \frac{C}{2} + 1$, then the number of accepting computation paths is $\qAnswers{q}{\canmod}_\ta \times C$, corresponding to cases where the TM guesses a mapping that is a c-match, then guess a number $0 \leq i <C$. The total number of computation paths is $C \times (2m - 2)$, corresponding to a guess of one of the $C$ mappings, then the guess of an integer $0 \leq i \leq 2m-3$.  
\item
If $m < \frac{C}{2} + 1$, then the number of accepting computation paths is $$\qAnswers{q}{\canmod}_\ta \times (2C - 2m + 2) + (C - \qAnswers{q}{\canmod}_\ta) \times (C - 2m + 2)= C(C - 2m + \qAnswers{q}{\canmod}_\ta + 2),$$ corresponding to the sum of the number of cases where we guess a c-match followed by an integer $0 \leq i \leq 2C -2m + 1$ and the  number of cases where we guess a mapping that is not a c-match followed by an integer $i$ with $0 \leq i < C - 2m +2$. The total number of computation paths is $C \times (2C - 2m + 2)$ (guess one of the $C$ mappings, then guess an integer $0 \leq i \leq 2C -2m + 1$).
\end{itemize}
In both cases, it is easily verified that: $$\qAnswers{q}{\canmod}_\ta \geq m \Longleftrightarrow \frac{\# \mathsf{accepting \ computation \ paths}}{\# \mathsf{possible \ computation \ paths}} > \frac{1}{2}.$$
(Note that in the first case, we always have $m \geq 2$, so the value $2m-2$ in the denominator is positive, while in the second case, $C \geq 1$ implies that the value $(2C - 2m + 2)$ in the denominator is positive.)
\end{proof}

We next give the proof of $\PP$-hardness.

\begin{proof}
We recall that the lower bound is by reduction from the following $\PP$-complete problem: given a propositional formula $\psi$ in  CNF and number $n$, decide whether $\psi$ has at least $n$ satisfying assignments.

Consider an instance of this problem, given by the formula $\psi := \exists \tu \bigwedge^l_{k = 1} \xi_k $ (with $\xi_k$ is a $3$-clause) and number $N$. We consider the KB $\kb_\psi= (\emptyset, \abox_\psi)$, which has an empty TBox, and whose ABox $\abox_\psi$ contains the following assertions:
%
%
%
%
\begin{itemize}
\item
$\rolestyle{Clause_k}(\ind{a}, \ind{\xi^p_k})$ for each clause $\xi_k$ and each $p \in \{1, ... 7\}$, with each $\ind{\xi^p_k}$ representing one of the 7 satisfying assignments for the clause $\xi_k$;
\item $\rolestyle{Asn_1}(\ind{\xi^p_k}, \xi^p_k(\omega_k^1))$, $\rolestyle{Asn_2}(\ind{\xi^p_k}, \xi^p_k(\omega_k^2))$ and $\rolestyle{Asn_3}(\ind{\xi^p_k}, \xi^p_k(\omega_k^3))$ for each $p = 1, ... 7$ and each clause $\xi_k$, where $\xi^p_k(\omega_k^i)$ is the truth value ($\true$ or $\false$) assigned by $\xi^p_k$ to the $i$th variable occurring in the $k$th clause.
\end{itemize}
It may be helpful to refer to Figure \ref{pict-cm}, which depicts the canonical model of $\canmodof{\kb_\psi}$ for an example formula $\psi$.

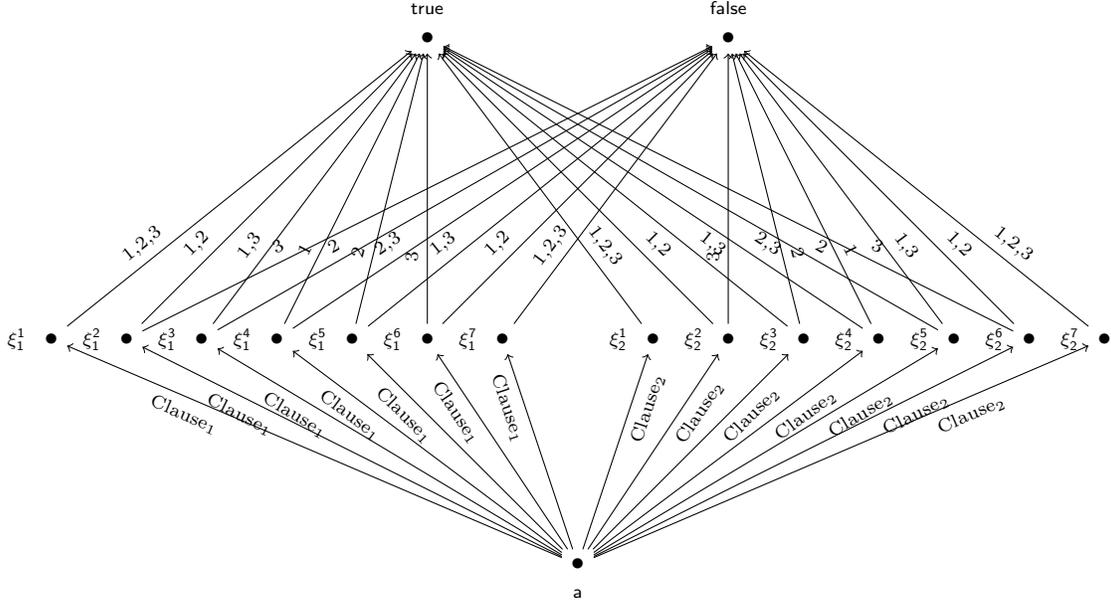
\begin{figure}[t]
\centering
\begin{tikzpicture}


\node at ( 0, -3) (psi) 		[label=below:$\scriptstyle \ind{a}$]		{$\bullet$};

\node at (-7, 0) (sig11)	[label=left:$\scriptstyle \ind{\xi_1^1}$]   {$\bullet$};
\node at (-6, 0) (sig12)	[label=left:$\scriptstyle \ind{\xi_1^2}$]   {$\bullet$};
\node at (-5, 0) (sig13)	[label=left:$\scriptstyle \ind{\xi_1^3}$]   {$\bullet$};
\node at (-4, 0) (sig14)	[label=left:$\scriptstyle \ind{\xi_1^4}$]   {$\bullet$};
\node at (-3, 0) (sig15)	[label=left:$\scriptstyle \ind{\xi_1^5}$]   {$\bullet$};
\node at (-2, 0) (sig16)	[label=left:$\scriptstyle \ind{\xi_1^6}$]   {$\bullet$};
\node at (-1, 0) (sig17)	[label=left:$\scriptstyle \ind{\xi_1^7}$]   {$\bullet$};

\node at ( 1, 0) (sig21)	[label=left:$\scriptstyle \ind{\xi_2^1}$]   {$\bullet$};
\node at ( 2, 0) (sig22)	[label=left:$\scriptstyle \ind{\xi_2^2}$]   {$\bullet$};
\node at ( 3, 0) (sig23)	[label=left:$\scriptstyle \ind{\xi_2^3}$]   {$\bullet$};
\node at ( 4, 0) (sig24)	[label=left:$\scriptstyle \ind{\xi_2^4}$]   {$\bullet$};
\node at ( 5, 0) (sig25)	[label=left:$\scriptstyle \ind{\xi_2^5}$]   {$\bullet$};
\node at ( 6, 0) (sig26)	[label=left:$\scriptstyle \ind{\xi_2^6}$]   {$\bullet$};
\node at ( 7, 0) (sig27)	[label=left:$\scriptstyle \ind{\xi_2^7}$]   {$\bullet$};

\node at (-2, 4) (t) 	[label=above:$\scriptstyle \ind{true}$]		{$\bullet$};
\node at ( 2, 4) (f) 	[label=above:$\scriptstyle \ind{false}$]		{$\bullet$};

\path[every node/.style={sloped}]
(psi) edge [->] node[below, near end] 			{$\scriptstyle \rolestyle{Clause_1}$} (sig11)
(psi) edge [->] node[below, near end] 			{$\scriptstyle \rolestyle{Clause_1}$} (sig12)
(psi) edge [->] node[below, near end] 			{$\scriptstyle \rolestyle{Clause_1}$} (sig13)
(psi) edge [->] node[below, near end] 			{$\scriptstyle \rolestyle{Clause_1}$} (sig14)
(psi) edge [->] node[below, near end] 			{$\scriptstyle \rolestyle{Clause_1}$} (sig15)
(psi) edge [->] node[below, near end] 			{$\scriptstyle \rolestyle{Clause_1}$} (sig16)
(psi) edge [->] node[below, near end] 			{$\scriptstyle \rolestyle{Clause_1}$} (sig17)

(psi) edge [->] node[below, near end] 			{$\scriptstyle \rolestyle{Clause_2}$} (sig21)
(psi) edge [->] node[below, near end] 			{$\scriptstyle \rolestyle{Clause_2}$} (sig22)
(psi) edge [->] node[below, near end] 			{$\scriptstyle \rolestyle{Clause_2}$} (sig23)
(psi) edge [->] node[below, near end] 			{$\scriptstyle \rolestyle{Clause_2}$} (sig24)
(psi) edge [->] node[below, near end] 			{$\scriptstyle \rolestyle{Clause_2}$} (sig25)
(psi) edge [->] node[below, near end] 			{$\scriptstyle \rolestyle{Clause_2}$} (sig26)
(psi) edge [->] node[below, near end] 			{$\scriptstyle \rolestyle{Clause_2}$} (sig27)

(sig11) edge [->] node[above, near start] 	{$\scriptstyle 1,2,3$} (t)

(sig12) edge [->] node[above, near start] 	{$\scriptstyle 1,2$} (t)
(sig12) edge [->] node[above, near start] 	{$\scriptstyle 3$} (f)

(sig13) edge [->] node[above, near start] 	{$\scriptstyle 1,3$} (t)
(sig13) edge [->] node[above, near start] 	{$\scriptstyle 2$} (f)

(sig14) edge [->] node[above, near start] 	{$\scriptstyle 1$} (t)
(sig14) edge [->] node[above, near start] 	{$\scriptstyle 2, 3$} (f)

(sig15) edge [->] node[above, near start] 	{$\scriptstyle 2$} (t)
(sig15) edge [->] node[above, near start] 	{$\scriptstyle 1, 3$} (f)

(sig16) edge [->] node[above, near start] 	{$\scriptstyle 3$} (t)
(sig16) edge [->] node[above, near start] 	{$\scriptstyle 1, 2$} (f)

(sig17) edge [->] node[above, near start] 	{$\scriptstyle 1,2,3$} (f)

(sig21) edge [->] node[above, near start] 	{$\scriptstyle 1,2,3$} (t)

(sig22) edge [->] node[above, near start] 	{$\scriptstyle 1,2$} (t)
(sig22) edge [->] node[above, near start] 	{$\scriptstyle 3$} (f)

(sig23) edge [->] node[above, near start] 	{$\scriptstyle 1,3$} (t)
(sig23) edge [->] node[above, near start] 	{$\scriptstyle 2$} (f)

(sig24) edge [->] node[above, near start] 	{$\scriptstyle 2,3$} (t)
(sig24) edge [->] node[above, near start] 	{$\scriptstyle 1$} (f)

(sig25) edge [->] node[above, near start] 	{$\scriptstyle 2$} (t)
(sig25) edge [->] node[above, near start] 	{$\scriptstyle 1,3$} (f)

(sig26) edge [->] node[above, near start] 	{$\scriptstyle 3$} (t)
(sig26) edge [->] node[above, near start] 	{$\scriptstyle 1,2$} (f)

(sig27) edge [->] node[above, near start] 	{$\scriptstyle 1,2,3$} (f)
;

\end{tikzpicture}
\caption{The canonical model $\canmodof{\kb_\psi}$ with $\psi = (u_1 \vee \neg u_2 \vee \neg u_3) \wedge (\neg u_1 \vee u_3 \vee u_4)$}
\label{pict-cm}
\end{figure}
As for the query, we consider the following \exhaustive rooted CCQ (depicted in Figure \ref{ex-pp-query}):
$$
q_\psi := \exists z_{\xi_1} \dots \exists z_{\xi_l} \ \exists z_{u_1} \dots \exists z_{u_n} \bigwedge_{k = 1}^l \left( \rolestyle{Clause_k}(\ind{a}, z_{\xi_k}) \wedge \bigwedge_{i = 1}^3 \left( \rolestyle{Asn_i}(z_{\xi_k}, z_{\omega_k^i}) \right) \right)
$$
\begin{figure}[!h]
\centering
\begin{tikzpicture}
[every node/.style={scale=0.9, sloped}]
\node at ( 0, 1) (psi) 		[label=below:$\ind{a}$]			{$\gx$};
\node at (-2, 3) (xi1)		[label=left:$z_{\xi_1}$]   		{$\gz$};
\node at ( 2, 3) (xi2)		[label=right:$z_{\xi_2}$]   	{$\gz$};
\node at (-3, 5) (u1) 	[label=above:$z_{u_1}$]	{$\gz$};
\node at (-1, 5) (u2) 	[label=above:$z_{u_2}$]	{$\gz$};
\node at ( 1, 5) (u3) 	[label=above:$z_{u_3}$]	{$\gz$};
\node at ( 3, 5) (u4) 	[label=above:$z_{u_4}$]	{$\gz$};

\path
(psi) edge [->] node[below] {$\rolestyle{Clause_1}$} 	(xi1) 
(psi) edge [->] node[below] {$\rolestyle{Clause_2}$} 	(xi2)

(xi1) edge [->] node[near start, below] 			{$\rolestyle{Asn_1}$} (u1)
(xi1) edge [->] node[near start, above] 			{$\rolestyle{Asn_2}$} (u2)
(xi1) edge [->] node[near start, below] {$\rolestyle{Asn_3}$} (u3)
(xi2) edge [->] node[near start, below] 			{$\rolestyle{Asn_1}$} (u1)
(xi2) edge [->] node[near start, above] 			{$\rolestyle{Asn_2}$} (u3)
(xi2) edge [->] node[near start, below] {$\rolestyle{Asn_3}$} (u4);

\end{tikzpicture}
\caption{The query $q_\psi$ with $\psi = (u_1 \vee \neg u_2 \vee \neg u_3) \wedge (\neg u_1 \vee u_3 \vee u_4)$}
\label{ex-pp-query}
\end{figure}
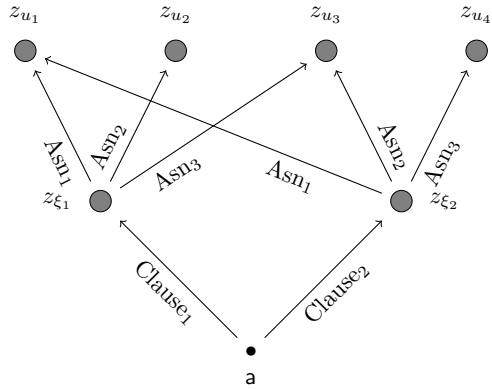

To complete the proof, we establish the following claim. 

\noindent \paragraph{Claim.} $ (\emptyset, [N, +\infty]) \in \qIntervals{q_\psi}{\kb_\psi} \Longleftrightarrow  \, \psi$ has at least $N$ satisfiying assignments

\paragraph{$(\Rightarrow)$}

Assume $ (\emptyset, [N, +\infty]) \in \qIntervals{q_\psi}{\kb_\psi}$. This implies in particular that there are $N$ c-matches for $q_\psi$ in $\canmodof{\kb_\psi}$. Since the TBox is empty, the domain of $\canmodof{\kb_\psi}$ is $\ainds(\Amc_\psi)$, and  $\canmodof{\kb_\psi}$ makes true precisely the assertions in $\Amc_\psi$. By examining $q_\psi$ and $\Amc_\psi$, we see that each of the matches of $q_\psi$ in $\canmodof{\kb_\psi}$ maps each of the variables $z_{u_i}$ to either $\true$ or $\false$. We can therefore associate with each match $\match$ the following truth assignment for the variables $u_1, \ldots, u_n$: $\tau_\match(u_i) = \match(z_{u_i})$. 
By further examining the definition of the individuals $\ind{\xi_k^p}$ and the roles $\rolestyle{Asn_1}, \rolestyle{Asn_2}, \rolestyle{Asn_3}$, it is easy to verify that each $\tau_\match$ is a satisfying assignment for $\psi$. Moreover, since we know we have $N$ such assignments, it only remains to show that each match $\match$ yields a distinct assignment $\tau_\match$. To see why this is the case, observe that once we know the images of all of the variables $z_{u_i}$, there is a unique way of mapping the variables $z_{\xi_p}$. It follows that $\psi$ has at least $N$ satisfying assignments.


\paragraph{$(\Leftarrow)$}
Assume $\psi$ has at least $N$ satisfying assignments. 
Therefore, we have $\tau_1, \dots \tau_N$ distinct assignments for $u_1, \dots u_n$ satisfying $\psi$. This ensures that, if we define $\match_{\tau_m}(z_{u_i}) = \tau_m(u_i)$, we can always extend the mapping $\match_{\tau_m}(z_{u_i})$ into a match for $q_\psi$, yielding  $N$ distinct matches. Note that this holds in any model since we only need the `ABox part' of the model, hence $(\emptyset, [N, +\infty]) \in \qIntervals{q_\psi}{\kb_\psi}$.

\end{proof}

\end{toappendix}

\subsection{\Exhaustive Rooted CCQs in $\DLr$}
We now turn to \dlh\ KBs. Our first result is negative: \exhaustive rooted CCQs do not enjoy 
lower data complexity. This is shown by another reduction from $\tcol$ which involves 
ideas from our proof of Theorem \ref{rootedconphard} and the proof of Lemma 16 from  \cite{kostylevreutter:count}.

\begin{theoremrep}
In \dlh, \exhaustive rooted CCQ answering is $\coNP$-complete w.r.t.\ data complexity.
\end{theoremrep}

\begin{proof}
The main idea is the same as in proof of Theorem~\ref{rootedconphard}. However, due to the lack of existential variables, we can no longer `reach' the colors without taking into account the paths leading to them. To address this difficulty, we translate into our context an idea from \cite{kostylevreutter:count}, which takes advantage of role inclusions. 

Starting from an instance $\graph = (\vertices, \edges)$ of the decision problem $\tcol$, we consider the ABox $\abox_\graph$ given by:
\begin{align*}
\abox_\graph =& \{\rolestyle{Vertex}(\ind{a}, \ind{u}) \suchthat u \in \vertices \} \cup \{ \rolestyle{Edge}(\ind{u_1}, \ind{u_2}) \suchthat (u_1, u_2) \in \edges \} \\
& \cup \{ \rolestyle{Vertex}(\ind{a}, \ind{a_v}), \rolestyle{Edge}(\ind{a_v}, \ind{a_v}), \rolestyle{HasCol}(\ind{a_v}, \ind{r}) \} \\
& \cup \{\rolestyle{Colors}(\ind{u}, \ind{r}) \suchthat u \in \vertices \} \cup \{\rolestyle{Colors}(\ind{u}, \ind{g}) \suchthat u \in \vertices \} \cup \{\rolestyle{Colors}(\ind{u}, \ind{b}) \suchthat u \in \vertices \}
\end{align*}
and the TBox $\tbox := \{ \exists \rolestyle{Vertex^-} \incl \exists \rolestyle{HasCol}, \rolestyle{HasCol} \incl \rolestyle{Colors} \}$, and we denote by $\kb_\graph = (\tbox, \abox_\graph)$ the resulting KB. A part of the canonical model of $\kb_\graph$ is depicted in Figure~\ref{stronglycanmodconp}. As in the proof of Theorem~\ref{rootedconphard},
we use $\exists \rolestyle{Vertex^-} \incl \exists \rolestyle{HasCol}$ to assign colors to vertices, and the more general role $\rolestyle{Colors}$ will be used to detect colors.  

\begin{figure}[h]
\centering
\begin{tikzpicture}
[every node/.style={scale=0.8, sloped}]
\node at ( 2, 0) (root) [label={above:$\ind{a}$}]	  {$\bullet$};
\node at ( 4, 0) (av) [label={right:$\ind{a_v}$}]	  {$\bullet$};
\node at (-2, 1) (u)	[label=above right:$\ind{u_1}$]	  		{$\bullet$};
\node at (-2,-1) (v)	[label=below right:$\ind{u_2}$]			{$\bullet$};
\node at (-2.5, 3) (uc) 	[label={above:$\ind{u_1}\rolestyle{HasCol}$}]					{$\circ$};
\node at (-2.5,-3) (vc) 	[label={below:$\ind{u_2}\rolestyle{HasCol}$}]					{$\circ$};
\node at (-5, 1.5) (r) 	[label=above:$\ind{b}$]	{$\bullet$};
\node at (-5, 0) (g)   	[label=above:$\ind{g}$]	{$\bullet$};
\node at (-5,-1.5) (b)   [label=above:$\ind{r}$]	{$\bullet$};

\path 
(root) edge [->] node[above] {$\rolestyle{Vertex}$} (u)
(root) edge [->] node[below] {$\rolestyle{Vertex}$} (v)

(u) edge [->] node[below] 	{$\rolestyle{Edge}$} (v)
(u) edge [->] node[above] 	{$\rolestyle{HasCol}$} (uc)
(v) edge [->] node[below] 	{$\rolestyle{HasCol}$} (vc)

(root) edge [->] node[above] {$\rolestyle{Vertex}$} (av)
(av) edge [loop above] node [above]	{$\rolestyle{Edge}$} (av)
(u) edge [->] node [near end, above]	{$\rolestyle{Colors}$} (r)
(u) edge [->] node [near end, above]	{$\rolestyle{Colors}$} (g)
(u) edge [->] node [near end, below]	{$\rolestyle{Colors}$} (b)
(v) edge [->] node [near end, above]	{$\rolestyle{Colors}$} (r)
(v) edge [->] node [near end, below]	{$\rolestyle{Colors}$} (g)
(v) edge [->] node [near end, below]	{$\rolestyle{Colors}$} (b);

\draw[->, rounded corners] (av) -- (3,-4) -- node [above] {$\rolestyle{HasCol}$} (-4, -4) -- (b);

\end{tikzpicture}
\caption{A part of $\canmodof{\kb_\graph}$ with $(u_1, u_2) \in \edges$.}
\label{stronglycanmodconp}
\end{figure}

We consider the two following \exhaustive rooted CCQs:
\begin{align*}
q^{edge} &= \exists z_c ~ \exists z_1 ~ \exists z_2 ~ \rolestyle{Vertex}(\ind{a}, z_1)  \wedge  \rolestyle{Vertex}(\ind{a}, z_2) \wedge \rolestyle{Edge}(z_1, z_2)  \wedge  \rolestyle{HasCol}(z_1, z_c)  \wedge  \rolestyle{HasCol}(z_2, z_c) \\
q^{col} &= \exists z_v ~ \exists z  ~ \rolestyle{Vertex}(\ind{a}, z_v) \wedge \rolestyle{Colors}(z_v, z) 
\end{align*}
and let $q$ be the query obtained by taking the conjunction of these two queries and keeping all of the variables existentially quantified. The query $q$ is displayed in Figure \ref{stronglyrootedqueryconp}. Observe that while it is similar to the query from the proof of Theorem~\ref{rootedconphard} (see Figure \ref{rootedqueryconp}), the two existential variables in that query ($y_c,y$) have been replaced with counting variables ($z_c, z_v$), and one of the $\rolestyle{HasCol}$ atom has been changed to a $\rolestyle{Colors}$ atom.

\begin{figure}[h]
\centering
\begin{tikzpicture}
[every node/.style={scale=0.8, sloped}]
\node at ( 0, 3) (r) [label=below:$\ind{a}$]  {$\gx$};
\node at (-1.5, 4) (z1)	[label=above:$z_1$]	{$\gz$};
\node at (-1.5, 2) (z2) 	[label=below:$z_2$]	{$\gz$};
\node at (-3, 3) (yc) 	[label=left:$z_c$]	{$\gz$};
\node at (2, 3) (yv)	[label=below:$z_v$]  {$\gz$};
\node at ( 4, 3) (zc) 	[label=right:$z$]	{$\gz$};

\path
(r) edge [->] node[above] {$\rolestyle{Vertex}$} 	(z1)
(r) edge [->] node[below] {$\rolestyle{Vertex}$} 	(z2)
(r) edge [->] node[above] {$\rolestyle{Vertex}$} 	(yv)
(yv) edge [->] node[above] {$\rolestyle{Colors}$} 	(zc)
(z1) edge [->] node[below] {$\rolestyle{Edge}$} 	(z2)
(z1) edge [->] node[above] {$\rolestyle{HasCol}$} 	(yc)
(z2) edge [->] node[below] {$\rolestyle{HasCol}$} 	(yc); 
\end{tikzpicture}
\caption{The \exhaustive rooted CCQ $q$, which is the conjunction of $q^{edge}$ (left part) and $q^{col}$ (right part).}
\label{stronglyrootedqueryconp}
\end{figure}
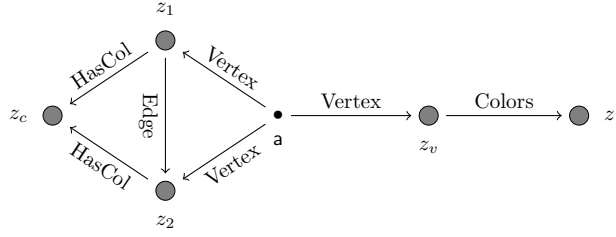

%
%
%
%
%
%
%

\noindent It is not hard to see that $(\ta_\emptyset, [3 |\vertices| + 1, +\infty]) \in \qIntervals{q}{\kb_\graph}$. 
Indeed, there are at least $3 |\vertices|$ matches of $q$ in any model $\I$ of $\kb$, obtained as follows:
$$
z_1, z_2  \mapsto  \ind{a_v} \qquad
z_c  \mapsto  \ind{r} \qquad
z_v \mapsto \ind{u} ~ (u \in \vertices) \qquad
z  \mapsto  \ind{r} \suchthat \ind{g} \suchthat \ind{b}
$$
and one additional match given by:
$$
z_1, z_2, z_v  \mapsto  \ind{a_v} \qquad
z_c, z  \mapsto  \ind{r}
$$
To complete the proof, we establish the following claim. \medskip 

\noindent \textbf{Claim.} $ (\emptyset, [3 |\vertices| + 2, + \infty]) \in \qIntervals{q}{\kb_\graph} \Longleftrightarrow \graph \notin$ $\tcol. $
 
\paragraph{$(\Rightarrow)$} This direction is proven in the same manner as the claim in the proof of Theorem~\ref{rootedconphard}. 
We assume $(\ta_\emptyset, [3|\vertices|+2, +\infty]) \in \qIntervals{q}{\kb_\graph}$ and take a possible coloring $\tau : \vertices \rightarrow \{ \ind{r}, \ind{g}, \ind{b} \}$. We then use $\tau$ to build a model $\I_\tau$  of $\kb_\graph$ and use the existence of an additional match $\match$ to show that $\tau$ contains a monochromatic edge (hence $\graph \notin \tcol$).
%
 
\paragraph{$(\Leftarrow)$} 
Assume $\graph \notin \tcol$, and take some model $\I$ of $\kb_\graph$. By Lemma \ref{cm-lemma}, there is a homomorphism $f : \canmodof{\kb_\graph} \rightarrow \I$. 
Define $\tau:  \vertices \rightarrow \domain{\I}$ as follows: $\tau(u)= f(\ind{u}\rolestyle{HasCol})$. Note that $\tau$ is well defined, as the inclusion $\exists \rolestyle{Vertex^-} \incl \exists \rolestyle{HasCol}$ ensures that there is an element $\ind{u}\rolestyle{HasCol}$ in $\canmodof{\kb_\graph}$.  There are two cases to consider:
\begin{itemize}
\item If there exists $u \in \vertices$ such that $\tau(u) \notin \{ \ind{r}, \ind{g}, \ind{b} \}$, then the axiom $\rolestyle{HasCol} \incl \rolestyle{Colors}$ ensures $(\ind{u}^\I, \tau(u)) \in \rolestyle{Colors}^\I$, which provides an additional match of $q^{color}$ in $\I$ with $z \mapsto \tau(u)$ and $z_v \mapsto \ind{u}^\I$.
\item Else, since $\graph \notin \tcol$, there exists an edge $(u_1, u_2) \in \edges$ such that $\tau(u_1) = \tau(u_2)$. It yields a new match given by: 
 $$
 z  \mapsto  \ind{r} \quad
 z_v \mapsto  \ind{a_v} \quad
 z_1 \mapsto  \ind{u_1} \quad
 z_2  \mapsto  \ind{u_2} \quad
 z_c  \mapsto  \tau(u_1) ~ ( = \tau(u_2)) 
 $$
\end{itemize}
%
In both cases, there is an additional c-match for $q$. We thus obtain 
 $(\ta_\emptyset, [3|\vertices| + 2, +\infty]) \in \qIntervals{q}{\kb_\graph}$.
\end{proof}
%
%

More positively, we can show an improved $\coNEXP$ upper bound in combined complexity for \exhaustive rooted CCQs. 
We briefly sketch the proof, which involves highly non-trivial modifications to the argument used for general CCQs. 

We first introduce a 
more refined notion of interleaving, which replaces the mapping $f'$ 
by the following mapping $f^{*}$:
$$\begin{array}{rcl}
f^*(a) & = & f(a) \\
f^{*}(\omega R) & =& \left\{ \begin{array}{ll}
f(\omega R) & \text{ if } f^{*}(\omega), f(\omega R) \in \Delta^* \\
f^{*}(\omega) R & \text{ otherwise} \\
\end{array} \right.
\end{array}
$$
It is possible to prove that when $q$ is an \exhaustive rooted CCQ, this modified interleaving yields a countermodel. Moreover, 
it has a very particular structure, essentially corresponding to the canonical model of the restriction of $f(\canmod)$ to $\domain{*}$ (viewed as an ABox). Importantly, this means that instead of guessing a whole countermodel, it suffices to guess an initial, exponential-size portion (the $|q|$-neighborhood of $\domain{*}$), providing the basis for a $\coNEXP$ decision procedure. 

\begin{theoremrep}
In \dlh, \exhaustive rooted CCQ answering is in $\coNEXP$ w.r.t.\ combined complexity.
\end{theoremrep}

\begin{proof}
We may assume w.l.o.g.\ that the initial homomorphism $f: \canmod \rightarrow \I$ is chosen to respect the following property ($\bigstar$): if $w_1 \rolestyle{R}, w_2 \rolestyle{R} \in \domain{\canmod}$ and $f(w_1) = f(w_2)$, then $f(w_1 \rolestyle{R}) = f(w_2 \rolestyle{R})$. Such a homomorphism can easily be built starting from an arbitrary homomorphism $g$, by choosing a `main branch' whenever a choice is possible and copying from it. Formally, given a breadth-first ordering $\preceq$ of elements in $\domain{\canmod}$, we start by setting $f := g$. Then, we explore the elements according to $\preceq$ and at each step, say at element $w_1$, we explore all elements $w_2$ such that $w_1 \preceq w_2$. If ever $f(w_1) = f(w_2)$, we redefine $f(w_2 w) := f(w_1 w)$ for every word $w$ such that $w_1 w, w_2 w \in \domain{\canmod}$. Since $\preceq$ is breadth-first, the image $f(w_1)$ will no longer be redefined after step $w_1$, which ensures the resulting homomorphism $f$ is well-defined. 

Recall that we introduced in the body of the paper a more refined notion of interleaving, which replaces the mapping $f'$ by the following mapping $f^{*}$:
$$\begin{array}{rcl}
f^{*} : \domain{\canmod} & \rightarrow & \domain{*} \cup \domain{\canmod} \\
a & \mapsto & f(a) \\
w R & \mapsto & \left\{ \begin{array}{ll}
f(w R) & \text{ if } f^{*}(w), f(w R) \in \Delta^* \\
f^{*}(w) R & \text{ otherwise} \\
\end{array} \right.
\end{array}
$$

We let $\Imc^*$ be the interpretation obtained by applying $f^*$ to $\canmod$. It is helpful to observe that $\Imc^*$ essentially coincides with the canonical model $\Cmc^{\kb^*}$ of the KB $\kb^*$ whose TBox is $\tbox$ and whose ABox $\abox^*$ consists of the facts from $\domain{*} \cap f^*(\Cmc^\kb)$ (treating such elements as ABox individuals). More explicitly, $\abox^*$ contains  the concept assertion $\rolestyle{A}(\ind{t})$ for each atomic concept $A \in \cnames$ and domain element $t \in \domain{*} \cap f^*(\Cmc^\kb)$ such that $t \in f^*(\rolestyle{A}^{\canmod})$, and the role assertion $\rolestyle{R}(\ind{t_1}, \ind{t_2})$ for each atomic role $\rolestyle{R} \in \rnames$ and domain elements $t_1, t_2 \in \domain{*} \cap f^*(\Cmc^\kb)$ such that $(t_1, t_2) \in f^*(\rolestyle{R}^{\canmod})$. 
%

This alternative way of viewing $\Imc^*$, together with our assumption $(\bigstar)$, makes clear that the following mapping is a homomorphism from $\I^*$ to $\I$:  
$$
\begin{array}{rcl}
\rho^* : \domain{\I^*} & \rightarrow & \domain{\I }\\
   	     f^*(d)  &     \mapsto & f(d).
\end{array}
$$
Indeed, $(\bigstar)$ ensures the choice of $d$ doesn't affect the image $f(d)$, thus $\rho^*$ is well-defined. Formally, we proceed by induction on elements of $\I^*$. If $f^*(d) \in \domain{*}$, then by the definition of $f^*$, we must have either $d \in \individuals$ with $f^*(d) = f(d)$, or $d = w \rolestyle{R}$ with $f^*(w) \in \domain{*}$, $f(d) \in \domain{*}$ and $f^*(d) = f(d)$. In both cases, $\rho^*(f^*(d)) = f(d) = f^*(d)$, which is independent from the choice of $d$. Otherwise, suppose $\rho^*$ is well-defined for $\omega$, and consider some $f^*(d) = \omega \rolestyle{R} \notin \domain{*}$ and $d'$ such that $f^*(d') = f^*(d) = \omega \rolestyle{R}$. By the definition of $f^*$, we must have $d = w \rolestyle{R}$ with $f^*(w) = \omega$, and same for $d'$, that is, $d' = w' \rolestyle{R}$ with $f^*(w') = \omega$. By our inductive assumption, we have $\rho^*(\omega) = f(w) = f(w')$. Property ($\bigstar$) now ensures $f(w \rolestyle{R}) = f(w' \rolestyle{R})$, that is $\rho^*( f^*( d ) ) = \rho^*( f^*( d' ))$. 

The mapping $\rho^*$ being a homomorphism then follows from the definition of concept and role interpretations in $\I^*$. 
In particular, this means that $\I^*$ is a model of $\kb$.
Compared to the homomorphism $\rho$ used to connect the interleaving $\Imc'$ with the countermodel $\Imc$, we have lost the property that $\rho^{-1}(\domain{*}) = \domain{*}$. Therefore, proving that $\I^*$ is a countermodel requires a different argument that exploits the \exhaustive rooted assumption on the query.

Consider a match $\match : \tx \cup \tz \rightarrow \domain{\I^*}$ of $q(\ta)$ in $\I^*$. Let us first suppose that there is a counting variable $z \in \tz$ such that $\match(z) \notin \domain{*}$, in which case we must have $\match(z) = \ind{t} w$ for some $\ind{t} \in \domain{*} \cap f^*(\Cmc^\kb)$ and some non-empty word $w$. Since $q$ is \exhaustive rooted, all intermediate elements $\ind{t} w'$ with $w'$ a prefix of $w$, must be reached by some other counting variables. In particular, one of these counting variables, say $z_0$, must map onto $\ind{t} w_0$, with $w_0$ the first symbol of $w$. From the definition of $f^*$, we also have a word $w_\ind{t}$ such that $f^*(w_\ind{t})=f(w_\ind{t}) = \ind{t}$.
However, via the homomorphism $\rho^*$, we can transform $\match$  
into a match  $\rho^* \circ \match :  \tx \cup \tz \rightarrow \domain{\I}$ in the original countermodel $\I$. In particular, we will have $\rho^*(\match(z_0)) = \rho^*(\ind{t}w_0) = \rho^*(f^*(w_\ind{t}) w_0) = \rho^*(f^*(w_\ind{t} w_0)) = f(w_\ind{t} w_0)$. Thus, $f(w_\ind{t} w_0)$ belongs to the image of the match $\rho^* \circ \match$ in $\I$. From the definition of $\Delta^*$, we can thus infer that $f(w_\ind{t} w_0) \in \domain{*}$. But since $f^*(w_\ind{t}) = \ind{t} \in \domain{*}$ and $f(w_\ind{t} w_0) \in \domain{*}$, we have $f^*(w_\ind{t} w_0) = f(w_\ind{t} w_0)$, and therefore the element $\ind{t} w_0$ is not introduced by $f^*$ (it would contradict the property ($\bigstar$)), which contradicts $z_0$ mapping onto this element. Therefore, this situation, that is, the existence of a match in $\I^*$ with a counting variable mapping outside $\domain{*}$, does not occur.  
Hence, we have $\match(\tx \cup \tz) \subseteq \domain{*}$. Then since $\rho^*_{|\domain{*}} = \mathsf{id}$, we have  $\rho^* \circ \match= \match $, which shows that the mapping $\match \mapsto \rho^* \circ \match$ is injective. This means that $\I$  contains at least as many c-matches as $\I^*$, and since $\I$ is a countermodel,  $\I^*$ must also be a countermodel.

As observed earlier, the obtained countermodel $\I^*$ has a particular structure: it can be seen as the canonical model of an ABox $\Amc^*$
whose size is polynomially bounded by the size of $\domain{*}$, itself being single exponential in the size of the input. 
The modified interleaving thus allows us to improve the algorithm used in the general case. Indeed, it is now sufficient for a Turing machine to (i) guess an ABox of single-exponential size in $|\kb|$ and $|q|$,  and (ii) check that the canonical model of the guessed ABox and original TBox contains fewer c-matches for $q(\ta)$ than the integer provided as input. Importantly, due to our assumption that $q$ is \exhaustive rooted, matches cannot reach elements in the canonical model that have depth greater than $|q|$. There are thus only single-exponentially many domain elements that may appear in the image of a match, and so it is possible to enumerate and count all matches in single-exponential time w.r.t.\ $|\kb|$ and $|q|$.
\end{proof}

\section{Best Certain Answers}

The definition of certain answers implies that if $(\ta, [m, M]) \in \qIntervals{q}{\kb}$, then we also have $(\ta, [m', M']) \in \qIntervals{q}{\kb}$ for every $m' \leq m$ and $M' \geq M$. It is naturally of interest to focus on certain answers providing the best bounds, i.e., those of the form 
$ (\ta, [\min_{\I \models \kb} q_\ta^{\I}, \max_{\I \models \kb} q_\ta^{\I} ]) $. 

In this section, we show that the problem of identifying the best lower bound ($\min_{\I \models \kb} q_\ta^{\I}$) 
is $\DP$-complete in data complexity. 
It is easily seen that checking whether $m$ is such an optimal bound can be done in $\DP$, by making a call to a $\coNP$ oracle (is $(\ta, [m, +\infty]) \in \qIntervals{q}{\kb}$?) and an $\NP$ oracle (is $(\ta, [m+1, +\infty]) \notin \qIntervals{q}{\kb}$?).
The $\DP$-hardness of this problem was left as an open question by Kostylev and Reutter. 
\begin{theoremrep}\label{dpthm} The following problem is $\DP$-hard in data complexity: given a \dlc\ KB $\Kmc=(\Tmc, \Amc)$, rooted CCQ $q$,
tuple $\ta$, and number $m$, decide whether $m=\min_{\I \models \kb} q_\ta^{\I}$. 
\end{theoremrep}
\begin{proofsketch}
We give a reduction from the following problem ($\DP$-complete due to \cite{garey:planarnp}): given \emph{planar}
graphs $\graph_1$ and $\graph_2$, decide if $\graph_1 \in \tcol$ and $\graph_2 \notin \tcol$.

Let the TBox $\TmcCol$ and ABoxes $\Amc_{\graph_1}, \Amc_{\graph_2}$ be defined as in the proof of Theorem~\ref{rootedconphard}.
Rename the individuals to ensure $\ainds(\Amc_{\graph_1}) \cap \ainds(\Amc_{\graph_2})=\emptyset$, then set $\Kmc=(\TmcCol, \Amc_{\graph_1} \cup \Amc_{\graph_2})$. 
Let $q^{color}_1$ and $q^{edge}_1$ (resp. $q^{color}_2$ and $q^{edge}_2$) be defined as before, but using disjoint variables and the root individual from the $\Amc_{\graph_1}$ (resp.\ $\Amc_{\graph_2}$). 
The challenge is to make sure that we can determine the 3-colorability status of the two graphs solely by looking at the number of c-matches of the query. To be able to distinguish $\graph_1$ from $\graph_2$, 
we  introduce an asymmetry by duplicating the color counter query for $\graph_1$, i.e., create a copy $q^{color}_0$ of $q^{color}_1$ that uses fresh variables but the same root individual. 
We then take the query 
\[ q() := q^{color}_0~\wedge~q^{color}_1~\wedge~q^{edge}_1~\wedge~q^{color}_2~\wedge~q^{edge}_2.\]

%
%
%
%

We claim $(\ta_\emptyset, [36, +\infty]) \in \qIntervals{q}{\kb}$ iff  $\graph_1 \in \tcol$ and $\graph_2 \not \in \tcol$. 
This is proven by a case analysis, summarized here: 
\begin{center}
\begin{tabular}{lcc}
   & $\graph_1 \in \tcol$      & $\graph_1 \notin \tcol$		\\
\midrule
$\graph_2 \in \tcol$ & $27~(= 3 \times 3 \times 3)$  & $48~(= 4 \times 4 \times 3)$ \\
$\graph_2 \notin \tcol$ & $36~(= 3 \times 3 \times 4)$  & $64~(= 4 \times 4 \times 4)$ \\
\bottomrule 
\end{tabular}
\end{center}
\noindent Each of the four cells displays the least value of $m$ such that $(\ta_\emptyset, [m, +\infty]) \in \qIntervals{q}{\kb}$, under different assumptions on the 3-colorability of $\graph_1$ and $\graph_2$. 
To establish these values, one must first prove that every model has at least this many c-matches, 
and then exhibit a model that realizes the exact number. For the latter, we utilize our assumption that the graphs are planar, hence $4$-colorable \cite{gonthier:fourcolorcoq}, which we use to show that the minimal number of c-matches is realized in a model that encodes proper 3- or 4-colorings of the graphs. 
\end{proofsketch}

\begin{proof}
\renewcommand{\red}{\ind{r}}
\renewcommand{\green}{\ind{g}}
\renewcommand{\blue}{\ind{b}}

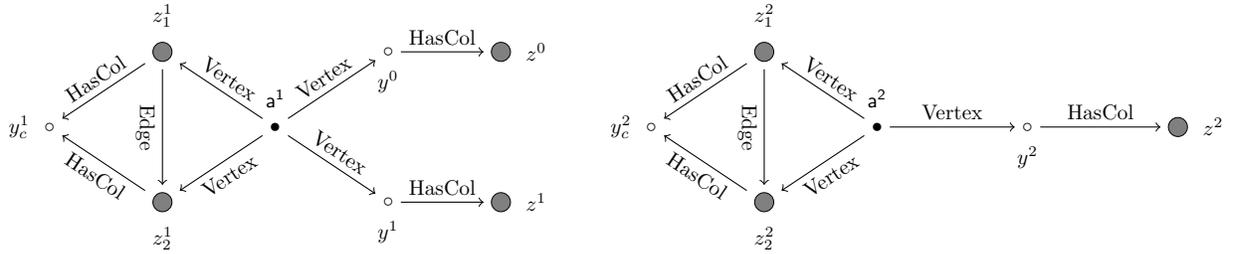
\begin{figure}[h]
\centering
\begin{tikzpicture}
[every node/.style={scale=0.8, sloped}]
\node at ( 0, 3) (r1) [label=above:$\ind{a^1}$]  {$\gx$};
\node at (-1.5, 4) (z11)	[label=above:$z_1^1$]	{$\gz$};
\node at (-1.5, 2) (z21) 	[label=below:$z_2^1$]	{$\gz$};
\node at (-3, 3) (yc1) 	[label=left:$y_c^1$]	{$\gy$};
\node at (1.5, 2) (yv1)	[label=below:$y^1$]  {$\gy$};
\node at (3, 2) (zc1) 	[label=right:$z^1$]	{$\gz$};

\node at (1.5, 4) (yv0)	[label=below:$y^0$]  {$\gy$};
\node at (3, 4) (zc0) 	[label=right:$z^0$]	{$\gz$};

\node at ( 8, 3) (r2) [label=above:$\ind{a^2}$]  {$\gx$};
\node at (6.5, 4) (z12)	[label=above:$z_1^2$]	{$\gz$};
\node at (6.5, 2) (z22) 	[label=below:$z_2^2$]	{$\gz$};
\node at (5, 3) (yc2) 	[label=left:$y_c^2$]	{$\gy$};
\node at (10, 3) (yv2)	[label=below:$y^2$]  {$\gy$};
\node at ( 12, 3) (zc2) 	[label=right:$z^2$]	{$\gz$};

\path
(r1) edge [->] node[above] {$\rolestyle{Vertex}$} 	(z11)
(r1) edge [->] node[below] {$\rolestyle{Vertex}$} 	(z21)
(r1) edge [->] node[above] {$\rolestyle{Vertex}$} 	(yv1)
(yv1) edge [->] node[above] {$\rolestyle{HasCol}$} 	(zc1)
(z11) edge [->] node[below] {$\rolestyle{Edge}$} 	(z21)
(z11) edge [->] node[above] {$\rolestyle{HasCol}$} 	(yc1)
(z21) edge [->] node[below] {$\rolestyle{HasCol}$} 	(yc1)

(r1) edge [->] node[above] {$\rolestyle{Vertex}$} 	(yv0)
(yv0) edge [->] node[above] {$\rolestyle{HasCol}$} 	(zc0)

(r2) edge [->] node[above] {$\rolestyle{Vertex}$} 	(z12)
(r2) edge [->] node[below] {$\rolestyle{Vertex}$} 	(z22)
(r2) edge [->] node[above] {$\rolestyle{Vertex}$} 	(yv2)
(yv2) edge [->] node[above] {$\rolestyle{HasCol}$} 	(zc2)
(z12) edge [->] node[below] {$\rolestyle{Edge}$} 	(z22)
(z12) edge [->] node[above] {$\rolestyle{HasCol}$} 	(yc2)
(z22) edge [->] node[below] {$\rolestyle{HasCol}$} 	(yc2); 
\end{tikzpicture}
\caption{The rooted CCQ $q$, which is the conjunction of $q_1, q_0$ (left part) and $q_2$ (right part).}
\label{dp_rooted_query}
\end{figure}

\begin{figure}[h]
\centering
\begin{tikzpicture}
[every node/.style={scale=0.8, sloped}]
\node at ( 0, 4) (root1)	[label=below:$\ind{a^1}$]	  		{$\bullet$};
\node at (-2, 5) (u1)	[label=above:$\ind{u_1^1}$]	  		{$\bullet$};
\node at (-2, 3) (v1)	[label=below:$\ind{u_2^1}$]			{$\bullet$};
\node at (-4, 5) (uc1) 	[label=above:$\ind{u_1^1}\rolestyle{HasCol}$]							{$\circ$};
\node at (-4, 3) (vc1) 	[label=below:$\ind{u_2^1}\rolestyle{HasCol}$]							{$\circ$};
\node at ( 4, 5) (r1) 	[label=right:$\ind{r^1}$]		{$\bullet$};
\node at ( 4, 4) (g1)   	[label=right:$\ind{g^1}$]	{$\bullet$};
\node at ( 4, 3) (b1)   [label=right:$\ind{b^1}$]	{$\bullet$};
\node at ( 2, 4) (a1)	[label=below:$\ind{a_v^1}$]			{$\bullet$};
	
\node at ( 0, 0) (root2)	[label=below:$\ind{a^2}$]	  		{$\bullet$};
\node at (-2, 1) (u2)	[label=above:$\ind{u_1^2}$]	  		{$\bullet$};
\node at (-2,-1) (v2)	[label=below:$\ind{u_2^2}$]			{$\bullet$};
\node at (-4, 1) (uc2) 	[label=above:$\ind{u_1^2}\rolestyle{HasCol}$]							{$\circ$};
\node at (-4,-1) (vc2) 	[label=below:$\ind{u_2^2}\rolestyle{HasCol}$]							{$\circ$};
\node at ( 4, 1) (r2) 	[label=right:$\ind{r^2}$]		{$\bullet$};
\node at ( 4, 0) (g2)   	[label=right:$\ind{g^2}$]	{$\bullet$};
\node at ( 4,-1) (b2)   [label=right:$\ind{b^2}$]	{$\bullet$};
\node at ( 2, 0) (a2)	[label=below:$\ind{a_v^2}$]			{$\bullet$};
	
\path 
(root1) edge [->] node[above] {$\rolestyle{Vertex}$} (u1)
(root1) edge [->] node[below] {$\rolestyle{Vertex}$} (v1)
(root1) edge [->] node[above] {$\rolestyle{Vertex}$} (a1)

(u1) edge [<->] node[below] 	{$\rolestyle{Edge}$} (v1)
(u1) edge [->] node[below] 	{$\rolestyle{HasCol}$} (uc1)
(v1) edge [->] node[above] 	{$\rolestyle{HasCol}$} (vc1)

(a1) edge [loop above] node [above]	{$\rolestyle{Edge}$} (a1)
(a1) edge [->] node [above, near end]	{$\rolestyle{HasCol}$} (r1)
(a1) edge [->] node [above, near end]	{$\rolestyle{HasCol}$} (g1)
(a1) edge [->] node [below, near end]	{$\rolestyle{HasCol}$} (b1)

(root2) edge [->] node[above] {$\rolestyle{Vertex}$} (u2)
(root2) edge [->] node[below] {$\rolestyle{Vertex}$} (v2)
(root2) edge [->] node[above] {$\rolestyle{Vertex}$} (a2)

(u2) edge [<->] node[below] 	{$\rolestyle{Edge}$} (v2)
(u2) edge [->] node[below] 	{$\rolestyle{HasCol}$} (uc2)
(v2) edge [->] node[above] 	{$\rolestyle{HasCol}$} (vc2)

(a2) edge [loop above] node [above]	{$\rolestyle{Edge}$} (a2)
(a2) edge [->] node [above, near end]	{$\rolestyle{HasCol}$} (r2)
(a2) edge [->] node [above, near end]	{$\rolestyle{HasCol}$} (g2)
(a2) edge [->] node [below, near end]	{$\rolestyle{HasCol}$} (b2)
;

\end{tikzpicture}
\caption{A part of $\canmodof{\kb_{(\graph_1, \graph_2)}}$ with $(u_1^1, u_2^1) \in \edges_1$ and $(u_1^2, u_2^2) \in \edges_2$.}
\label{dp_rooted_canmod}
\end{figure}

We provide more details on the case analysis mentioned in the body of the paper. In what follows, $\I$ will denote an arbitrary model of $\Kmc=(\TmcCol, \Amc_{\graph_1} \cup \Amc_{\graph_2})$. We first remark that every model contains the c-matches given by:
$$z^0, z^1 \mapsto \red^1 \suchthat \green^1 \suchthat \blue^1 \quad z_1^1, z_2^1 \mapsto \ind{a_v}^1
\qquad
z^2 \mapsto \red^2 \suchthat \green^2 \suchthat \blue^2 \quad z_1^2, z_2^2 \mapsto \ind{a_v}^2$$
Hence, $\qAnswers{q}{\I}_\emptyset \geq 3 \times 3 \times 1 \times 3 \times 1 = 27$. 

In what follows, we will use $\I^\graph_\tau$ to denote a minimal model of $\kb_\graph$ complying with a given coloring $\tau$ of a graph $\graph$, constructed as in the proof of Theorem~\ref{rootedconphard}. We observe that if $\tau_1$  and $\tau_2$ are respectively colorings for the graphs $\graph_1$ and $\graph_2$, then the interpretation $\I^{\graph_1}_{\tau_1} \cup \I^{\graph_2}_{\tau_2}$ which is the disjoint union of $\I^{\graph_1}_{\tau_1}$ and $\I^{\graph_2}_{\tau_2}$ is a model of the considered KB $\kb$.
We use such models to establish the minimum number of c-matches in the four different cases:
\begin{itemize}
\item $\graph_1, \graph_2 \in \tcol$: We have already seen that every model of $\kb$ contains at least $27$ c-matches. 
Let $\tau_1$  (resp.\ $\tau_2$) be a $3$-coloring for $\graph_1$ (resp.\ $\graph_2$). Then the model $\I^{\graph_1}_{\tau_1} \cup \I^{\graph_2}_{\tau_2}$ has exactly $27$ c-matches.
\item $\graph_1 \in \tcol, \graph_2 \notin \tcol$: As $\graph_2$ is not 3-colorable, the part of $\I$ describing $\graph_2$ must either introduce a fourth color, providing a new value for $z^2$ (hence at least $3 \times 3 \times 1 \times 4 \times 1 = 36$ c-matches), or contain a monochromatic edge, providing another possible value for $(z_1^2, z_2^2)$ (hence at least $3^2 \times 1 \times 3 \times 2 = 54$ c-matches). Therefore, every model contains at least $36$ c-matches for $q$. To show we cannot ensure more than $36$ c-matches, let 
$\tau_1$  (resp.\ $\tau_2$) be a $3$-coloring (resp.\ $4$-coloring) for $\graph_1$ (resp $\graph_2$).  Then $\I^{\graph_1}_{\tau_1} \cup \I^{\graph_2}_{\tau_2}$ has exactly $36$ c-matches. 
\item $\graph_1 \notin \tcol, \graph_2 \in \tcol$:
The part of $\I$ describing $\graph_1$ must introduce either a fourth color, providing a new value for $z^0$ and $z^1$ (hence  at least $4 \times 4 \times 1 \times 3 \times 1 = 48$ c-matches), or contain a monochromatic edge, providing another possible value for $(z_1^1, z_2^1)$ (hence at least $3 \times 3 \times 2 \times 3 \times 1 = 54$ c-matches). It follows that every model contains at least $48$ c-matches. To show this is the best value that can be attained, let $\tau_1$  (resp.\ $\tau_2$) be a $4$-coloring (resp.\ $3$-coloring) for $\graph_1$ (resp.\ $\graph_2$). Then $\I^{\graph_1}_{\tau_1} \cup \I^{\graph_2}_{\tau_2}$ has exactly $48$ c-matches.
\item $\graph_1, \graph_2 \notin \tcol$:
For each of the two graphs, $\I$ must introduce either a fourth color or a monochromatic edge. There are four cases to consider: 
\begin{center}
\begin{tabular}{lcc}
   &  Fourth color in $\graph_1$'s part     & Monochromatic edge in $\graph_1$'s part	\\
\midrule
Fourth color in $\graph_2$'s part & $4^2 \times 1 \times 4 \times 1 = 64$  & $3^2 \times 2 \times 4 \times 1 = 72$ \\
Monochromatic edge in $\graph_2$'s part & $4^2 \times 1 \times 3 \times 2 = 96$  & $3^2 \times 2 \times 3 \times 2 = 108$ \\
\bottomrule 
\end{tabular}
\end{center}
We therefore see that every model contains at least $64$ c-matches of $q$. To realize the minimal number, we 
%
let $\tau_1$  (resp.\ $\tau_2$) be a $4$-coloring (resp.\ $4$-coloring) for $\graph_1$ (resp.\ $\graph_2$) and observe that $\I^{\graph_1}_{\tau_1} \cup \I^{\graph_2}_{\tau_2}$ has exactly $64$ c-matches.
\end{itemize}
This completes the case analysis, the rest of the argument is contained in the proof sketch. 
\end{proof}


The preceding reduction can be adapted to show 
$\DP$-hardness also
for the two kinds of CCQs from \cite{kostylevreutter:count}, but without the rootedness restriction.


\begin{toappendix}

\paragraph*{$\DP$-hardness for Count queries from \cite{kostylevreutter:count}}
\begin{proof}
We recall that the Count queries from  \cite{kostylevreutter:count} are obtaining by requiring all of the non-answer variables to be counting variables. The queries from the preceding reduction do not satisfy this restriction, as they use existential variables, but we can modify the reduction in order to make it work for such queries. 

In the modified reduction, each vertex is described in the ABox with a specific concept, either $\rolestyle{Vertex_1}$ or $\rolestyle{Vertex}_2$ depending on which graph it appears in.  The TBox contains the following axioms:
$$\{ \rolestyle{Vertex}_1 \incl \exists \rolestyle{HasCol}_1, \rolestyle{Vertex}_2 \incl \exists \rolestyle{HasCol}_2, \exists \rolestyle{HasCol}_1^- \incl \rolestyle{Color}_1, \exists \rolestyle{HasCol}_2^- \incl \rolestyle{Color}_2 \}.$$
The subqueries $q^{edge}_i$  and $q^{col}_i$ are then modified as follows for $i \in \{1,2\}$:
\begin{align*}
q^{edge}_i = & \exists z_c^i ~ \exists z^i_1 ~ \exists z^i_2 ~ \rolestyle{Edge}(z^i_1, z^i_2)  \wedge  \rolestyle{HasCol_i}(z^i_1, z^i_c)  \wedge  \rolestyle{HasCol_i}(z^i_2, z^i_c)\smallskip\\
q^{col}_i = &  \exists z^i  ~  \rolestyle{Color_i}(z^i)
\end{align*}
and $q^{col}_0$ is redefined as: $\exists z^0 ~  \rolestyle{Color_1}(z^0)$. 

It is easily verified that after these modifications, the query $q$ now corresponds to a Count query as defined in \cite{kostylevreutter:count}. The query $q$ is displayed in Figure \ref{dp_count_query}, and the slightly ajusted canonical model $\canmodof{\kb_{(\graph_1, \graph_2)}}$ is displayed in Figure~\ref{dp_count_canmod}. \medskip\\
We can then redo the argument in the same manner as before, and the case analysis will give rise to precisely the same numbers of c-matches. 

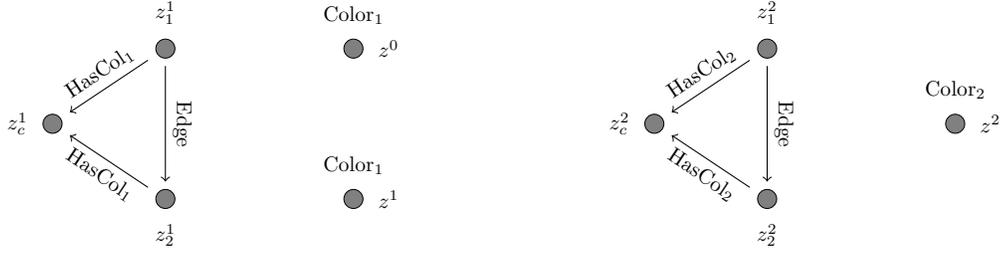
\begin{figure}[h]
\centering
\begin{tikzpicture}
[every node/.style={scale=0.8, sloped}]

\node at (-1.5, 4) (z11)	[label=above:$z_1^1$]	{$\gz$};
\node at (-1.5, 2) (z21) 	[label=below:$z_2^1$]	{$\gz$};
\node at (-3, 3) (zc1) 	[label=left:$z_c^1$]	{$\gz$};
\node at (1, 2) (z1) 	[label=right:$z^1$, label=above:$\rolestyle{Color_1}$] {$\gz$};

\node at (1, 4) (z0) 	[label=right:$z^0$, label=above:$\rolestyle{Color_1}$]	{$\gz$};

\node at (6.5, 4) (z12)	[label=above:$z_1^2$]	{$\gz$};
\node at (6.5, 2) (z22) 	[label=below:$z_2^2$]	{$\gz$};
\node at (5, 3) (zc2) 	[label=left:$z_c^2$]	{$\gz$};
\node at (9, 3) (z2) 	[label=right:$z^2$, label=above:$\rolestyle{Color_2}$]	{$\gz$};

\path
(z11) edge [->] node[above] {$\rolestyle{Edge}$} 	(z21)
(z11) edge [->] node[above] {$\rolestyle{HasCol_1}$} 	(zc1)
(z21) edge [->] node[below] {$\rolestyle{HasCol_1}$} 	(zc1)

(z12) edge [->] node[above] {$\rolestyle{Edge}$} 	(z22)
(z12) edge [->] node[above] {$\rolestyle{HasCol_2}$} 	(zc2)
(z22) edge [->] node[below] {$\rolestyle{HasCol_2}$} 	(zc2)
;
\end{tikzpicture}
\caption{The Count CQ $q$, which is the conjunction of $q^{edge}_1, q^{col}_1, q^{col}_0$ (left part) and $q^{edge}_2, q^{col}_2$ (right part).}
\label{dp_count_query}
\end{figure}

\begin{figure}[h]
\centering
\begin{tikzpicture}
[every node/.style={scale=0.8, sloped}]
\node at (-2, 5) (u1)	[label=above:$\ind{u_1^1}$, label=right:$\rolestyle{Vertex_1}$]	  		{$\bullet$};
\node at (-2, 3) (v1)	[label=below:$\ind{u_2^1}$, label=right:$\rolestyle{Vertex_1}$]			{$\bullet$};
\node at (-4, 5) (uc1) 	[label=above:$\ind{u_1^1}\rolestyle{HasCol_1}$]							{$\circ$};
\node at (-4, 3) (vc1) 	[label=below:$\ind{u_2^1}\rolestyle{HasCol_1}$]							{$\circ$};
\node at ( 4, 5) (r1) 	[label=below:$\ind{r^1}$, label=right:$\rolestyle{Color_1}$]		{$\bullet$};
\node at ( 4, 4) (g1)   	[label=below:$\ind{g^1}$, label=right:$\rolestyle{Color_1}$]	{$\bullet$};
\node at ( 4, 3) (b1)   [label=below:$\ind{b^1}$, label=right:$\rolestyle{Color_1}$]	{$\bullet$};
\node at ( 2, 4) (a1)	[label=below:$\ind{a_v^1}$]			{$\bullet$};

\node at (-2, 1) (u2)	[label=above:$\ind{u_1^2}$, label=right:$\rolestyle{Vertex_2}$]	  		{$\bullet$};
\node at (-2,-1) (v2)	[label=below:$\ind{u_2^2}$, label=right:$\rolestyle{Vertex_2}$]			{$\bullet$};
\node at (-4, 1) (uc2) 	[label=above:$\ind{u_1^2}\rolestyle{HasCol_2}$]							{$\circ$};
\node at (-4,-1) (vc2) 	[label=below:$\ind{u_2^2}\rolestyle{HasCol_2}$]							{$\circ$};
\node at ( 4, 1) (r2) 	[label=below:$\ind{r^2}$, label=right:$\rolestyle{Color_2}$]		{$\bullet$};
\node at ( 4, 0) (g2)   	[label=below:$\ind{g^2}$, label=right:$\rolestyle{Color_2}$]	{$\bullet$};
\node at ( 4,-1) (b2)   [label=below:$\ind{b^2}$, label=right:$\rolestyle{Color_2}$]	{$\bullet$};
\node at ( 2, 0) (a2)	[label=below:$\ind{a_v^2}$]			{$\bullet$};
	
\path
(u1) edge [<->] node[above] 	{$\rolestyle{Edge}$} (v1)
(u1) edge [->] node[below] 	{$\rolestyle{HasCol_1}$} (uc1)
(v1) edge [->] node[above] 	{$\rolestyle{HasCol_1}$} (vc1)

(a1) edge [loop above] node [above]	{$\rolestyle{Edge}$} (a1)
(a1) edge [->] node [above, near end]	{$\rolestyle{HasCol_1}$} (r1)

(u2) edge [<->] node[above] 	{$\rolestyle{Edge}$} (v2)
(u2) edge [->] node[below] 	{$\rolestyle{HasCol_2}$} (uc2)
(v2) edge [->] node[above] 	{$\rolestyle{HasCol_2}$} (vc2)

(a2) edge [loop above] node [above]	{$\rolestyle{Edge}$} (a2)
(a2) edge [->] node [above, near end]	{$\rolestyle{HasCol_2}$} (r2)
;

\end{tikzpicture}
\caption{A part of $\canmodof{\kb_{(\graph_1, \graph_2)}}$ with $(u_1^1, u_2^1) \in \edges_1$ and $(u_1^2, u_2^2) \in \edges_2$.}
\label{dp_count_canmod}
\end{figure}

\end{proof}

\paragraph*{$\DP$-hardness for Cntd queries from \cite{kostylevreutter:count}}
\begin{proof}
We recall that the Cntd queries from  \cite{kostylevreutter:count} correspond to CCQs with exactly one counting variable. 
As in the previous reductions, we aim to force additional matches whenever an input graph is not $3$-colorable, and the challenge is to track of the amount of colors used to color the two graphs. 

Having only a single counting variable forces us to count colors used for $\graph_1$ in exactly the same as we count those used for~$\graph_2$. In particular, the asymmetry we introduced in the query must now be introduced into the ABox. This is done by considering a copy of our first graph. However, this is not enough as two different graphs could use the same additional color, making it impossible to detect with our single counting variable that both graphs are using more than three colors. Therefore, we will  provide a set of basic colors \emph{for each graph} and additionally check whether a graph uses a color that is intended for another graph. Concretely, we achieve this by connecting vertices from different graphs using a new role $\rolestyle{Diff}$, and by adding a new subquery that will generate new c-matches whenever two vertices connected by $\rolestyle{Diff}$ use the same color.

Let us now give a more formal description of the construction.  As mentioned earlier, we will introduce a copy $\graph_0 = (\vertices_0, \edges_0)$ of the graph $\graph_1$. Without loss of generality, we can assume that $\vertices_0 \cap \vertices_1 \cap \vertices_2 = \emptyset$. As ABox individuals, we will use: 
\begin{itemize}
\item an individual name $\ind{u}$ for each $u \in \vertices_0 \cup \vertices_1 \cup \vertices_2$, to represent our graphs;
\item  individuals $\ind{r}_0, \ind{g}_0, \ind{b}_0$ (resp.\ $\ind{r}_1, \ind{g}_1, \ind{b}_1$ and $\ind{r}_2, \ind{g}_2, \ind{b}_2$), intended to color $\graph_0$ (resp.\ $\graph_1$ and $\graph_2$);
\item auxiliary individuals for vertices ($\ind{a}_0$, $\ind{a}_1$, $\ind{a}_2$, $\ind{c}$, $\ind{d}$, $\ind{e}$) and auxiliary individuals for colors ($\ind{r}$, $\ind{g}$, $\ind{b}$).
\end{itemize}
We then consider the following ABox:
\begin{align*}
\abox_{(\graph_1, \graph_2)} =
& \{ \rolestyle{Vertex}(\ind{u}) \suchthat u \in \vertices_0 \cup \vertices_1 \cup \vertices_2 \} \\
& \cup \{ \rolestyle{Edge}(\ind{u_1}, \ind{u_2}) \suchthat (u_1, u_2) \in \edges_0 \cup \edges_1 \cup \edges_2 \}  \\
& \cup \{ \rolestyle{Edge}(\ind{a}_0, \ind{a}_0), \rolestyle{Edge}(\ind{a}_1, \ind{a}_1), \rolestyle{Edge}(\ind{a}_2, \ind{a}_2), \rolestyle{Edge}(\ind{c}, \ind{c}), \rolestyle{Edge}(\ind{d}, \ind{d}) \} \\
& \cup \{ \rolestyle{Diff}(\ind{u_1}, \ind{u_2}) \suchthat u_1 \in \vertices_i, u_2 \in \vertices_j, i \neq j \} \\
& \cup \{ \rolestyle{Diff}(\ind{u}, \ind{a}_i) \suchthat u \in \vertices_j, i, j \in \{0, 1, 2\}, i \neq j  \} \\
& \cup \{ \rolestyle{Diff}(\ind{a}_0, \ind{a}_0), \rolestyle{Diff}(\ind{a}_1, \ind{a}_1), \rolestyle{Diff}(\ind{a}_2, \ind{a}_2), \rolestyle{Diff}(\ind{c}, \ind{c}), \rolestyle{Diff}(\ind{e}, \ind{e}) \} \\
& \cup \{ \rolestyle{Aux_e}(\ind{a}_0, \ind{a}_0), \rolestyle{Aux_e}(\ind{a}_1, \ind{a}_1),  \rolestyle{Aux_e}(\ind{a}_2, \ind{a}_2), \rolestyle{Aux_e}(\ind{d}, \ind{d}) \} \\
& \cup \{ \rolestyle{Aux_e}(\ind{e}, \ind{u}) \suchthat u \in \vertices_0 \cup \vertices_1 \cup \vertices_2 \} \\
& \cup \{ \rolestyle{Aux_e}(\ind{u}, \ind{c}) \suchthat u \in \vertices_0 \cup \vertices_1 \cup \vertices_2 \} \\
& \cup \{ \rolestyle{Aux_d}(\ind{a}_0, \ind{a}_0), \rolestyle{Aux_d}(\ind{a}_1, \ind{a}_1),  \rolestyle{Aux_d}(\ind{a}_2, \ind{a}_2), \rolestyle{Aux_d}(\ind{e}, \ind{e}) \} \\
& \cup \{ \rolestyle{Aux_d}(\ind{d}, \ind{u}) \suchthat u \in \vertices_0 \cup \vertices_1 \cup \vertices_2 \} \\
& \cup \{ \rolestyle{Aux_d}(\ind{u}, \ind{c}) \suchthat u \in \vertices_0 \cup \vertices_1 \cup \vertices_2 \} \\
& \cup \{ \rolestyle{HasCol}(\ind{a}_i, t) \suchthat t \in \{ \ind{r}_i, \ind{g}_i, \ind{b}_i \}, i \in \{0, 1, 2\} \} \\
& \cup \{ \rolestyle{HasCol}(\ind{c}, \ind{r}), \rolestyle{HasCol}(\ind{d}, \ind{r}), \rolestyle{HasCol}(\ind{d}, \ind{g}), \rolestyle{HasCol}(\ind{d}, \ind{b}), \rolestyle{HasCol}(\ind{e}, \ind{r}), \rolestyle{HasCol}(\ind{e}, \ind{g}), \rolestyle{HasCol}(\ind{e}, \ind{b}) \}.
\end{align*}
and the TBox
$\tbox := \{ \rolestyle{Vertex} \incl \exists \rolestyle{HasCol} \}$. We denote by $\kb_\graph = (\tbox, \abox)$ the resulting KB. A part of the canonical model of $\kb$ is depicted in Figure~\ref{dp_cntd_canmod}.


We consider the three following subqueries:
\begin{align*}
q^{diff}(y) &= \exists y_1^{d} ~ \exists y_2^{d} ~ \exists y_c^{d} ~ \rolestyle{Aux_{d}}(y, y_1^{d}) \wedge \rolestyle{Diff}(y_1^{d}, y_2^{d})  \wedge   \rolestyle{HasCol}(y_1^{d}, y_c^{d})  \wedge  \rolestyle{HasCol}(y_2^{d}, y_c^{d}) \\
q^{edge}(y) &= \exists y_1^{e} ~ \exists y_2^{e} ~ \exists y_c^{e} ~ \rolestyle{Aux_{e}}(y, y_1^{e}) \wedge \rolestyle{Edge}(y_1^{e}, y_2^{e})  \wedge   \rolestyle{HasCol}(y_1^{e}, y_c^{e})  \wedge  \rolestyle{HasCol}(y_2^{e}, y_c^{e}) \\
q^{col}(y) &= \exists z  ~ \rolestyle{HasCol}(y, z)
\end{align*}
and let $q = \exists y ~ q^{diff}(y) \wedge q^{edge} \wedge q^{col}$ be the complete CCQ, which corresponds to a Cntd query class as there is only one counting variable $z$. The query $q$ is displayed in Figure \ref{dp_cntd_query}.\medskip\\

\begin{figure}[h]
\centering
\begin{tikzpicture}
[every node/.style={scale=0.8, sloped}]

\node at ( 0, 0) (y)	[label=below:$y$]	{$\gy$};
\node at ( 0, 2) (z) 	[label=above:$z$]	{$\gz$};
\node at (-2, 0) (ye1) 	[label=above:$y^e_1$]	{$\gy$};
\node at (-4, 0) (ye2) 	[label=left:$y^e_2$] {$\gy$};
\node at (-3,-2) (yec) 	[label=below:$y^e_c$] {$\gy$};
\node at ( 2, 0) (yd1) 	[label=above:$y^d_1$]	{$\gy$};
\node at ( 4, 0) (yd2) 	[label=right:$y^d_2$] {$\gy$};
\node at ( 3,-2) (ydc) 	[label=below:$y^d_c$] {$\gy$};

\path
(y) edge [->] node[above] {$\rolestyle{HasCol}$} 	(z)

(y) edge [->] node[above] {$\rolestyle{Aux_e}$} 	(ye1)
(ye1) edge [->] node[above] {$\rolestyle{Edge}$} 	(ye2)
(ye1) edge [->] node[below] {$\rolestyle{HasCol}$} 	(yec)
(ye2) edge [->] node[below] {$\rolestyle{HasCol}$} 	(yec)

(y) edge [->] node[above] {$\rolestyle{Aux_d}$} 	(yd1)
(yd1) edge [->] node[above] {$\rolestyle{Diff}$} 	(yd2)
(yd1) edge [->] node[below] {$\rolestyle{HasCol}$} 	(ydc)
(yd2) edge [->] node[below] {$\rolestyle{HasCol}$} 	(ydc)
;
\end{tikzpicture}
\caption{The Cntd CQ $q$, which is the conjunction of $q^{edge}$ (left part), $q^{diff}$ (right part) and $q^{col}$ (upper part).}
\label{dp_cntd_query}
\end{figure}
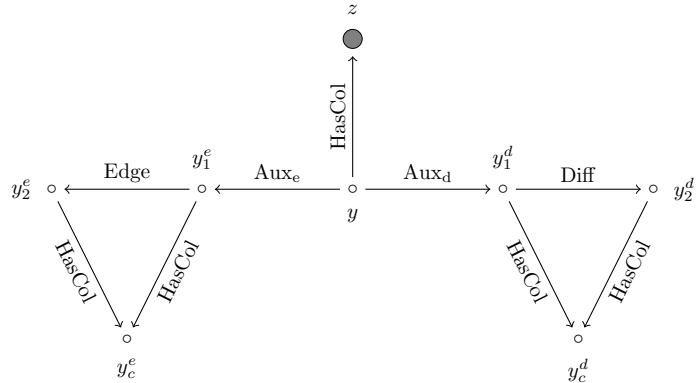
\noindent\textbf{Claim}: $(\ta_\emptyset, [10, +\infty]) \in \qIntervals{q}{\kb}$ iff  $\graph_1 \in \tcol$ and $\graph_2 \not \in \tcol$.\smallskip\\ We prove this claim using the following case analysis:
\begin{center}
\begin{tabular}{lcc}
   & $\graph_1 \in \tcol$      & $\graph_1 \notin \tcol$		\\
\midrule
$\graph_2 \in \tcol$ & $9~(= 3 + 3 + 3)$  & $11~(= 4 + 4 + 3)$ \\
$\graph_2 \notin \tcol$ & $10~(= 3 + 3 + 4)$  & $12~(= 4 + 4 + 4)$ \\
\bottomrule\\
\end{tabular}
\end{center}

To obtain the values in the preceding table, consider an arbitrary model $\I$ of $\kb$, along with a homomorphism $f : \canmod \rightarrow \I$. First observe that there are always 9 c-matches, which are obtained from the matches given by:
$$z \mapsto \ind{r}_i \suchthat \ind{g}_i \suchthat \ind{b}_i \qquad y, y^{e}_1, y^{e}_2, y^{d}_1, y^{d}_2 \mapsto \ind{a}_i \qquad y^{e}_c, y^{d}_c \mapsto \ind{r}_i
\qquad (i \in \{ 0, 1, 2 \})$$
Hence $\qAnswers{q}{\I}_\emptyset \geq 3 + 3 + 3 = 9$.

Furthermore, let us define $\tau_\I : \vertices_0 \cup \vertices_1 \cup \vertices_2 \rightarrow \domain{\I}$ as follows: $\tau_\I(u)= f(\ind{u}\rolestyle{HasCol})$. We'll use the notation $\tau_\I(\vertices_i)$ to refer to the set $\{\tau_\I(u) \mid u \in \vertices_i\}$.  
Notice that, if $\tau_\I(\vertices_i) \cap \tau_\I(\vertices_j) \neq \emptyset$ with $i \neq j$, that is, we have $u \in \graph_i, v \in \graph_j$ with $i \neq j$ and $\tau_\I(u) = \tau_\I(v)$, then we have $3$ additional c-matches corresponding to the matches given by:
$$
z, y^{e}_c \mapsto \ind{r} \suchthat \ind{g} \suchthat \ind{b}  \qquad
 y, y^{e}_1, y^{e}_2 \mapsto \ind{d} \qquad 
 y^{d}_1 \mapsto \ind{u} \qquad
 y^{d}_2 \mapsto \ind{v} \qquad
 y^{d}_c  \mapsto \tau_\I(u)
 $$
Therefore, in such a model $\I$, we have $\qAnswers{q}{\I}_\emptyset \geq 9 + 3 = 12$, and thus sufficiently many c-matches w.r.t.\ the numbers in the table. We will therefore assume in the following that 
$\tau_\I(\vertices_i) \cap \tau_\I(\vertices_j) = \emptyset$ for $i \neq j$
(assumption $(i)$).

The same applies in the case where $\tau_\I(\vertices_i) \cap \{ \ind{r}_j, \ind{g}_j, \ind{b}_j \} \neq \emptyset$ for $i \neq j$, as one can exhibit the same three additional c-matches by replacing the individual $\ind{v}$ by $\ind{a}_j$ in the latter definition of matches. Therefore, we can also assume in what follows that $\tau_\I(\vertices_i) \cap \{ \ind{r}_j, \ind{g}_j, \ind{b}_j \} = \emptyset$ for all $i \neq j$ (assumption $(ii)$).

Finally, notice that if $\tau_\I$ introduces a monochromatic edge, i.e.\ an edge $(u, v) \in \edges_0 \cup \edges_1 \cup \edges_2$ such that $\tau_\I(u) = \tau_\I(v)$, we again have $3$ additional c-matches obtained from the matches given by:
 $$
z, y^{d}_c \mapsto \ind{r} \suchthat \ind{g} \suchthat \ind{b}  \qquad
y, y^{d}_1, y^{d}_2 \mapsto \ind{e} \qquad 
 y^{e}_1 \mapsto \ind{u} \qquad
 y^{e}_2 \mapsto \ind{v} \qquad
 y^{e}_c  \mapsto \tau_\I(u)
 $$
Therefore, we can also restrict our attention to models without monochromatic edges (assumption $(iii)$).
Any model that satisfies properties $(i)$, $(ii)$ and $(iii)$ will be called \emph{non-trivial}.

We now proceed to consider the four cases. In each case, the minimal amount of c-matches is obtained by exhibiting a model built from colorings for each graph that use a minimal amount of colors. The only important difference w.r.t\ preceding reductions is that when more than one graph utilizes a fourth color, we need to use distinct fourth colors for each graph. We now complete the proof by showing that every non-trivial model has at least the number of c-matches as listed in the table. 
\begin{itemize}
\item $\graph_1, \graph_2 \in \tcol$:
We have already seen that every model contains at least $9$ c-matches. 
%
\item $\graph_1 \notin \tcol, \graph_2 \in \tcol$:
Since $\graph_0$ and $\graph_1$ are not 3-colorable, any non-trivial model $\I$ must satisfy  $\tau_\I(\vertices_0) \geq 4$ and $\tau_\I(\vertices_1) \geq 4$, due to assumption $(iii)$. In particular, we have a vertex $u_0 \in \vertices_0$ (resp.\ $u_1 \in \vertices_1$) such that $\tau_\I(u_0) \notin \{ \ind{r}_0, \ind{g}_0, \ind{b}_0 \}$ (resp.\ $\tau_\I(u_1) \notin \{ \ind{r}_1, \ind{g}_1, \ind{b}_1 \}$). This yields the following matches:
$$
z \mapsto \tau_\I(u_i) \qquad
y \mapsto \ind{u_i} \qquad
y^{e}_1, y^{e}_2, y^{d}_1, y^{d}_2 \mapsto \ind{c}
\qquad y^{e}_c, y^{d}_c \mapsto \ind{r}
\qquad (i \in \{ 0, 1 \})
$$
which give rise to two new c-matches because of assumptions $(i)$ (ensuring the two colors $\tau_\I(u_0)$ and $\tau_\I(u_1)$ are different) and $(ii)$ (ensuring $\tau_\I(u_0)$ and $\tau_\I(u_1)$ are different from the colors in the $9$ basic c-matches).
Hence, every non-trivial model contains at least $11$ c-matches.
\item $\graph_1 \in \tcol, \graph_2 \notin \tcol$:
Since $\graph_2$ is not in $\tcol$, any non-trivial model $\I$ must satisfy $\tau_\I(\vertices_2) \geq 4$ because of assumption $(iii)$. In particular, we have a vertex $u_2 \in \vertices_2$ such that $\tau_\I(u_2) \notin \{ \ind{r}_2, \ind{g}_2, \ind{b}_2 \}$. This provides a new match given by:
$$
z \mapsto \tau_\I(u_2) \qquad
y \mapsto \ind{u_2} \qquad
y^{e}_1, y^{e}_2, y^{d}_1, y^{d}_2 \mapsto \ind{c}
\qquad y^{e}_c, y^{d}_c \mapsto \ind{r}
$$
which gives rise to a new c-match because of the assumption $(ii)$ (which ensures $\tau_\I(u_2)$ is different from the colors in the   $9$ basic c-matches).
Hence, every non-trivial model contains at least $10$ c-matches.
\item $\graph_1, \graph_2 \notin \tcol$:
We can proceed similarly to the two previous cases to exhibit $u_0 \in \vertices_0,  u_1 \in \vertices_1, u_2 \in \vertices_2$ that are assigned new colors, providing three new matches given by:
$$
z \mapsto \tau_\I(u_i) \qquad
y \mapsto \ind{u_i} \qquad
y^{e}_1, y^{e}_2, y^{d}_1, y^{d}_2 \mapsto \ind{c}
\qquad y^{e}_c, y^{d}_c \mapsto \ind{r}
\qquad (i \in \{ 0, 1, 2 \})
$$
which give rise to three new c-matches because of assumptions $(i)$ (the colors $\tau_\I(u_0), \tau_\I(u_1), \tau_\I(u_2)$ are all different) and $(ii)$ (they are also different from the colors in the  $9$ basic c-matches).
Hence, we have that every non-trivial model contains at least $12$ matches.
\end{itemize}
\end{proof}

\end{toappendix}

%

%

%


\section{Conclusion \& Future Work}

We have revisited the issue of counting queries in OMQA and advanced our understanding of the complexity landscape, both by extending existing results to a more general notion of counting~CQ and by exploring when structural restrictions on the ontology and query can lead to improved complexity. 


There are several natural avenues for future study.  A first challenging problem is to provide a full classification of the data complexity of ontology-mediated queries (i.e.\ query-ontology pairs), in order to identify further tractable cases. 
%
It would also be relevant to extend the complexity study to
DLs with 
%
functional roles or quantified number restrictions, which would allow for non-trivial upper bounds on the number of matches. 
Tackling general CCQs for such DLs 
will likely require wholly different techniques from the model manipulations used in  
Section \ref{sectiongeneral}.  
However, a recent result by \citeauthor{DBLP:conf/semweb/CimaNKKGH19} (\citeyear{DBLP:conf/semweb/CimaNKKGH19}) shows 
that the canonical model property (Theorem \ref{canomini}) holds also for \dlf\ (which extends \dlc\ with functional roles), 
and hence both $\TC^0$ data complexity (Theorem \ref{tcz}) and our $\PP$-completeness result (Theorem~\ref{ppc}) for exhaustive rooted CCQs %
transfer 
to \dlf.



Much remains to be explored for queries involving other kinds of aggregate functions (min, max, sum, average), which manipulate data values. Recent studies of 
bag semantics for OMQA \cite{nikolauetal:bag,DBLP:conf/semweb/CimaNKKGH19} and 
databases with incomplete information \cite{DBLP:conf/lics/HernichK17,DBLP:conf/ijcai/ConsoleGL17}
provide important formal foundations for supporting such queries.

\section*{Acknowledgements} 
This work was partially supported by 
ANR project CQFD (ANR-18-CE23-0003).



\appendix

\newpage 
\bibliographystyle{named}
\bibliography{main}

\begin{thebibliography}{}

\bibitem[\protect\citeauthoryear{Aehlig \bgroup \em et al.\egroup
  }{2007}]{aehlig:tc0}
Klaus Aehlig, Stephen Cook, and Phuong Nguyen.
\newblock {\em Relativizing Small Complexity Classes and Their Theories}.
\newblock Springer Berlin Heidelberg, Berlin, Heidelberg, 2007.

\bibitem[\protect\citeauthoryear{Albert}{1991}]{DBLP:conf/vldb/Albert91}
Joseph Albert.
\newblock Algebraic properties of bag data types.
\newblock In {\em Proceedings of the 17th International Conference on Very
  Large Data Bases ({VLDB})}, pages 211--219, 1991.

\bibitem[\protect\citeauthoryear{Baader \bgroup \em et al.\egroup
  }{2017}]{DBLP:books/daglib/0041477}
Franz Baader, Ian Horrocks, Carsten Lutz, and Ulrike Sattler.
\newblock {\em An Introduction to Description Logic}.
\newblock Cambridge University Press, 2017.

\bibitem[\protect\citeauthoryear{Bailey \bgroup \em et al.\egroup
  }{2007}]{bailey:ppcomplete}
Delbert~D. Bailey, Víctor Dalmau, and Phokion~G. Kolaitis.
\newblock Phase transitions of {PP}-complete satisfiability problems.
\newblock {\em Discrete Applied Mathematics}, 155(12):1627 -- 1639, 2007.

\bibitem[\protect\citeauthoryear{Bienvenu and
  Ortiz}{2015}]{DBLP:conf/rweb/BienvenuO15}
Meghyn Bienvenu and Magdalena Ortiz.
\newblock Ontology-mediated query answering with data-tractable description
  logics.
\newblock In {\em Tutorial Lectures of the 11th Reasoning Web International
  Summer School}, pages 218--307, 2015.

\bibitem[\protect\citeauthoryear{Bienvenu \bgroup \em et al.\egroup
  }{2012}]{DBLP:conf/kr/BienvenuLW12}
Meghyn Bienvenu, Carsten Lutz, and Frank Wolter.
\newblock Query containment in description logics reconsidered.
\newblock In {\em Proceedings of the 13th International Conference on the
  Principles of Knowledge Representation and Reasoning ({KR})}, 2012.

\bibitem[\protect\citeauthoryear{Bienvenu \bgroup \em et al.\egroup
  }{2013}]{DBLP:conf/ijcai/BienvenuOSX13}
Meghyn Bienvenu, Magdalena Ortiz, Mantas Simkus, and Guohui Xiao.
\newblock Tractable queries for lightweight description logics.
\newblock In {\em Proceedings of the 23rd International Joint Conference on
  Artificial Intelligence ({IJCAI})}, pages 768--774, 2013.

\bibitem[\protect\citeauthoryear{Bienvenu \bgroup \em et al.\egroup
  }{2015a}]{DBLP:conf/lics/BienvenuKP15}
Meghyn Bienvenu, Stanislav Kikot, and Vladimir~V. Podolskii.
\newblock Tree-like queries in {OWL} 2 {QL:} {S}uccinctness and complexity
  results.
\newblock In {\em Proceedings of the 30th Annual {ACM/IEEE} Symposium on Logic
  in Computer Science ({LICS})}, pages 317--328, 2015.

\bibitem[\protect\citeauthoryear{Bienvenu \bgroup \em et al.\egroup
  }{2015b}]{DBLP:journals/jair/BienvenuOS15}
Meghyn Bienvenu, Magdalena Ortiz, and Mantas Simkus.
\newblock Regular path queries in lightweight description logics: Complexity
  and algorithms.
\newblock {\em Journal of Artificial Intelligence Research (JAIR)},
  53:315--374, 2015.

\bibitem[\protect\citeauthoryear{Calvanese \bgroup \em et al.\egroup
  }{2007}]{calvaneseetal:dllite}
Diego Calvanese, Giuseppe De~Giacomo, Domenico Lembo, Maurizio Lenzerini, and
  Riccardo Rosati.
\newblock Tractable reasoning and efficient query answering in description
  logics: The {DL-Lite} family.
\newblock {\em Journal of Automated Reasoning (JAR)}, 39(3):385--429, 2007.

\bibitem[\protect\citeauthoryear{Calvanese \bgroup \em et al.\egroup
  }{2008}]{DBLP:conf/cikm/CalvaneseKNT08}
Diego Calvanese, Evgeny Kharlamov, Werner Nutt, and Camilo Thorne.
\newblock Aggregate queries over ontologies.
\newblock In {\em Proceedings of the 2nd International Workshop on Ontologies
  and Information Systems for the Semantic Web ({ONISW})}, pages 97--104, 2008.

\bibitem[\protect\citeauthoryear{Cima \bgroup \em et al.\egroup
  }{2019}]{DBLP:conf/semweb/CimaNKKGH19}
Gianluca Cima, Charalampos Nikolaou, Egor~V. Kostylev, Mark Kaminski,
  Bernardo~Cuenca Grau, and Ian Horrocks.
\newblock Bag semantics of dl-lite with functionality axioms.
\newblock In {\em Proceedings of the 18th International Semantic Web Conference
  ({ISWC})}, pages 128--144, 2019.

\bibitem[\protect\citeauthoryear{Console \bgroup \em et al.\egroup
  }{2017}]{DBLP:conf/ijcai/ConsoleGL17}
Marco Console, Paolo Guagliardo, and Leonid Libkin.
\newblock On querying incomplete information in databases under bag semantics.
\newblock In Carles Sierra, editor, {\em Proceedings of the 26th International
  Joint Conference on Artificial Intelligence ({IJCAI})}, pages 993--999, 2017.

\bibitem[\protect\citeauthoryear{Garey \bgroup \em et al.\egroup
  }{1976}]{garey:planarnp}
M.R. Garey, D.S. Johnson, and L.~Stockmeyer.
\newblock Some simplified {NP}-complete graph problems.
\newblock {\em Theoretical Computer Science}, 1(3):237 -- 267, 1976.

\bibitem[\protect\citeauthoryear{Gonthier}{2008}]{gonthier:fourcolorcoq}
Georges Gonthier.
\newblock Formal proof -- {The} four-color theorem.
\newblock {\em Notices of the American Mathematical Society},
  55(11):1382--1393, 2008.

\bibitem[\protect\citeauthoryear{Grau \bgroup \em et al.\egroup
  }{2013}]{DBLP:journals/jair/GrauHKKMMW13}
Bernardo~Cuenca Grau, Ian Horrocks, Markus Kr{\"{o}}tzsch, Clemens Kupke,
  Despoina Magka, Boris Motik, and Zhe Wang.
\newblock Acyclicity notions for existential rules and their application to
  query answering in ontologies.
\newblock {\em Journal of Artificial Intelligence Research ({JAIR})},
  47:741--808, 2013.

\bibitem[\protect\citeauthoryear{Guti{\'{e}}rrez{-}Basulto \bgroup \em et
  al.\egroup }{2015}]{DBLP:journals/ws/Gutierrez-Basulto15}
V{\'{\i}}ctor Guti{\'{e}}rrez{-}Basulto, Yazmin~Ang{\'{e}}lica
  Ib{\'{a}}{\~{n}}ez{-}Garc{\'{\i}}a, Roman Kontchakov, and Egor~V. Kostylev.
\newblock Queries with negation and inequalities over lightweight ontologies.
\newblock {\em Journal of Web Semantics (JWS)}, 35:184--202, 2015.

\bibitem[\protect\citeauthoryear{Hernich and
  Kolaitis}{2017}]{DBLP:conf/lics/HernichK17}
Andr{\'{e}} Hernich and Phokion~G. Kolaitis.
\newblock Foundations of information integration under bag semantics.
\newblock In {\em Proceedings of the 32nd Annual {ACM/IEEE} Symposium on Logic
  in Computer Science ({LICS})}, pages 1--12, 2017.

\bibitem[\protect\citeauthoryear{Kostylev and
  Reutter}{2015}]{kostylevreutter:count}
Egor~V. Kostylev and Juan~L. Reutter.
\newblock Complexity of answering counting aggregate queries over {DL-Lite}.
\newblock {\em Journal of Web Semantics (JWS)}, 33(1):94--111, 2015.

\bibitem[\protect\citeauthoryear{Libkin}{2001}]{DBLP:conf/icdt/Libkin01}
Leonid Libkin.
\newblock Expressive power of {SQL}.
\newblock In {\em Proceedings of the 8th International Conference on Database
  Theory ({ICDT})}, pages 1--21, 2001.

\bibitem[\protect\citeauthoryear{Nikolaou \bgroup \em et al.\egroup
  }{2019}]{nikolauetal:bag}
Charalampos Nikolaou, Egor~V. Kostylev, George Konstantinidis, Mark Kaminski,
  Bernardo~Cuenca Grau, and Ian Horrocks.
\newblock Foundations of ontology-based data access under bag semantics.
\newblock {\em Artificial Intelligence (AIJ)}, 274:91 -- 132, 2019.

\bibitem[\protect\citeauthoryear{Poggi \bgroup \em et al.\egroup
  }{2008}]{DBLP:journals/jods/PoggiLCGLR08}
Antonella Poggi, Domenico Lembo, Diego Calvanese, Giuseppe {De Giacomo},
  Maurizio Lenzerini, and Riccardo Rosati.
\newblock Linking data to ontologies.
\newblock {\em Journal of Data Semantics}, 10:133--173, 2008.

\bibitem[\protect\citeauthoryear{Xiao \bgroup \em et al.\egroup
  }{2018}]{DBLP:conf/ijcai/XiaoCKLPRZ18}
Guohui Xiao, Diego Calvanese, Roman Kontchakov, Domenico Lembo, Antonella
  Poggi, Riccardo Rosati, and Michael Zakharyaschev.
\newblock Ontology-based data access: {A} survey.
\newblock In {\em Proceedings of the 27th International Joint Conference on
  Artificial Intelligence ({IJCAI})}, pages 5511--5519, 2018.

\end{thebibliography}

\end{document}